\documentclass[english, 11pt, letterpaper]{article}

\usepackage[utf8]{inputenc} 
\usepackage[colorlinks,breaklinks=true]{hyperref}

\usepackage{framed}

\usepackage{caption}
\usepackage{subcaption}
\usepackage[margin=1in]{geometry}

\usepackage{paralist} 
\usepackage{babel}

\title{Complexity and Parametric Computation of Equilibria in Atomic Splittable Congestion Games via Weighted Block Laplacians\footnote{Funded by the Deutsche Forschungsgemeinschaft (DFG, German Research Foundation) under Germany's Excellence Strategy – The Berlin Mathematics Research Center MATH+ (EXC-2046/1, project ID: 390685689).}}

\author{Max Klimm\footnote{School of Business and Economics, Humboldt-Universität zu Berlin, Spandauer Straße 1, 10178 Berlin, Germany. \newline \texttt{ $\{$max.klimm, philipp.warode$\}$@hu-berlin.de}} \and Philipp Warode\footnotemark[2] }

\usepackage{amsthm,amsmath,amssymb} 
\usepackage{mathtools}

\usepackage{nicefrac}

\usepackage[ruled]{algorithm2e}

\theoremstyle{definition}
\newtheorem{definition}{Definition}[section]
\theoremstyle{plain}

\newtheorem{lemma}{Lemma}
\newtheorem{theorem}{Theorem}
\newtheorem{result}{Result}
\newtheorem*{theorem*}{Theorem}
\newtheorem{corollary}{Corollary}
\theoremstyle{remark}
\newtheorem{example}{Example}[section]

\newtheorem{claim}{Claim}

\usepackage{hyperref}
\usepackage{paralist}

\usepackage{todonotes}

\usepackage{tikz}
\usetikzlibrary{calc, shapes.geometric, arrows, positioning, angles, quotes, decorations.pathreplacing}
\tikzset{
arrow/.style={-latex},
>=stealth'}
\tikzstyle{solid} = =[circle, fill,inner sep=1.5pt,outer sep=0pt]
\tikzstyle{oct} = = [draw,regular polygon,regular polygon sides=8,shape border rotate=22.5,inner sep=0pt, minimum size=25pt]

\tikzstyle{rowedge} = =[-latex,blue,dashed,thick]
\tikzstyle{coledge} = =[-latex,red,dash dot,thick]

\usepackage{dsfont}

\usepackage{shortversion}
\newappendixcollect[1]{proof}{\proof}{\endproof}{\proof[Proof of #1]}{\endproof}
\newcommand{\R}{\mathbb{R}}
\newcommand{\N}{\mathbb{N}}

\usepackage{bm}
\renewcommand{\vec}[1]{\boldsymbol{\mathbf{#1}}}

\newcommand{\range}[1]{[#1]}

\newcommand{\sgn}{\mathop{\mathrm{sgn}}}
\newcommand{\diag}{\mathop{\mathrm{diag}}}

\newcommand{\argmin}{\mathop{\mathrm{argmin}}}

\newcommand{\vspan}{\mathop{\mathrm{span}}}
\newcommand{\kernel}{\mathop{\mathrm{ker}}}

\newcommand{\rank}{\mathop{\mathrm{rank}}}

\newcommand{\PPAD}{\mathsf{PPAD}}
\newcommand{\PLS}{\mathsf{PLS}}

\newcommand{\1}{\mathds{1}}

\newcommand{\pertVar}{\varepsilon}
\newcommand{\pertVec}{\vec{\pertVar}}

\newcommand{\support}{\mathcal{S}}
\newcommand{\altSupport}{\mathcal{S}'}

\newcommand{\potentialSpace}{\mathcal{P}}
\newcommand{\potentialRegion}[1][\support]{\mathcal{P}_{#1}}
\newcommand{\pertPotentialRegion}[1][\support]{\potentialRegion[#1]^{\pertVar}}

\newcommand{\supportBoundary}[1][\support]{\partial #1}

\newcommand{\potentialOffset}{\bar{\vec{d}}}
\newcommand{\normalVec}[1][\support, e, i]{\vec{q}_{#1}}
\newcommand{\snormalVec}[1][\support, e, i]{\hat{\vec{q}}_{#1}}
\newcommand{\lightNormalVec}[1][\altSupport, e, i]{\bar{\vec{q}}_{#1}}
\newcommand{\slightNormalVec}[1][\altSupport, e, i]{\hat{\bar{\vec{q}}}_{#1}}
\newcommand{\inducedflow}{\vec{\nu}}

\newcommand{\contNeighbors}[1][\support]{\mathfrak{N}_{#1}}
\newcommand{\lowerContNeighbors}[1][\support]{\contNeighbors[#1]^{\min}}
\newcommand{\upperContNeighbors}[1][\support]{\contNeighbors[#1]^{\max}}
\newcommand{\contNeighborsMax}[1][\support]{\contNeighbors[#1]^{\max}}
\newcommand{\contNeighborsMin}[1][\support]{\contNeighbors[#1]^{\min}}

\newcommand{\pertLowerContNeighbors}[1][\support]{\contNeighbors[#1]^{\min, \pertVar}}
\newcommand{\pertUpperContNeighbors}[1][\support]{\contNeighbors[#1]^{\max, \pertVar}}

\newcommand{\lambdaMin}[1][\support]{\lambda^{\min}_{#1}}
\newcommand{\lambdaMax}[1][\support]{\lambda^{\max}_{#1}}

\newcommand{\nullspaceParameter}{\xi}

\newcommand{\nullspaceMin}[1][\support]{\nullspaceParameter^{\min}_{#1}}
\newcommand{\nullspaceMax}[1][\support]{\nullspaceParameter^{\max}_{#1}}

\newcommand{\bpMin}[1][\support]{\vec{\pi}^{\min}_{#1}}
\newcommand{\bpMax}[1][\support]{\vec{\pi}^{\max}_{#1}}

\newcommand{\nullspaceD}[1][\support]{\Delta \vec{\pi}^{\text{N}}_{#1}}

\newcommand{\predF}{\mathsf{pred}}
\newcommand{\succF}{\mathsf{succ}}
\newcommand{\startF}{\mathsf{start}}

\newcommand{\supportSpace}{\myState}
\newcommand{\pertSupportSpace}{\supportSpace^{\pertVar}}
\newcommand{\myState}{\mathfrak{S}}
\newcommand{\mystate}{\mathfrak{s}}

\newcommand{\degenerateOffsets}{\mathfrak{B}^{\mathrm{deg}}}

\definecolor{color1}{RGB}{32,70,60}
\definecolor{color2}{RGB}{170,45,0}
\definecolor{color3}{RGB}{170,150,0}
\definecolor{color4}{RGB}{0,45,170}

\definecolor{playercolor1}{RGB}{210,45,60}
\definecolor{playercolor2}{RGB}{60,170,45}
\definecolor{playercolor3}{RGB}{45,60,170}

\usetikzlibrary{shapes}

\begin{document}

\maketitle
\begin{abstract}
We settle the complexity of computing an equilibrium in atomic splittable congestion games with player-specific affine cost functions $l_{e,i}(x) = a_{e,i} x + b_{e,i}$ as we show that the computation is $\mathsf{PPAD}$-complete. To prove that the problem is contained in $\mathsf{PPAD}$, we develop a homotopy method that traces an equilibrium for varying flow demands of the players. A key technique for this method is to describe the evolution of the equilibrium locally by a novel block Laplacian matrix where each entry of the Laplacian is a Laplacian again. Using the properties of this matrix  allows to recompute efficiently the Laplacian after the support of the equilibrium changes by matrix pivot operations. These insights give rise to a path following formulation for computing an equilibrium where states correspond to supports that are feasible for some demands and neighboring supports are feasible for increased or  decreased flow demands. A closer investigation of the block Laplacian system further allows to orient the states giving rise to unique predecessor and successor states thus putting the problem into $\mathsf{PPAD}$. For the $\PPAD$-hardness, we reduce from computing an approximate equilibrium of a bimatrix win-lose game. As a byproduct of our reduction we further show that computing a multi-class Wardrop equilibrium with class dependent affine cost functions is $\PPAD$-complete as well. 


As another byproduct of our $\PPAD$-completeness proof, we obtain an algorithm that computes a continuum of equilibria parametrized by the players' flow demand. For player-specific costs, the continuum may involve several increases and decreases of the demand and yields an algorithm that runs in polynomial space. For games with player-independent costs, only demand increases are necessary yielding an algorithm computing all equilibria as a function of the flow demand that runs in time polynomial in the output. 
\end{abstract}
\thispagestyle{empty}

\newpage
\pagestyle{plain}

\setcounter{page}{1}
\section{Introduction}

Congestion games are a central topic in algorithmic game theory with applications in traffic~\cite{beckmann1956,dafermos1972,wardrop1952,rosenthal1973}, telecommunication~\cite{altman2002,orda1993,richman2007}, and logistics \cite{cominetti2009}.
We are given a graph $G = (V,E)$ with a finite set of $k$ commodities, each specified by a triplet $(s_i,t_i,r_i)$  consisting of a source node $s_i \in V$, a target node $t_i \in V$, and a fixed demand rate $r_i \in \R_{\geq 0}$. Each edge $e \in E$ is endowed with a flow-dependent (and possibly commodity-specific) cost function $l_{e,i} : \R_{\geq 0} \to \R_{\geq 0}$ that maps its flow $\smash{\bar{x}_e = \sum_{i=1}^k x_e^i}$ to a cost $l_{e,i}(\bar{x}_e)$ experienced by all flow particles of commodity~$i$ using that edge.

In Wardrop's basic model~\cite{wardrop1952}, each commodity corresponds to a continuum of users. Each user acts selfishly in minimizing their total cost, i.e., the sum of the costs of the used edges. The corresponding equilibrium concept of a \emph{Wardrop equilibrium} is defined as a multi-commodity flow with the additional property that each commodity only uses paths of minimum cost. Wardrop equilibria exist under mild assumptions on the cost functions \cite{beckmann1956,schmeidler1970}. When the cost functions of an edge are equal for all commodities, an equilibrium can be computed in polynomial time by convex programming techniques \cite{beckmann1956}. For the case of commodity-dependent affine cost functions, Meunier and Pradeau~\cite{meunier2019} very recently showed that the problem to compute a Wardrop equilibrium lies in $\PPAD$, but left it as an open problem whether the problem is actually $\PPAD$-complete.

In the light of the rise of navigation systems such as Waze and TomTom and ride sharing platforms such as Lyft and Uber, and in view of the anticipated market penetration of autonomous cars, it is sensible to assume that in the near future several competing companies will control significant portions of the road traffic.
Similarly, the ongoing discussion on Net Neutrality Rules for the Internet are fueled by the fact that few companies constitute and control large portions of the internet traffic, e.g., these days Netflix and YouTube each constitute about $15 \%$ of total downstream traffic worldwide \cite{sandvine2018}.
In these scenarios some players may be willing to sacrifice the cost experienced by some of their traffic in order to improve the overall cost of their flow, see also the discussion in Catoni and Pallottino~\cite{catoni1991}. Atomic splittable congestion are a much more compelling model
in these situations. In such a game, each commodity corresponds to a single player who controls a splittable flow of traffic in the network. The goal of the player is to minimize the total cost of their flow defined as $\smash{C(\vec x) = \sum_{e \in E} x_{e}^i\,l_{e,i}(\bar{x}_e)}$.
A multi-commodity flow $\vec x = (\vec x^1,\dots,\vec x^k)$ is a Nash equilibrium if $\smash{C^i(\vec x) \leq C^i(\tilde{\vec x}^i, \vec x^i)}$ for all players~$i \in \{1,\dots,k\}$ and all $s_i$-$t_i$-flows of rate $r_i$, where $\smash{(\tilde{\vec x}^i, \vec x^i)}$ denotes the multi-commodity flow where all players $j \neq i$ send the flow $\vec x^j$ and player~$i$ sends the flow $\tilde{\vec x}^i$.

Despite considerable progress regarding the computational complexity of equilibria in general bimatrix games \cite{chen2009,daskalakis2009}, Wardrop equilibria \cite{beckmann1956,meunier2019}, and atomic-unsplittable congestion games \cite{ackermann2008,fabrikant2004}, much less
is known regarding the computation of equilibria in atomic splittable congestion games. For affine player-independent cost functions, Cominetti et al.~\cite{cominetti2009} showed that an equilibrium can be found by computing the minimum of a convex potential function, see also Huang~\cite{huang2013} for a combinatorial algorithm for special graph topologies. Bhaskar and Lolakapuri~\cite{bhaskar2018} proposed two algorithms with exponential worst-case complexity that compute $\epsilon$-approximate Nash equilibria in singleton games with convex costs. Harks and Timmermans~\cite{harks2017} developed a polynomial time algorithm that computes an equilibrium in singleton games with player-specific affine cost functions.

Besides leaving the complexity of computing an equilibrium wide open, the approaches above also yield only a \emph{single} equilibrium for a \emph{fixed} vector of player demands.
Moreover, the algorithms of Bhaskar and Lolakapuri~\cite{bhaskar2015} and Harks and Timmermans~\cite{harks2017} work only for singleton games played on a network with two nodes. In actual traffic scenarios, the assumption that the players' demand vector is fully known and fixed is unrealistic since demands often fluctuate. In this paper, we will work towards understanding how the equilibria in atomic splittable games change as a function of the players' demand vectors. These insights will lead to a proof that the computation of equilibria is $\PPAD$-complete.

\subsection{Our results and techniques}

We settle the complexity of computing an equilibrium in an atomic splittable congestion game with player-specific affine costs showing that it is $\PPAD$-complete. The complexity class $\PPAD$ (``polynomial parity argument on directed graphs'') captures the complexity of search  problems that can be solved by directed path-following algorithms \cite{papadimitriou1994}. In a nutshell, a problem is in $\PPAD$ if there is an exponential set of states $\myState$ and polynomially computable functions $\startF$, $\predF(\cdot)$, and $\succF(\cdot)$ computing a start state, and well-defined predecessor and successor states for a given state. A state $\mystate$ is a source if $\predF(\mystate) = \mystate$ and a sink if $\succF(\mystate) = \mystate$. We are guaranteed that $\startF$ is a source and the  goal is to compute a sink or another source. 
Most notably the problem to compute an equilibrium of a bimatrix game is $\PPAD$-complete, even in special cases \cite{chen2009,daskalakis2009,mehta2014}. 
In this paper, we show the $\PPAD$-completeness of the following problem: 

%
\newcommand{\atomicsplittable}{\textsc{Nash-Atomic-Splittable}}
\newcommand{\parametricatomicsplittable}{\textsc{Parametric-Nash-Atomic-Splittable}}
\begin{framed}
\noindent\atomicsplittable\\[4pt]
\begin{tabular}{@{} ll}
\noindent \textsc{Input:} & game $(G,K,l)$ graph $G \!=\! (V,E)$, 
 commodities $\smash{K = \bigl((s_i,t_i,r_i)\bigr)_{i \in [k]}}$,\\
& and cost functions $\smash{l_{e,i}(x) = a_{e,i}\,x + b_{e,i}}$  
 for some $a_{e,i} \in \R_{>0}$, $b_{e,i} \in \R_{\geq 0}$.\\
\noindent \textsc{Output:} & Nash equilibrium $\vec x$ of $(G,K,l)$.
\end{tabular}
\end{framed}
\begin{result}[cf.~Theorems~\ref{thm:PPADmembership} and \ref{thm:hardness}]
\atomicsplittable\ is $\PPAD$-complete.	
\end{result}

To show $\PPAD$-hardness, we reduce from the problem to compute an $n^{-\beta}$-approximate Nash equilibrium for a $n \times n$ win-lose bimatrix game that is $\PPAD$-complete for all $\beta > 0$ \cite{chen2007}. Our reduction requires only two players and produces an atomic splittable congestion game on a planar graph implying hardness even for this special case. This is an interesting contrast to the $\PLS$-completeness results for computing Nash equilibria in \emph{unsplittable} congestion games  using a non-constant number of players and highly non-planar graphs \cite{ackermann2008,fabrikant2004}. As a byproduct, we also obtain $\PPAD$-completeness for a related problem settling an open question from Meunier and Pradeau~\cite{meunier2019}.

\begin{result}[cf.~Theorem~\ref{thm:wardrop_hardness}]
Computing a multi-class Wardrop equilibrium is $\PPAD$-complete.	
\end{result}

The more challenging part of the proof is to show that \atomicsplittable\ is contained in $\PPAD$. To this end, we develop a path following algorithm that pivots over player supports in a similar fashion to the Lemke-Howson-algorithm for bimatrix games \cite{lemke1964}.
Our algorithm follows a continuous path of Nash equilibria $\vec x(\lambda)$ for demand rates $\lambda r_1, \dots, \lambda r_k$ with $\lambda \in [0,1]$.
During the course of the algorithm, $\lambda$ is changed in a continuous but non-monotonic matter.
\shortversion{The general idea is that for given player support, the equilibrium can be obtained by solving a linear system of the form $\vec y(\lambda) = \vec L \vec \pi - \vec d$ where $\vec y(\lambda) \in \R^{nk}$ is the players' excess vector, $\vec \pi \in \R^{nk}$ is a player-specific potential vector, and $\vec d \in \R^{nk}$ is an offset. We call the block matrix $\vec L \in \R^{mk \times mk}$ the \emph{block Laplacian} and show that it is non-singular (for all supports that need to be considered) allowing to solve for $\vec \pi$ and, thus, to describe the evolution of the equilibrium locally. Suitable pivoting operations on the matrix further allow to compute predecessor and successor supports eventually enabling to prove that the problem is in $\PPAD$.}{%
To describe the evolution of $\lambda$ consider an arbitrary value $\lambda^* \in [0,1)$ and a corresponding equilibrium $\vec x(\lambda^*)$ and let us fix the supports of $(\vec x^1,\dots,\vec x^k)$.
By the Karush-Kuhn-Tucker conditions for each player, we derive that in $\vec x(\lambda^*)$ every player only uses paths that minimize the marginal total cost. This implies that for every player~$i$, there is a vector of vertex potentials $\vec \pi^i$ such that that player~$i$ uses  edge~$e =  (u,v)$ if and only if the difference in vertex potentials $\pi^i_v - \pi^i_u$ is at least $b_{e,i}$. Going further, we can reformulate the Nash equilibrium conditions as a system of linear equations of the form $\vec y = \vec L \vec \pi - \vec d$ where $\vec y = ((\vec y^1)^\top, \dots, (\vec y^k)^{\top})^\top$ is a block excess vector containing the excess $y_v^i$ for each player~$i$ and each vertex $v$, $\vec \pi = ((\vec \pi^1)^\top, \dots, (\vec \pi^k)^{\top})^\top$ is the block potential vector, and $\vec d \in \R^{nk}$ is an appropriate offset. The matrix $\vec L \in \R^{nk \times nk}$ is a block matrix \shortversion{$\vec L = (\vec L^{ij})_{i,j \in [k]},$}{of the form
\begin{align*}
\vec{L} :=
\begin{pmatrix}
\phantom{-}\vec{L}^{11}		& -\vec{L}^{12} 	& \dotsm & -\vec{L}^{1k} \\
-\vec{L}^{21}	& \phantom{-}\vec{L}^{22}		& \dotsm & -\vec{L}^{2k} \\
\vdots 			& \vdots 			& \ddots& \vdots 		 \\
-\vec{L}^{k1}	& -\vec{L}^{k2} 	& \dotsm & \phantom{-}\vec{L}^{kk}
\end{pmatrix},	
\end{align*}}
where the diagonal matrices $\vec L^{ii}$ are weighted Laplacian matrices for the graph containing only the edges in the support of player~$i$. The off-diagonal matrices $-\vec L^{ij}$ and $-\vec L^{ji}$ with $i \neq j$ are the negative of weighted Laplacian matrices for the graph containing only the edges in the support of both player~$i$ and $j$. The weights of the matrices $\vec L^{ij}$ depend on the coefficients $a_{e,j}$ while the weights of the matrices $\vec L^{ji}$ depend on the coefficients $a_{e,i}$ so that $\vec L$ is non-symmetric. We call $\vec L$ the \emph{block Laplacian matrix} of the graph.
We show that $\rank(\vec L) = k(n-1)$, except for a degenerated case that will be discussed later. This implies that there is a bijection between excess vectors and potentials (after fixing without loss of generality the potentials $\pi_{s_i}^i = 0$ for all players~$i$). In particular, we obtain that for a given support the set of equilibria $\{\vec x(\lambda) : \lambda \in [0,1]\}$ has dimension at most one. When the dimension is exactly one, the one-dimensional linear space $\{\vec x(\lambda) : \lambda \in [0,1]\}$ hits the boundaries of the support at exactly two points, uniquely defining two neighboring support sets. These support sets will form $\succF$ and $\predF$ of the current support set. To obtain a unique orientation, we show that it suffices to consider the determinant of the block Laplacian after erasing the row and columns corresponding to $s_i$ for each player~$i$.
}

As a byproduct, we obtain an algorithm to solve the following problem of computing \emph{all} Nash equilibria of an atomic splittable congestion game as a function of the players' demand rates:\\[-18pt]
\begin{framed}
\noindent\parametricatomicsplittable\\[4pt]
\begin{tabular}{@{} ll}
\noindent \textsc{Input:} & game $(G,K,l)$ with graph $G \!=\! (V,E)$, 
commodities $\smash{K = \bigl((s_i,t_i,r_i)\bigr)_{i \in [k]}}$,\\
& and cost functions $l_{e,i}(x) = a_{e,i}\,x + b_{e,i}$ 
for some $a_{e,i} \in \R_{>0}$, $b_{e,i} \in \R_{\geq 0}$.\\
\noindent \textsc{Output:} & piece-wise affine function $\vec f \!:\! [0,1] \to \R^{mk}$ \\
& such that $\vec f(\lambda)$ is a Nash equilibrium for 
demand rates $\lambda r_1, \dots, \lambda r_k \;\forall \lambda \in [0,1]$.
\end{tabular}
\end{framed}

In order to solve \parametricatomicsplittable, we simply run the trivial $\PPAD$-algorithm that starts in the Nash equilibrium given by $\startF$. Then, it consecutively applies $\succF$ on the current state.
During the course of the algorithm, it may be necessary to decrease $\lambda$, which leads to operations that are not reflected in the output of the algorithm. In these cases, we recall the maximal lambda $\bar{\lambda}$ such that $\vec f$ has been outputted for $[0,\bar{\lambda}]$ and only proceed to output Nash equilibria once $\bar{\lambda}$ is reached again.
Not counting the space of the output (which may be exponential in the input), this algorithm can clearly be implemented in polynomial space.

\begin{result}[cf.~Theorem~\ref{thm:parametric:PSPACE}]
\parametricatomicsplittable\ can be solved in polynomial space.
\end{result}

For the special case of player-independent cost functions $l_{e,i} = l_{e,j}$ for all $e \in E$ and $i,j \in [k]$, we can show that in our algorithm no decrease of $\lambda$ is necessary. We then obtain the following.

\begin{result}[cf.~Theorem~\ref{thm:parametric:sym:complexity}]
\parametricatomicsplittable\ can be solved in output-polynomial time for non-$b$-degenerate atomic splittable congestion games with player-independent cost functions. In particular, the runtime is in  $\mathcal{O}((kn)^{2.4} + \tau (kn)^2 )$, where $\tau$ is the number of breakpoints of the piecewise affine function $\vec f$ returned by the algorithm. 
\end{result}

Our analysis further allows to obtain the following results regarding the multiplicity of equilibria.

\begin{result}[cf.~Corollaries~\ref{cor:odd}~and~\ref{cor:nonbdegenerate:as}]
Except for a nullset of offset vectors $\vec{b} = (b_{e,i})_{e \in E, i \in [k]} \in \R^{mk}_{\geq 0}$ and demands $\vec{r}$, the games have an odd number of Nash equilibria.
\end{result}

We also exhibit an example of a game with an infinite number of equilibria (cf.~Example~\ref{ex:8playerInfinite}).

\subsection{Related work}

Atomic splittable congestion games can be seen as a coalitional version of Wardrop equilibria \cite{wardrop1952} where a finite set of player each controls a non-negligible amount of flow \cite{haurie1985,marcotte1987}. The existence of pure Nash equilibria follows from standard fixed point arguments \cite{kakutani1941,rosen1965}.
Games with player-specific cost functions were studied by Orda et al.~\cite{orda1993} who showed that Nash equilibria are unique in networks of parallel edges. Richman and Shimkin~\cite{richman2007} characterized the set of two-terminal network topologies that are necessary and sufficient for uniqueness showing that equilibria are unique if and only if the networks are nearly parallel. The latter class of networks has been introduced by Milchtaich~\cite{milchtaich2005} to characterize the uniqueness of multi-class Wardrop equilibria. Harks and Timmermans~\cite{harks2018} characterized the uniqueness of equilibria in atomic splittable congestion games in terms of the combinatorial structure of the strategy set showing that equilibria are unique when the strategy space of each player is a bidirectional flow matroid.
Bhaskar et al.~\cite{bhaskar2015} showed that edge flows in an atomic splittable game need not be unique even when the cost functions are player-independent.
For games with cost functions that are monomials of degree at most three, edge flows are known to be unique \cite{altman2002}.
The price of anarchy of atomic splittable congestion games has been studied by Cominetti et al.~\cite{cominetti2009}, Harks~\cite{harks2011}, and Roughgarden and Schoppmann~\cite{roughgarden2015}.
Catoni and Pallottino~\cite{catoni1991} provided a paradox of a non-atomic game where replacing the non-atomic players of one commodity by an atomic player with the same demand decreases the overall performance of that commodity. Hayrapetyan et al.~\cite{hayrapetyan2006} and Bhaskar et al.~\cite{bhaskar2010} studied this effect in more detail.

Cominetti et al.~\cite{cominetti2009} showed that for games with player-independent affine cost functions an equilibrium can be computed efficiently by solving a quadratic program. For special network topologies (including series-parallel graphs), Huang~\cite{huang2013} gave a combinatorial algorithm. For the case of parallel links, Harks and Timmermans~\cite{harks2017} gave a polynomial algorithm that computes the equilibrium of an atomic splittable congestion game with player-specific affine costs. Bhaskar and Lolakapuri~\cite{bhaskar2018} provided an exponential algorithm that computes approximate equilibria in games with player-independent convex costs. They also showed that some decision problems involving equilibria are $\mathsf{NP}$-complete.

Further related are \emph{unsplittable} congestion games~\cite{rosenthal1973} where each commodity chooses a single path of the network. For results on the existence of equilibria, see \cite{ackermann2008,dunkel2008,gairing2013,goemans2005,harks2012,milchtaich1996,rosenthal1973}. Computing a pure Nash equilibrium in a congestion game with unweighted players and player-independent cost functions reduces to finding the local minimum of a potential function and is, thus, contained in the complexity class $\PLS$, the class of all local search problems with polynomially searchable neighborhoods as defined by~\cite{johnson1988}. Fabrikant et al.~\cite{fabrikant2004} and Ackermann et al.~\cite{ackermann2008} showed that computing a pure Nash equilibrium is in fact $\PLS$-complete.

The problem of computing a Nash equilibrium in an atomic splittable congestion game with player-specific affine costs may also be formulated as a linear complementarity problem (LCP). Harks and Timmermans~\cite{harks2017} showed this explicitly for the singleton case, but it is not hard to convince ourselves that such a formulation is also possible in the general case by using the vertex potentials. However, it is not clear whether the resulting LCP belongs to any of the classes for which it is known that Lemke's algorithm terminates \cite{adler2011,cottle1992,eaves1971}. In addition, Lemke's algorithm introduces an \emph{artificial} variable that is traced along the course of the algorithm. Our algorithm, in contrast, traces the equilibria along a meaningful increase of demand rates. In previous work \cite{klimm2019}, we developed an algorithm that computes all Wardrop equilibria parametrized by the flow demand. While relying on a similar idea, Wardrop equilibria are much more easy to handle since they can be computed in polynomial time and are essentially unique.   
For further homotopy methods for computing equilibria, see \cite{goldberg2002,herings2010,herings2002,katzenelson1965}.

\section{Preliminaries}

An \emph{atomic splittable congestion game} is a tuple $(G, K, l)$ where $G = (V, E)$ is a directed, weakly connected graph with $n$ vertices $V$ and $m$ edges $E$, the family $K = \big( (s_1, t_1, r_1)$, \dots, $(s_k, t_k, r_k) \big)$ contains $k$ triples each of which consisting of a source node $s_i \in V$, a sink node $t_i \in V$, and a demand rate $r_i \in \R_{\geq 0}$ for each of the $k$ players, and $l$ is a family of player-specific, strictly increasing affine linear cost functions $l = (l_{e,i})_{e \in E, i \in \range{k}}$ with $l_{e,i}(x) = a_{e,i} x + b_{e,i}$ for some $a_{e,i} \in \R_{>0}$ and $b_{e,i} \in \R_{\geq 0}$.

A feasible strategy for every player $i \in \range{k}$ is to route their demand $r_i$ between their terminal vertices $s_i$ and $t_i$. Thus, a strategy for player $i$ is a \emph{$s_i$-$t_i$-flow} of rate $r_i$, i.e., a non-negative vector $\vec{x}^i = (x^i_e)_{e \in E}$ in $\R^m_{\geq 0}$ satisfying the flow conservation constraints
\shortversion{
$\smash{y^i_v := \sum_{e \in \delta^+ (v)} x^i_e - \sum_{e \in \delta^- (v)} x^i_e = 0}$ for all $v \in V \setminus \{s_i,t_i\}$, $y^i_{s_i} = r_i$, and $y^i_{t_i} = -r_i$. 
}{
\begin{equation*} 
\sum_{e \in \delta^+ (v)} x^i_e - \sum_{e \in \delta^- (v)} x^i_e =
\begin{cases}
\phantom{-}r_i & \text{if } v = s_i, \\
- r_i & \text{if } v = t_i, \\
\phantom{-}0 & \text{otherwise}
\end{cases}
\end{equation*}
for every vertex $v \in V$.}
A strategy profile is a vector 
$\smash{\vec{x} = 
( (\vec{x}^1)^{\top}, \dotsc, (\vec{x}^k)^{\top})^{\top}}
\in \smash{\R^{k m}_{\geq 0}}$
containing the flow vectors of all players.
We use the notation $(\tilde{\vec{x}}^i, \vec{x}^{-i})$ for the strategy profile where player $i$ uses the flow $\tilde{\vec{x}}^i$ and all other players use their flow as in the strategy profile $\vec{x}$.
The cost $l_{e,i}$ experienced by the flow of the player $i$ on some edge $e$ depends on the \emph{total flow} $\bar{\vec{x}} = (\bar{x}_e)_{e \in E}$ where $\smash{\bar{x}_e = \sum_{i = 1}^k x^i_e}$. 
Every player wants to minimize the \emph{total cost} $C_i (\vec{x})$ experienced by the flow sent by this player, i.e.,
$C_i (\vec{x}) = \sum_{e \in E} x^i_e l_{e,i}(\bar{x}_e)$.
We say that $\vec{x}$ is a \emph{Nash equilibrium} if for every player $i \in \range{k}$ there is no profitable deviation from $\vec x$, i.e., 
$C_i (\vec{x}) \leq C_i(\tilde{\vec{x}}^i, \vec{x}^{-i})$
for all $s_i$-$t_i$-flows $\tilde{\vec{x}}^i$ of rate $r_i$.
The \emph{marginal total cost} of player $i$ on edge $e$ given the flow $\vec{x}$ is given by
$\smash{\mu^i_e(\vec{x}) := \frac{\partial}{\partial x^i_e} \, x^i_e l_{e,i}(\bar{x}_e) =
a_{e,i} \bar{x}_e + b_{e,i} + a_{e,i} x_e^i}$. We obtain the following characterization of Nash equilibria, see, e.g., Bhaskar et al.~\cite{bhaskar2015} for a reference.

\begin{lemma} \label{lem:equilibrium:marginalcost}
The strategy profile $\vec{x}$ is a Nash equilibrium flow if and only if, for every player $i$, $\vec{x}^i$ is a $s_i$-$t_i$-flow and 
$
\sum_{e \in P} \smash{\mu^i_e(\vec{x}) \leq \sum_{e \in Q} \mu^i_e (\vec{x})}
$
for all $s_i$-$t_i$-paths $P, Q$ with $x^i_e > 0$ for all $e \in P$.
\end{lemma}
Lemma~\ref{lem:equilibrium:marginalcost} states that $\vec{x}$ is a Nash equilibrium if and only if all path used by player~$i$ are also shortest path for that player with respect to the marginal costs. This enables us to give another characterization based on (shortest path) potentials.
\begin{lemma} \label{lem:equilibrium:potentials}
The flow $\vec{x}$ is a Nash equilibrium if and only if for all $i \in [k]$ there is a potential vector $\vec{\pi}^i = (\pi^i_v)_{v \in V}$ with
$\pi^i_w \!-\! \pi^i_v = \mu^i_e (\vec{x})$ if $x_e^i \!>\! 0$, and $\pi^i_w \!-\! \pi^i_v \leq \mu^i_e (\vec{x})$ if $x_e^i \!=\! 0$
for all $e =(v,w) \in E$.
\end{lemma}
We denote the block flow vector with $\smash{\vec{x} = ((\vec x^1)^{\top},\dots,(\vec{x}^k)^{\top})^{\top}}$ and the block potential vector with $\vec{\pi} = ((\vec \pi^1)^{\top},\dots,(\vec{\pi}^k)^{\top})^{\top}$. 
For every edge $e \in E$, let $S_e \subseteq \range{k}$ be some subset of the players. Then we call the family $\support := (S_e)_{e \in E}$ of these sets a \emph{support}. We say an edge $e$ is \emph{active for player~$i$} if $i \in S_e$, and \emph{$e$ is inactive for player $i$} otherwise. We say a Nash equilibrium $\vec{x}$ \emph{has the support $\support$} if for every edge $e = (v,w) \in E$, $i \in S_e$ implies that $\pi^i_w - \pi^i_v = \mu^i_e (\vec{x})$. Thus, $\support$ is a support of the Nash equilibrium $\vec{x}$ if for every active edge the inequality $\pi^i_w - \pi^i_v \leq \mu^i_e (\vec{x})$ is satisfied with equality. In general, there are multiple supports for the same equilibrium, since for every player~$i$, any edge with $x_e^i = 0$ and $\pi^i_w - \pi^i_v = \mu^i_e (\vec{x})$ can be considered to be an active or an inactive edge.
Given a support of a Nash equilibrium, the equilibrium flow can be obtained by solving the system of equations $\pi^i_w - \pi^i_v = \mu^i_e (\vec{x}), i \in S_e, e = (v,w) \in E$ with the unknowns $\vec{\pi}$ and $\vec{x}$. For affine cost functions this reduces to a system of linear equations that can be solved explicitly. This will be the key element of our analysis and we discuss this in further detail in the next section.

\shortversiononly{
For ease of exposition, we defer all proofs to the appendix at the end of the paper. The reader may also consult the full version of this paper \cite{klimm2019preprint}.
}
 
\section{Equilibrium structure}

In this section we develop a characterization of equilibria in atomic splittable congestion games for different demand rates parametrized by a factor $\lambda$.
To this end, we use the conditions of Lemma~\ref{lem:equilibrium:potentials} to characterize all flows $\vec{x}$ and potentials $\vec{\pi}$ that describe Nash equilibria.
Given a fixed support $\support$, we call a flow $\vec{x}$ \emph{induced by a potential $\vec{\pi}$ wrt. the support $\support$} if for every edge $e \in E$ we have $\pi^i_w - \pi^i_v = \mu^i_e (\vec{x})$ for all $i \in S_e$ and $x_e^i = 0$ for all $i \notin S_e$.
As we will prove in Lemma~\ref{lem:incudedflow}, for a fixed support $\support$, the flow $\vec x$ induced by a potential $\vec \pi$ is unique.
That is, for every support $\support$, there is a well-defined function $\inducedflow_\support : \R^{nk} \to \R^{mk}$ such that $\vec x = \inducedflow_\support(\vec \pi)$ is a flow satisfying
\begin{equation}
\label{eq:inducedflow:system}
\begin{split}
\mu^i_e (\vec{x}) &= \pi^{i}_w - \pi^i_v 
	\phantom{0} \qquad
	\text{ if } i \in S_e, \\
x_e^i &= 0  
	\phantom{\pi^{i}_w - \pi^i_v} \qquad
	\text{ if } i \notin S_e
\end{split}
\end{equation}
for all edges $e = (v,w)$. Note that $\vec x$ need not be a flow satisfying any given demands and may even be negative. In the following, we will work towards finding those potential vectors $\vec \pi$ for which the induced flow $\inducedflow_{\support}(\vec \pi)$ are feasible for some demand vector $\lambda \vec r$ and some support $\support$.  
To this end, we call a potential vector $\vec \pi \in \R^{nk}$ a \emph{$\lambda$-potential} for the support $\support$ if $\inducedflow_{\support}(\vec \pi)$ is feasible for the demand $\lambda \vec r$ and satisfies the conditions of Lemma~\ref{lem:equilibrium:potentials}.  



The remainder of the section is structured as follows. In Section~\ref{sec:laplacian} we show that the induced flow functions $\inducedflow_{\support}$ are well-defined and express necessary and sufficient conditions on $\lambda$-potentials in terms of a matrix refered to as \emph{block Laplacian matrix}.
In Section~\ref{sec:potentialspace}, we use these characterizations in order to show that for a fixed support $\support$, the space of $\lambda$-potentials for some $\lambda \in [0,1]$ is a zero or one dimensional subset of $\R^{nk}$ under some non-degeneracy assumption. Finally, in Section~\ref{sec:neighbors}, we show that for a given support that admits a Nash equilibrium there are two well-defined \emph{neighboring supports} that admitting Nash equilibria for slightly higher or slightly lower demand rates.
The latter will be used in Section~\ref{sec:membership} in order to define the functions $\predF$ and $\succF$ that put the equilibrium computation problem into $\mathsf{PPAD}$. This also yields an algorithm for parametric equilibrium computation as a byproduct that we will study in more detail in Section~\ref{sec:parametric}.


\subsection{Weighted block Laplacians}

\label{sec:laplacian}

In this subsection, we assume that we are given a fixed support $\support$. For ease of exposition, we omit the dependence on $\support$ for $\inducedflow = \inducedflow_{\support}$ and related notations in this chapter.
In order to find a closed form representation of the solution $\vec{x}$ to the system \eqref{eq:inducedflow:system}, we need to introduce some notation. For any two players $i,j \in \range{k}$ and any edge $e \in E$ we define the coefficient $\omega^{ij}_e$ as follows: $\omega^{ij}_e := 1$ if $\{i,j\} \subseteq S_e$ and $\omega^{ij}_e := 0$ otherwise, i.e., the coefficient $\omega^{ij}_e$ indicates whether the edge~$e$ is active for both player~$i$ and player~$j$. We also write $\omega^i_e$ as a short hand for $\omega^{ii}_e$. For further reference, we define the diagonal matrices $\vec{\Omega}^i := \smash{\diag (\omega^{i}_{e_1}, \dotsc, \omega^{i}_{e_m})}$ and the block-diagonal matrix $\vec{\Omega}$ that contains all matrices $\vec{\Omega}^i$ as block-diagonal elements. For any edge $e \in E$, let $\kappa_e := |S_e|$ be the number of players using edge~$e$. We denote by $\bar{\vec{K}} := \smash{\diag\big( \frac{1}{\kappa_{e_1}+1}, \dotsc, \frac{1}{\kappa_{e_m} + 1} \big) \in \R^{m \times m}}$ the diagonal matrix containing all $\kappa$ values and by $\vec{K} \in \R^{mk \times mk}$ the block matrix that contains $k \times k$-times the block $\bar{\vec{K}}$.
We further define diagonal matrix $\tilde{\vec{C}}^{ij} \in \R^{m \times m}$ containing weights for every edge by
\begin{align*}
\tilde{\vec{C}}^{ij} :=
\begin{cases}
\diag\Bigl(\tfrac{\kappa_{e_1}+1-\omega_{e_1}^i}{(\kappa_{e_1}+1)a_{e_1,i}}\,\omega_{e_1}^{i},\dots,\tfrac{\kappa_{e_m}+1-\omega_{e_m}^i}{(\kappa_{e_m}+1)a_{e_m,i}}\,\omega_{e_m}^{i}\Bigr)
	&\text{if } i = j, \\
\diag\Big(\tfrac{1}{(\kappa_{e_1}+1)a_{e_1,j}}\,\omega_{e_1}^{j}, \dots, \tfrac{1}{(\kappa_{e_m}+1)a_{e_m,j}}\,\omega_{e_m}^{j}\Bigr) 
	&\text{otherwise.}
\end{cases}
\end{align*}
We define the $mk \times mk$ block-matrix $\tilde{\vec{C}}$ with the blocks $\tilde{\vec{C}}^{ii}$ on the block-diagonal and the blocks $- \vec{C}^{ij}$ as off-diagonal blocks, i.e.,
\[
\tilde{\vec{C}} =
\begin{pmatrix}
\tilde{\vec{C}}^{11} & \dotsm & - \tilde{\vec{C}}^{1k} \\
\vdots & \ddots & \vdots \\
-\tilde{\vec{C}}^{k1} & \dotsm & \tilde{\vec{C}}^{kk}
\end{pmatrix}
.
\]
Finally, we define the matrix $\vec{C} := \vec{\Omega} \tilde{\vec{C}}$.
In particular, $\vec{C}$ has the same block structure as $\tilde{\vec{C}}$, the diagonal blocks have the form $\smash{\vec{C}^{ii} = \smash{\tilde{\vec{C}}}^{ii}}$ and the off diagonal blocks have the form
\begin{align*}
\vec{C}^{ij} &:= \diag\Big(\tfrac{1}{(\kappa_{e_1}+1)a_{e_1,j}}\omega_{e_1}^{ij}, \dots, \tfrac{1}{(\kappa_{e_m}+1)a_{e_m,j}}\omega_{e_m}^{ij}\Bigr).
\end{align*}
The matrices $\vec{\Omega}^i$ encode the activity status of all edges for every player $i$, hence, $\vec{\Omega}$ encodes the support. The matrices $\bar{\vec{K}}$ and $\vec{K}$ contain normalization factors depending on the number of players. At last, the matrices $\tilde{\vec{C}}$ and $\vec{C}$ encode the linear relationship between potentials and flows in the conditions of Lemma~\ref{lem:equilibrium:potentials}---we will prove this in Lemma~\ref{lem:incudedflow} and Theorem~\ref{thm:lambdapotential} below. Overall, the above definitions imply
\begin{align*}
\vec{\tilde{C}} = (\vec{I}_{km} - \vec{K} \vec{\Omega}) \vec{A}^{-1}
\text{ and } 
\vec{C} = \vec{\Omega} (\vec{I}_{km} - \vec{K} \vec{\Omega}) \vec{A}^{-1}
,
\end{align*}
where $\vec{I}_{km}$ is the $km \times km$ identity matrix and $\vec{A}^{-1} = \diag\big( \frac{1}{a_{e_1, 1}}, \dotsc, \frac{1}{a_{e_m, k}} \big)$ is the diagonal matrix containing all reciprocals of the slopes of all cost functions.

We denote by $\vec{\Gamma}$ the vertex-edge-incidence matrix of the graph $G$ and let $\vec{G} \in \R^{nk \times mk}$ be the block-diagonal matrix containing $k$ copies of $\vec{\Gamma}$ on the diagonal. Then, in particular, the vector $\vec{G}^{\top} \! \vec{\vec{\pi}}$ contains all potential differences $\pi^i_w - \pi^i_v$ for all edges $e = (v,w)$ and players~$i$, and the vector $\vec{G} \vec{x}$ contains the excess flow of every player at every vertex. Using all these definitions, we obtain the following representation of the induced flow function $\inducedflow$.
\begin{lemma} \label{lem:incudedflow}
Given a fixed support $\support$ and a potential vector $\vec{\pi} \in \R^{nk}$ the system~\eqref{eq:inducedflow:system} has the unique solution
\[
\vec{x} = \vec{C} (\vec{G}^{\top} \vec{\pi} - \vec{b})
\]
where $\vec{b} = (b_{e,i})_{e \in E, i \in \range{k}}$ is the vector of all offsets of the cost functions. Thus, the induced flow function $\inducedflow$ is well-defined. Furthermore, the total flow wrt. $\vec{x}$ on edge~$e$ can be computed as $\bar{x}_e = \vec{u}_{e,i}^{\top} \vec{K} \vec{\Omega} \vec{A}^{-1} ( \vec{G}^{\top} \! \vec{\pi} - \vec{b})$ for all $i \in \range{k}$, where $\vec{u}_{e,i}$ is the unit vector corresponding to player~$i$ and edge~$e$.
\end{lemma}
\begin{proof}
Using the definition of $\mu_e^i(\vec{x}) = a_{e,i} \big( \sum_{j \in S_e} x_e^j + x^i_e \big) + b_{e,i}$ we can rewrite the system~\eqref{eq:inducedflow:system} as 
\begin{align*}
2 x^i_e + \sum_{\substack{j \in S_e \\ j \neq i}} x_e^j &= \rho^i_e
	&\text{if } i \in S_e , \\
x_e^i &= \rho^i_e
	&\text{if } i \notin S_e ,
\end{align*}
where $\rho^i_e = \omega^i_e \, \frac{\pi^i_w - \pi^i_v - b_{e,i}}{a_{e,i}}$. Denote by $\vec{J} \in \R^{mk \times m}$ the block matrix that contains $k$ copies of the $m \times m$ identity matrix $\vec{I}_{m}$. Then the above system can be expressed as
$
\big( \vec{I}_{mk} + \vec{\Omega} \vec{J} ( \vec{\Omega} \vec{J} )^{\top} \big) \vec{x} = \vec{\rho}
$. 
The Sherman-Morrison-Woodbury formula states that 
\[
\big( \vec{I}_{km} + \vec{\Omega} \vec{J} (\vec{\Omega} \vec{J} )^{\top} \big)^{-1}
= \vec{I}_{km} - \vec{\Omega} \vec{J} \big( \vec{I}_m + (\vec{\Omega} \vec{J} )^{\top} \vec{\Omega} \vec{J} \big)^{-1} ( \vec{\Omega} \vec{J} )^{\top}
.
\]
We observe $( \vec{I}_m + (\vec{\Omega} \vec{J})^{\top} \vec{\Omega} \vec{J} )^{-1}=  \big(\vec{I}_m + \sum_{j=1}^k \vec{\Omega}^{j}\big)^{-1} = \big(\! \diag ( \kappa_{e_1} + 1, \dotsc, \kappa_{e_m} +1  ) \big)^{-1} = \bar{\vec{K}}$ and, thus, confirm that this inverse exists. Since $\vec{\rho} = \vec{\Omega} \tilde{\vec{A}}^{-1} (\vec{G}^{\top} \vec{\pi} - \vec{b})$, the system \eqref{eq:inducedflow:system} has the unique solution
$
\vec{x} = \vec{I}_{km} - \vec{\Omega} \vec{J} \bar{\vec{K}} ( \vec{\Omega} \vec{J} )^{\top} \! \vec{\rho}
=  \vec{C} (\vec{G}^{\top} \! \vec{\pi} - \vec{b})
.
$

By definition of the matrices $\vec{K}$ and $\vec{\Omega}$, we see that if we fix some edge $e \in E$, then for every player $i \in \range{k}$ the rows of the matrix $\vec{K} \vec{\Omega}$ are the same. In particular, we obtain
\begin{equation} \label{lem:incudedflow:subclaim}
\vec{u}_{e,i}^{\top} \vec{K} \vec{\Omega} = \frac{1}{\kappa_e + 1} \sum_{j = 1}^{k} \vec{u}_{e,j}^{\top} \vec{\Omega} = \frac{1}{\kappa_e + 1} \sum_{j \in S_e} \vec{u}_{e,j}^{\top}
.
\end{equation}
Then, the total flow on edge $e$ is
\begin{align*}
\bar{x}_e 
&= \sum_{j \in S_e} x_e^j 
= \sum_{j \in S_e} \vec{u}_{e,j}^{\top} \vec{C} \big( \vec{G}^{\top} \vec{\pi} - \vec{b} \big) \\
&= \Big( \sum_{j \in S_e} \vec{u}_{e,j}^{\top} - \sum_{j \in S_e} \vec{u}_{e,j}^{\top} \vec{K} \vec{\Omega} \Big) \tilde{\vec{C}} \big( \vec{G}^{\top} \vec{\pi} - \vec{b} \big) \\
&\stackrel{\eqref{lem:incudedflow:subclaim}}{=} \Big( (\kappa_e+ 1) \vec{u}_{e,i}^{\top} \vec{K} \vec{\Omega} - \kappa_e \vec{u}_{e,i}^{\top} \vec{K} \vec{\Omega} \Big) \tilde{\vec{C}} \big( \vec{G}^{\top} \vec{\pi} - \vec{b} \big) \\
&= \vec{u}_{e,i}^{\top} \vec{K} \vec{\Omega} \tilde{\vec{C}} \big( \vec{G}^{\top} \vec{\pi} - \vec{b} \big)
.
\end{align*}
\end{proof}

Given a potential vector $\vec{\pi}$, Lemma~\ref{lem:incudedflow} allows to compute the flow $\vec{x} = \inducedflow (\vec{\pi}) = \vec{C} ( \vec{G}^{\top} \! \vec{\pi} - \vec{b})$ induced by the potential $\vec{\pi}$. Then the vector $\vec{y} := \vec{G} \vec{x}$ contains the excess flows for every player at every vertex, i.e., the values $y^i_v = \sum_{e \in \delta^+(v)} x^i_e - \sum_{e \in \delta^-(v)} x^i_e$. 

\begin{definition}
For every support, we refer to the matrix $\vec{L} := \vec{G} \vec{C} \vec{G}^{\top}$ as the \emph{block Laplacian matrix} of the support $\mathcal{S}$.
\end{definition}
By the definition of the matrices $\vec{C}$ and $\vec{G}$, we observe that $\vec{L}$ has the form
\[
\vec{L} = 
\begin{pmatrix}
\phantom{-}\vec{L}^{11}		& -\vec{L}^{12} 	& \dotsm & -\vec{L}^{1k} \\
-\vec{L}^{21}	& \phantom{-}\vec{L}^{22}		& \dotsm & -\vec{L}^{2k} \\
\vdots 			& \vdots 			& \ddots& \vdots 		 \\
-\vec{L}^{k1}	& -\vec{L}^{k2} 	& \dotsm & \phantom{-}\vec{L}^{kk}
\end{pmatrix}
\]
where every block $\vec{L}^{ij} = \vec{\Gamma} \vec{C}^{ij} \vec{\Gamma}^{\top}$ is a weighted Laplacian matrix of $G$ with weights $\vec{C}^{ij}$ justifying the name \emph{block Laplacian matrix}. By definition, the excess vector of the flow $\vec{x}$ induced by some potential $\vec{\pi}$ wrt.\ the support $\support$ can be expressed as
\begin{equation}\label{eq:laplacian:excess}
\vec{y} = \vec{L} \vec{\pi} - \vec{d}
,
\end{equation}
where $\vec{d} := \vec{G} \vec{C} \vec{b}$.
Overall, we can compute for every potential vector $\vec{\pi} \in \R^{nk}$ an induced flow vector $\vec{x} = \inducedflow ( \vec{\pi} )$ such that $\vec{\pi}$ and $\vec{x}$ satisfy~\eqref{eq:inducedflow:system}. To ensure that, furthermore, the induced flow $\vec{x} = \inducedflow (\vec{\pi})$ is a Nash equilibrium, we additionally need that $\vec{x}$ is non-negative and satisfies the flow conservation as well as the inequality conditions of Lemma~\ref{lem:equilibrium:potentials}.
To this end, we define the vector
$\smash{\Delta \vec{y} := ( (\Delta \vec{y}^1)^{\top}, \dotsc, (\Delta \vec{y}^k)^{\top} )^\top}$
with $\smash{\Delta y^i_v = -r_i}$ if $v = s_i$, $\smash{\Delta y^i_v = r_i}$ if $v = t_i$ and $\smash{\Delta y^i_v = 0}$ otherwise. We refer to $\Delta \vec{y}$ as the \emph{excess direction}.

We introduce a new coefficient $\sigma^i_{e} := 1$ if $i \in S_e$ and $\sigma^i_{e} = -1$ if $i \notin S_e$. We then define the diagonal matrices $\vec{\Sigma}^i := \smash{ \diag (\sigma^i_{e_1}, \dots, \sigma^i_{e_m}) }$ for every player $i$ and the block diagonal matrix $\vec{\Sigma}$ containing the matrices $\vec{\Sigma}^i$ as block entries, similar to the matrices $\vec{\Omega}^i$ and $\vec{\Omega}$. Then we obtain the following result characterizing $\lambda$-potentials, i.e., potentials that induce a Nash equilibrium flow for demands $\lambda \vec{r}$ for some $\lambda \in [0,1]$.
\begin{theorem} \label{thm:lambdapotential}
The potential vector $\vec{\pi} \in \R^{nk}$ is a $\lambda$-potential with $\lambda \in [0,1]$ if and only if
\begin{subequations}
\label{eq:thm:lambdapotential}
\begin{align}
\vec{L} \vec{\pi} - \vec{d} &= \lambda \Delta \vec{y} 
	\label{eq:thm:lambdapotential:laplace} \\
\label{eq:thm:lambdapotential:constraints}
\vec{W} (\vec{G}^{\top} \! \vec{\pi} - \vec{b} ) &\geq \vec{0}
,
\end{align} 
where $\vec{W} := \vec{\Sigma} \tilde{\vec{C}}$.
\end{subequations}
\end{theorem}
\begin{proof}
A potential vector $\vec{\pi}$ satisfies Equation~\ref{eq:thm:lambdapotential:laplace} if and only if $\vec{\pi}$ and $\inducedflow (\vec{\pi})$ satisfy system~\ref{eq:inducedflow:system} and the induced flow $\inducedflow (\vec{\pi})$ satisfies the demand vector  $\lambda \vec{r}$.

Now fix some player~$i$ and some edge~$e$. Then the row of \eqref{eq:thm:lambdapotential:constraints} that corresponds to the pair $(e,i)$ yields
\begin{equation} \label{eq:thm:lambdapotential:subclaim}
\sigma^i_e \vec{u}^{\top}_{e,i} \tilde{\vec{C}} (\vec{G}^{\top} \! \vec{\pi} - \vec{b} ) \geq \vec{0}
,
\end{equation}
where, again, $\vec{u}_{e,i}$ is the unit vector corresponding to the pair $(e,i)$. If $i \in S_e$, then $\sigma^i_e = \omega^i_e = 1$. Thus, $\sigma^i_e \vec{u}_{e,i}^{\top} \tilde{\vec{C}} = \vec{u}_{e,i}^{\top} \vec{C}$. Hence, \eqref{eq:thm:lambdapotential:subclaim} yields
$0 \leq \vec{u}_{e,i}^{\top} \vec{C} (\vec{G}^{\top} \vec{\pi} - \vec{b}) = \vec{u}_{e,i}^{\top} \inducedflow (\vec{\pi})$ ensuring the non-negativity of the induced flow $\inducedflow (\vec{\pi})$. If, on the other hand, $i \notin S_e$, we can use the formula for the total flow from Lemma~\ref{lem:incudedflow}, and Equation~\eqref{eq:thm:lambdapotential:subclaim} yields
\begin{align*}
0 &\leq
- \vec{u}_{e,i}^{\top} \tilde{\vec{C}} (\vec{G}^{\top} \! \vec{\pi} - \vec{b})  
= - \Big( \vec{u}_{e,i}^{\top} \vec{A}^{-1} \big(\vec{G}^{\top} \! \vec{\pi} - \vec{b}\big) - \bar{x}_e \Big) \\
&= - \Big( \frac{\pi_w^i - \pi_v^i - b_e^i}{a_{e,i}} - \bar{x}_e \Big)
= - \Big( \frac{ \pi_w^i - \pi_v^i - \mu_e^i(\vec{x}) }{a_{e,i}} \Big)
.
\end{align*}
Multiplying this inequality with $- a_{e,i}$ shows that it is equivalent to the inequality conditions in Lemma~\ref{lem:equilibrium:potentials}. 
Hence, a potential satisfies the inequalities \eqref{eq:thm:lambdapotential:constraints} if and only if its induced flow $\inducedflow (\vec{\pi})$ is non-negative for all $i \in S_e$ and the inequality conditions of Lemma~\ref{lem:equilibrium:potentials} are met for all $i \notin S_e$.
Overall, we have shown that the system~\eqref{eq:thm:lambdapotential} is equivalent to the conditions of Lemma~\ref{lem:equilibrium:potentials} concluding the poof.
\end{proof}
Informally, Theorem~\ref{thm:lambdapotential} can be read as follows. If $\vec{\pi}$ is a solution to~\eqref{eq:thm:lambdapotential} then Equation~\eqref{eq:thm:lambdapotential:laplace} ensures that the induced flow $\inducedflow (\vec{\pi})$ is a flow satisfying the demands $\lambda \vec{r}$. The induced flow $\inducedflow (\vec{\pi})$ and $\vec{\pi}$ satisfy system~\eqref{eq:inducedflow:system} by definition. Finally, Equation~\eqref{eq:thm:lambdapotential:constraints} ensures that the flow is non-negative on active edges and that the potential inequalities from Lemma~\ref{lem:equilibrium:potentials} are satisfied for inactive edges.

\subsection{Potential space and line segments}

\label{sec:potentialspace}

Theorem~\ref{thm:lambdapotential} gives a characterization of $\lambda$-potentials, i.e, of potentials that induce to Nash equilibrium flows. The aim of this section is to characterize the set of \emph{all} $\lambda$-potentials of some fixed support, if such potentials exist.
Before we proceed, we have to address some natural sources of ambiguity in our setting. It is obvious from Lemma~\ref{lem:equilibrium:potentials} as well as from Lemma~\ref{lem:incudedflow} that the induced flow does only depend on potential differences on edges rather than absolute potential values. In fact, the vertex-edge-incidence matrix $\vec{\Gamma}$ has rank $n-1$ (assuming that $G$ is weakly connected) and its left nullspace is spanned by the all-ones-vector $\1 \in \R^{n}$. Hence, the left nullspace of the matrix $\vec{G}$ is
\begin{align*}
\mathcal{N} 
	:=  \big\{ (\alpha_1 \1^{\top}, \dotsc, \alpha_k \1^{\top})^\top \in \R^{nk} \; \big\vert \; \alpha_1, \dotsc, \alpha_k \in \R \big\}
.
\end{align*}
To avoid ambiguity in the potential space, we restrict ourselves to potentials that are normalized to zero at the source vertices of the players. That is, we consider the vector space
\begin{align*}
\potentialSpace := \big\{ \vec{\pi} \in \R^{nk} \; \vert \; \pi^i_{s_i} = 0 \text{ for all players } i \in \range{k} \big\}
\end{align*}
that we refer to as the \emph{potential space}.
The potential space is isomorphic to $\R^{nk} / \mathcal{N}$ and has the advantage that the potentials $\pi^i_v$ can be interpreted as the marginal costs of any $s_i$-$v$-path used by player $i$.
Another technical problem with ambiguity may arise if not all vertices of $G$ are connected by active edges for some player. In this case, the potential of that player can be shifted by a constant on each of the connected components without changing the flow. To avoid this problem, we say a support $\mathcal{S}$ is a \emph{total support} if, for every player $i$, the subgraph of $G$ containing only edges~$e$ with $i \in S_e$ is connected. It is not hard to show, that for every Nash equilibrium $\vec x$, there is a corresponding total support $\support$ such that it is no restriction to consider total supports only. 

Denote by $\potentialRegion$ the set of all $\lambda$-potentials for the support $\support$. With the characterization from Theorem~\ref{thm:lambdapotential} we obtain 
\begin{align*}
\potentialRegion = \big\{
	\vec{\pi} \in \potentialSpace
	\; \vert \;
	\exists \lambda \in [0,1] :\;
	&\vec{L} \vec{\pi} - \vec{d} = \lambda \Delta \vec{y}, 
	\vec{W} (\vec{G}^{\top} \! \vec{\pi} - \vec{b} ) \geq \vec{0}
\big\}
.
\end{align*}
We say, a support is \emph{feasible} if it is total and $\potentialRegion \neq \emptyset$, i.e., if there is some $\lambda$ such that there is a Nash equilibrium for the demand vector $\lambda \vec{r}$. 
Our next step is to find solutions to the equation $\vec{L} \vec{\pi} - \vec{d} = \lambda \Delta \vec{y}$. First, we observe the following properties of the matrix $\vec{C}$.

\begin{lemma} \label{lem:matricesC}
If $\support$ is total, then $\ker(\vec{C} \vec{G}^{\top}) = \mathcal{N}$.
\end{lemma}
\begin{proof}
Using the definition of $\tilde{\vec{C}}$ it can easily be seen that $\tilde{\vec{C}}$ is strictly diagonal-dominant and, thus, non-singular. Let $\hat{\vec{C}}$ be the matrix obtained from $\tilde{\vec{C}}$ by removing all rows and columns corresponding to player-edge-pairs $(e,i)$ with $i \notin S_e$. Then this matrix is still strictly diagonal dominant and non-singular. Let $\hat{\vec{G}}$ be the matrix obtained from $\vec{G}$ by deleting the columns corresponding to these player-edge-pairs. Then $\ker( \hat{\vec{C}} \hat{\vec{G}}^{\top} ) = \ker( \hat{\vec{G}}^{\top} ) = \mathcal{N}$ since $\support$ is a total support. Further, the non-zero rows of $\vec{C} \vec{G}^{\top}$ are exactly $\hat{\vec{C}} \hat{\vec{G}}^{\top}$. Hence, $\ker ( \vec{C} \vec{G}^{\top} ) = \ker (\hat{\vec{C}} \hat{\vec{G}}^{\top}) = \mathcal{N}$.
\end{proof}
Lemma~\ref{lem:matricesC} shows that the induced flow function $\inducedflow (\vec{\pi}) = \vec{C} (\vec{G}^{\top} \! \vec{\pi} - \vec{b})$ is injective on $\potentialSpace$, i.e., there is a one-to-one correspondence between (induced) flows and potentials. Furthermore, we see that the nullspace of the block Laplacian matrix $\vec{L}$ is a subset of $\mathcal{N}$. If $\ker( \vec{L} ) = \mathcal{N}$, we can solve $\vec{L} \vec{\pi} - \vec{d} = \lambda \Delta \vec{y}$ uniquely for $\vec{\pi}$ in $\potentialSpace$. In this case, there is a unique equilibrium potential (and flow) with the support $\support$ for every excess vector, and thus for every demand rate $\lambda \vec{r}$. If, on the other hand, $\dim(\ker(\vec{L})) > \dim(\mathcal{N})$, there are actually infinitely many equilibrium potentials, and thus also flows, for the same demand. If this situation occurs for some support $\support$, we call $\support$ $a$-degenerate.

\begin{definition}
A feasible support $\support$ is called \emph{$a$-degenerate} if $\dim(\ker(\vec{L})) > \dim(\mathcal{N})$. We say a game is $a$-degenerate, if there is a feasible support $\support$ that is $a$-degenerate.
\end{definition}

See Example~\ref{ex:8playerInfinite} in Section~\ref{sec:aDegeneracy} for a concrete example with a $a$-degenerate support $\support$ and a continuum of equilibria for fixed demands.

For a matrix $\vec{A}$, we denote by $\vec{A}^+$ a generalized inverse of $\vec{A}$, that is a matrix $\vec{A}^+$ satisfying $\vec{A} \vec{A}^+ \vec{A} = \vec{A}$. If the linear system $\vec{A} \vec{x} = \vec{b}$ is consistent, then $\vec{x} = \vec{A}^+ \vec{b}$ is a solution to the system for every generalized inverse $\vec{A}^+$. Using generalized inverses, we can express a solution to \eqref{eq:thm:lambdapotential:laplace} as $\vec{\pi} = \vec{L}^+ (\lambda \Delta \vec{y} + \vec{d})$. For non-$a$-degenerate supports, there is a unique solution in $\potentialSpace$ to this system. Thus there is a unique generalized inverse that we denote by $\vec{L}^*$ such that $\vec{\pi} = \vec{L}^* (\lambda \Delta \vec{y} + \vec{d})$ is the unique solution to \eqref{eq:thm:lambdapotential:laplace}.
\begin{corollary} \label{cor:finite}
A non-$a$-degenerate game has a finite number of Nash equilibria for every demand~$\vec{r}$.
\end{corollary}
\begin{proof}
If the game is non-$a$-degenerate, then for fixed $\lambda = 1$ there is at most one $\vec{\pi} \in \potentialSpace$ satisfying $\vec{L} \vec{\pi} - \vec{d} = \lambda \Delta \vec{y}$. Hence, there is at most one equilibrium flow for every support $\support$. Since there are only finitely many supports, the claim follows.
\end{proof}
We assume that all supports $\support$ are non-$a$-degenerate and discuss how to deal with $a$-degeneracy in Section~\ref{sec:aDegeneracy}. 
Let us denote with $\Delta \vec{\pi} := \vec{L}^* \Delta \vec{y}$ the \emph{(unoriented) potential direction of the support $\support$} and by $\potentialOffset := \vec{L}^* \vec{d}$ the \emph{offset potential of support $\support$}.
Then the unique solution to \eqref{eq:thm:lambdapotential:laplace} can be expressed as $\vec{\pi} = \lambda \Delta \vec{\pi} + \potentialOffset$ and we see that all $\lambda$-potentials for a given support lie on a line in the potential space. The following lemma gives an alternative representation of $\potentialRegion$.

\begin{lemma} \label{lem:lambdapotentialsLineSegment}
For every feasible support $\support$, there are numbers $0 \leq \lambda^{\min}_{\support} \leq \lambda^{\max}_{\support} \leq 1$ such that
\[
\potentialRegion = \big\{ \lambda \Delta \vec{\pi} + \potentialOffset \;\vert\; \lambda^{\min}_{\support} \leq \lambda \leq \lambda^{\max}_{\support} \big \}
.
\]
\end{lemma}

\begin{proof}
Using that $\lambda \Delta \vec{\pi} + \potentialOffset$ is the unique solution in $\potentialSpace$ of $\vec{L} \vec{\pi} - \vec{d} = \lambda \Delta \vec{y}$, we get
\[
\potentialRegion =\big\{
	\vec{\pi} \in \potentialSpace
	\; \vert \;
	\exists \lambda \in [0,1] :
	\vec{W} \big(\vec{G}^{\top} \! ( \lambda \Delta \vec{\pi} + \potentialOffset ) - \vec{b} \big) \geq \vec{0}
\big\}
.
\]
Denote by $\vec{u}_{e,i}$ the unit vector corresponding to the edge-player-pair $(e,i)$. Then $\vec{w}_{e,i} := \vec{W}^{\top} \vec{u}_{e,i}$ is the $(e,i)$-th row of the matrix $\vec{W}$. 
Then, we express the constraint in $\potentialRegion$ as
\[
\lambda \vec{w}_{e,i}^{\top} \vec{G}^{\top} \! \Delta \vec{\pi} 
\geq 
\vec{w}_{e,i}^{\top} \big( \vec{b}  -  \vec{G}^{\top} \potentialOffset \big)
\]
for every $e \in E$ and $i \in \range{k}$. Solving for $\lambda$, we obtain that the value
\begin{align*}
\lambda^{\min}_{\support} &:=
	\max \bigg\{ \frac{\vec{w}_{e,i}^{\top} \big( \vec{b} \! - \! \vec{G}^{\top} \potentialOffset \big)}{\vec{w}_{e,i}^{\top} \vec{G}^{\top} \! \Delta \vec{\pi}}
		\, \bigg\vert \,
		\vec{w}_{e,i}^{\top} \vec{G}^{\top} \! \Delta \vec{\pi} < 0
	\bigg\} \cup \{0\} \\
\lambda^{\max}_{\support} &:=
	\min \bigg\{ \frac{\vec{w}_{e,i}^{\top} \big( \vec{b} \! - \! \vec{G}^{\top} \potentialOffset \big)}{\vec{w}_{e,i}^{\top} \vec{G}^{\top} \! \Delta \vec{\pi}}
		\, \bigg\vert \,
		\vec{w}_{e,i}^{\top} \vec{G}^{\top} \! \Delta \vec{\pi} > 0
	\bigg\} \cup \{1\}
\end{align*}
are the lower and upper bounds on $\lambda$. (Note that, additionally, $\vec{w}_{e,i}^{\top} (\vec{b} -  \vec{G}^{\top} \potentialOffset) \leq 0$ must be satisfied for all $(e,i)$ with $\vec{w}_{e,i}^{\top} \vec{G}^{\top} \! \Delta \vec{\pi} = 0$. Since we assumed that $\support$ is feasible, this is true.)
\end{proof}
Lemma~\ref{lem:lambdapotentialsLineSegment} shows that the set of all $\lambda$-potentials of a given support $\support$ is a line segment, potentially reduced to a single point when $\lambda_{\support}^{\min} = \lambda_{\support}^{\max}$. In particular, all $\lambda$-potentials of the support $\support$ are convex combinations of the two potentials $\vec{\pi}^{\min}_{\support} := \lambda^{\min}_{\support} \Delta \vec{\pi} + \potentialOffset$ and $\vec{\pi}^{\max}_{\support} := \lambda^{\max}_{\support} \Delta \vec{\pi} + \potentialOffset$, c.f. Figure~\ref{fig:lambdapotentials}. We denote by $\partial \potentialRegion = \{ \vec{\pi}^{\min}_{\support}, \vec{\pi}^{\max}_{\support} \}$ the set of \emph{boundary potentials of the support $\support$}.
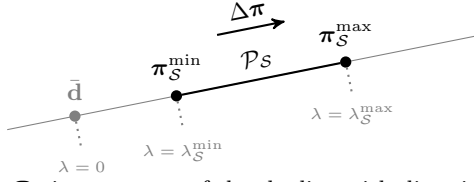
\begin{figure}[tb]
\begin{center}
\begin{tikzpicture}[scale=0.9]
\useasboundingbox (-1,-.5) rectangle (6,2);
\draw[gray]
	(-1,-1/5) -- (6,6/5);
	
\node[solid, gray] (anchor) at (0,0) {} node[above=0 of anchor, gray] {\footnotesize $\potentialOffset$};	

\node[solid] (pimin) at (1.5,1.5/5) {} node[above=0 of pimin] {\footnotesize $\vec{\pi}^{\min}_{\mathcal{S}}$};	
\node[solid] (pimax) at (4,4/5) {} node[above=0 of pimax] {\footnotesize $\vec{\pi}^{\max}_{\mathcal{S}}$};	

\draw[gray, dotted, thick]
	(anchor) -- ++(.1,-.5) node[below] {\tiny $\lambda = 0$}
	(pimin) -- ++(.1,-.5) node[below] {\tiny $\lambda = \lambda^{\min}_{\mathcal{S}}$}
	(pimax) -- ++(.1,-.5) node[below] {\tiny $\lambda = \lambda^{\max}_{\mathcal{S}}$};

\draw[thick]
	(pimin) -- (pimax) node[midway, above, sloped] {\footnotesize $\potentialRegion$};
	
\draw[thick, ->]
	(2.25 - 0.15,2.25/5 +0.75) -- (3.25 - 0.15, 3.25/5 + 0.75) node[midway, above, sloped] {\footnotesize $\Delta \vec{\pi}$};
\end{tikzpicture}
\caption{The set of all $\lambda$-potentials $\potentialRegion$ is a segment of the the line with direction $\Delta \vec{\pi}$ with extreme points $\vec{\pi}^{\min}_{\mathcal{S}}$ and $\vec{\pi}^{\max}_{\mathcal{S}}$. All potentials on this line segment induce Nash equilibria for demands $\lambda \vec{r}$ with $\lambda^{\min}_{\mathcal{S}} \leq \lambda \leq \lambda^{\max}_{\mathcal{S}}$.}
\label{fig:lambdapotentials}
\end{center}
\end{figure}

\subsection{Neighboring supports}
\label{sec:neighbors}

So far, we considered a fixed support $\support$ and characterized the $\lambda$-potentials for this demand. In this subsection, we study how different supports are connected in order to find a way to find feasible supports for all $\lambda \in [0,1]$. We begin by introducing the notion of \emph{neighboring supports}.

\begin{definition}
We say, two supports $\mathcal{S}, \mathcal{S}'$ are \emph{$(e,i)$-neighboring} if
\begin{enumerate}[(i)]
\item $S_{\tilde{e}} = S'_{\tilde{e}}$ for all $\tilde{e} \neq e$ and
\item $S'_e \setminus S_e = \{i\}$ or $S_e \setminus S'_e = \{i\}$.
\end{enumerate} 
\end{definition}
We denote by $N(\support, e, i)$ the unique $(e,i)$-neighboring support $\altSupport$.
By definition, two neighboring supports only differ by the activity status of one edge~$e$ for one particular player~$i$. Hence, the matrices $\vec{C}_{\support}$ and $\vec{C}_{\altSupport}$ as well as the block Laplacians $\vec{L}_{\support}$ and $\vec{L}_{\altSupport}$ are closely related. In fact, the following theorem shows that the respective matrices can be obtained from each other by a rank-$1$-update. To this end, we introduce for every support $\support$ and any pair of edge $e \in E$ and player $i \in \range{k}$ two new vectors
$\normalVec := (\vec{W}_{\support} \vec{G}^{\top})^{\top} \vec{u}_{e,i}$
and 
$\lightNormalVec[\support, e, i] := \big( (\vec{I}_{km} - \vec{K}_{\support} \vec{\Omega}_{\support} ) \vec{G}^{\top} \big)^{\top} \vec{u}_{e,i}$. Geometrically, the vector $\normalVec$ is the normal vector of the hyperplane corresponding to the $(e,i)$-th inequality of $\vec{W} (\vec{G}^{\top} \! \vec{\pi} - \vec{b}) \geq 0$ bounding the polytope $\potentialRegion$.

\begin{theorem} \label{thm:neighboringLaplacians}
For a pair $(e,i) \in E \times [k]$, let $\altSupport := N(\support, e, i)$ be the $(e,i)$-neighbor of the support $\support$.
Then,
\begin{enumerate}[(i)]
\item \label{thm:neighboringLaplacians:laplacians}
	$
	\vec{L}_{\mathcal{S}'} = \vec{L}_{\mathcal{S}} + \lightNormalVec \normalVec^{\top}
	$%
	,
\item \label{thm:neighboringLaplacians:inverses}
	$
	\vec{L}_{\mathcal{S}'}^{+}=\vec{L}_{\mathcal{S}}^{+} - \frac{1}{1 +  \normalVec^{\top} \vec{L}_{\mathcal{S}}^{+} \lightNormalVec} \vec{L}_{\mathcal{S}}^{+} \lightNormalVec \normalVec^{\top} \vec{L}^{+}_{\mathcal{S}}
	$%
	,
\item \label{thm:neighboringLaplacians:directions}
	$ \displaystyle
	\sgn \! \big( \normalVec^{\top}  \Delta \vec{\pi}_{\mathcal{S}} \big) = - \frac{\sigma_{\support}}{\sigma_{\altSupport}} \sgn \! \big( \normalVec[\altSupport, e, i]^{\top} \Delta \vec{\pi}_{\mathcal{S}'} \big)
	$%
	,
\end{enumerate}
where $\sigma_{\support} := \sgn( \det ( \hat{\vec{L}}_{\support} ) )$ is the sign of the determinant of the submatrix of $\vec{L}_{\support}$ obtained by deleting the first row and first column. 
\end{theorem}
\begin{proof}
We begin the proof by establishing a connection between the matrices $\vec{C}_{\support}$ and $\vec{C}_{\altSupport}$ of two neighboring supports.
\begin{claim} \label{clm:thm:neighboringLaplace:neighboringC}
$
\vec{C}_{\mathcal{S}'} = \vec{C}_{\mathcal{S}} +
\big( \vec{I}_{km} - \vec{\Omega}_{\mathcal{S}'} \vec{K}_{\mathcal{S}'}  \big) \vec{u}_{e,i} \vec{u}_{e,i}^{\top} \vec{\Sigma}_{\mathcal{S}} \big( \vec{I}_{km} - \vec{K}_{\mathcal{S}} \vec{\Omega}_{\mathcal{S}}  \big) \vec{A}^{-1}
$ whenever $\support$ and $\altSupport$ are $(e,i)$-neighbors.
\end{claim}
\begin{proof}[Proof of Claim~\ref{clm:thm:neighboringLaplace:neighboringC}]
We make the following observations: The only difference between $\mathcal{S}$ and $\mathcal{S}'$ is by definition the different status of edge $e$ for player $i$.
Thus, $\omega^{i'}_{e'} (\altSupport) = \omega^{i'}_{e'} (\support)$ whenever $e' \neq e$ or $i' \neq i$ and $\omega^i_e (\altSupport) = \omega^{i}_{e} (\support) + \sigma^{\support}_{e,i}$. Thus, also $\vec{\Omega}_{\mathcal{S}'}$ differs only in one value from $\vec{\Omega}_{\mathcal{S}}$ and we get
\[
\vec{\Omega}_{\mathcal{S}'} = \vec{\Omega}_{\mathcal{S}} + \sigma^{\mathcal{S}}_{e, i} \vec{u}_{e,i} \vec{u}^{\top}_{e,i} = \vec{\Omega}_{\mathcal{S}} + \vec{u}_{e,i} \vec{u}^{\top}_{e,i} \vec{\Sigma}_{\mathcal{S}}
.
\]
Further, the number of players using an edge is unchanged for all edges but $e$. Hence, $\kappa^{\mathcal{S}'}_{e'} = \kappa^{\mathcal{S}}_{e'}$ for all $e' \neq e$, and $\kappa^{\mathcal{S}'}_{e} = \kappa^{\mathcal{S}}_{e} + \sigma^{\mathcal{S}}_{e,i}$. Define $\Delta \kappa := - \frac{\sigma^{\mathcal{S}}_{e,i}}{(\kappa^{\mathcal{S}'}_e + 1) (\kappa^{\mathcal{S}}_e +1)} = \frac{1}{\kappa^{\mathcal{S}'}_e + 1} - \frac{1}{\kappa^{\mathcal{S}}_e + 1}$. Then
\[
\bar{\vec{K}}_{\mathcal{S}'} = \bar{\vec{K}}_{\mathcal{S}} + \Delta \kappa \, \vec{u}_{e} \vec{u}_{e}^{\top} = \bar{\vec{K}}_{\mathcal{S}}
\]
where $\vec{u}_e$ denotes the $m$-dimensional unit vector that corresponds to edge $e$. 
We observe that we can express $\vec{K}_{\support}$ as $\vec{K}_{\mathcal{S}} = \vec{J} \bar{\vec{K}}_{\mathcal{S}} \vec{J}^{\top}$ where $\vec{J}$ is the $km \times m$ block matrix that contains $k$ copies of the $m \times m$ identity matrix. Then, in particular, $\vec{K}_{\mathcal{S}'} - \vec{K}_{\mathcal{S}} = \vec{J} \Delta \kappa \vec{u}_e \vec{u}^{\top}_e \vec{J}^{\top}$. 
If we use that $\vec{J}^{\top} \vec{u}_{e,i} = \vec{u}_e$ and $\bar{\vec{K}}_{\mathcal{S}} \vec{u}_{e} = \frac{1}{\kappa^{\mathcal{S}}_e + 1} \vec{u}_e$, then we obtain
\begin{equation}
\begin{split}
\vec{K}_{\mathcal{S}'} \vec{u}_{e,i} \vec{u}_{e,i}^{\top} \vec{\Sigma}_{\mathcal{S}} \vec{K}_{\mathcal{S}}
&= \vec{J} \bar{\vec{K}}_{\mathcal{S}'} \vec{J}^{\top} \vec{u}_{e,i} \vec{u}_{e,i}^{\top} \vec{\Sigma}_{\mathcal{S}} \vec{J} \bar{\vec{K}}_{\mathcal{S}} \vec{J}^{\top}  \\
&= \vec{J} \frac{1}{\kappa^{\mathcal{S}'}_e+1} \vec{u}_{e} \vec{u}_{e}^{\top} \frac{\sigma^{\mathcal{S}}_{e,i}}{\kappa^{\mathcal{S}'}_e+1} \vec{J}^{\top} 
= -\vec{J} \Delta \kappa \vec{u}_{e} \vec{u}_{e}^{\top}  \vec{J}^{\top}
= \vec{K}_{\mathcal{S}} - \vec{K}_{\mathcal{S}'}
.
\end{split} \label{eq:lem:neighboringCandWmatrices:subclaim}
\end{equation}
Overall, all this yields
{\allowdisplaybreaks
\begin{align*}
\vec{C}_{\mathcal{S}'} - \vec{C}_{\mathcal{S}}
&= \Big( \vec{\Omega}_{\mathcal{S}'} \big( \vec{I}_{km} - \vec{K}_{\mathcal{S}'} \vec{\Omega}_{\mathcal{S}'} \big) - \vec{\Omega}_{\mathcal{S}} \big( \vec{I}_{km} - \vec{K}_{\mathcal{S}} \vec{\Omega}_{\mathcal{S}} \big) \Big) \vec{A}^{-1} \\
&= \big(
	(\vec{\Omega}_{\mathcal{S}} + \vec{u}_{e,i} \vec{u}_{e,i}^{\top} \vec{\Sigma}_{\mathcal{S}})
	(\vec{I}_{km} - \vec{K}_{\mathcal{S}'} (\vec{\Omega}_{\mathcal{S}} + \vec{u}_{e,i} \vec{u}_{e,i}^{\top} \vec{\Sigma}_{\mathcal{S}} )) 
	- \vec{\Omega}_{\mathcal{S}} + \vec{\Omega}_{\mathcal{S}} \vec{K}_{\mathcal{S}} \vec{\Omega}_{\mathcal{S}}
	\big) \vec{A}^{-1} \\
&= \big(
	\vec{u}_{e,i} \vec{u}_{e,i}^{\top} \vec{\Sigma}_{\mathcal{S}} 
	- \vec{\Omega}_{\mathcal{S}} \vec{K}_{\mathcal{S}'} \vec{\Omega}_{\mathcal{S}} 
	- \vec{\Omega}_{\mathcal{S}} \vec{K}_{\mathcal{S'}} \vec{u}_{e,i} \vec{u}_{e,i}^{\top} \vec{\Sigma}_{\mathcal{S}}
	- \vec{u}_{e,i} \vec{u}_{e,i}^{\top} \vec{\Sigma}_{\mathcal{S}} \vec{K}_{\mathcal{S'}}  \vec{\Omega}_{\mathcal{S}}
	\\
	& \qquad \qquad \qquad
	- \vec{u}_{e,i} \vec{u}_{e,i}^{\top} \vec{\Sigma}_{\mathcal{S}} \vec{K}_{\mathcal{S}'} \vec{u}_{e,i} \vec{u}_{e,i}^{\top} \vec{\Sigma}_{\mathcal{S}}
	+ \vec{\Omega}_{\mathcal{S}} \vec{K}_{\mathcal{S}} \vec{\Omega}_{\mathcal{S}}
	\big) \vec{A}^{-1} \\
&= \big(
	\vec{u}_{e,i} \vec{u}_{e,i}^{\top} \vec{\Sigma}_{\mathcal{S}}
	- \vec{\Omega}_{\mathcal{S}} ( \vec{K}_{\mathcal{S'}} - \vec{K}_{\mathcal{S}} ) \vec{\Omega}_{\mathcal{S}}
	- \vec{u}_{e,i} \vec{u}_{e,i}^{\top} \vec{\Sigma}_{\mathcal{S}} ( \vec{K}_{\mathcal{S'}} - \vec{K}_{\mathcal{S}} ) \vec{\Omega}_{\mathcal{S}}
	\\
	& \qquad \qquad \qquad
	+ \vec{u}_{e,i} \vec{u}_{e,i}^{\top} \vec{\Sigma}_{\mathcal{S}} \vec{K}_{\mathcal{S}}  \vec{\Omega}_{\mathcal{S}}
	- \vec{\Omega}_{\mathcal{S}} \vec{K}_{\mathcal{S'}} \vec{u}_{e,i} \vec{u}_{e,i}^{\top} \vec{\Sigma}_{\mathcal{S}}
	- \vec{u}_{e,i} \vec{u}_{e,i}^{\top} \vec{\Sigma}_{\mathcal{S}} \vec{K}_{\mathcal{S}'} \vec{u}_{e,i} \vec{u}_{e,i}^{\top} \vec{\Sigma}_{\mathcal{S}}
	\big) \vec{A}^{-1} \\
&\stackrel{\mathclap{\eqref{eq:lem:neighboringCandWmatrices:subclaim}}}{=} \big(
	\vec{u}_{e,i} \vec{u}_{e,i}^{\top} \vec{\Sigma}_{\mathcal{S}}
	+ \vec{\Omega}_{\mathcal{S}} \vec{K}_{\mathcal{S}'} \vec{u}_{e,i} \vec{u}_{e,i}^{\top} \vec{\Sigma}_{\mathcal{S}} \vec{K}_{\mathcal{S}} \vec{\Omega}_{\mathcal{S}}
	+ \vec{u}_{e,i} \vec{u}_{e,i}^{\top} \vec{\Sigma}_{\mathcal{S}} \vec{K}_{\mathcal{S}'} \vec{u}_{e,i} \vec{u}_{e,i}^{\top} \vec{\Sigma}_{\mathcal{S}} \vec{K}_{\mathcal{S}} \vec{\Omega}_{\mathcal{S}}
	\\
	& \qquad \qquad \qquad
	+ \vec{u}_{e,i} \vec{u}_{e,i}^{\top} \vec{\Sigma}_{\mathcal{S}} \vec{K}_{\mathcal{S}}  \vec{\Omega}_{\mathcal{S}}
	- \vec{\Omega}_{\mathcal{S}} \vec{K}_{\mathcal{S'}} \vec{u}_{e,i} \vec{u}_{e,i}^{\top} \vec{\Sigma}_{\mathcal{S}}
	- \vec{u}_{e,i} \vec{u}_{e,i}^{\top} \vec{\Sigma}_{\mathcal{S}} \vec{K}_{\mathcal{S}'} \vec{u}_{e,i} \vec{u}_{e,i}^{\top} \vec{\Sigma}_{\mathcal{S}}
	\big) \vec{A}^{-1} \\	
&= \big(
	\vec{u}_{e,i} \vec{u}_{e,i}^{\top} \vec{\Sigma}_{\mathcal{S}} 
	- \vec{\Omega}_{\mathcal{S}} \vec{K}_{\mathcal{S'}} \vec{u}_{e,i} \vec{u}_{e,i}^{\top} \vec{\Sigma}_{\mathcal{S}} 
	- \vec{u}_{e,i} \vec{u}_{e,i}^{\top} \vec{\Sigma}_{\mathcal{S}} 
	\big) (\vec{I}_{km} - \vec{K}_{\mathcal{S}} \vec{\Omega}_{\mathcal{S}} ) \vec{A}^{-1} \\
&= ( \vec{I}_{km} - \vec{\Omega}_{\mathcal{S}'} \vec{K}_{\mathcal{S}'}  ) \vec{u}_{e,i} \vec{u}_{e,i}^{\top} \vec{\Sigma}_{\mathcal{S}} ( \vec{I}_{km} - \vec{K}_{\mathcal{S}} \vec{\Omega}_{\mathcal{S}}  ) \vec{A}^{-1}
\end{align*}
}
concluding the proof of Claim~\ref{clm:thm:neighboringLaplace:neighboringC}.
\end{proof}

With Claim~\ref{clm:thm:neighboringLaplace:neighboringC}, we get
\begin{align*}
\vec{L}_{\mathcal{S}'} - \vec{L}_\mathcal{S} 
&= \vec{G} ( \vec{C}_{\mathcal{S}'} - \vec{C}_{\mathcal{S}} ) \vec{G}^{\top} \\
&= \vec{G} \big( \vec{I}_{km} - \vec{\Omega}_{\mathcal{S}'} \vec{K}_{\mathcal{S}'}  \big) \vec{u}_{e,i} \vec{u}_{e,i}^{\top}      \underbrace{ \vec{\Sigma}_{\mathcal{S}} \big( \vec{I}_{km} - \vec{K}_{\mathcal{S}} \vec{\Omega}_{\mathcal{S}}  \big) \vec{A}^{-1}    }_{= \vec{\Sigma}_{\mathcal{S}} \vec{C}_{\mathcal{S}} = \vec{W}_{\support}   } \vec{G}^{\top}
= \lightNormalVec[\altSupport, e, i] \normalVec^{\top}
.
\end{align*}
proving \emph{(\ref{thm:neighboringLaplacians:laplacians})}.
Consider the matrix $\hat{\vec{L}}_{\mathcal{S}}$, obtained from $\vec{L}_{\mathcal{S}}$ by deleting the rows and columns corresponding to the source vertex respective source vertex for every player. Let $\snormalVec$ and $\slightNormalVec[\mathcal{S}',e,i]$ be the vectors obtained from $\normalVec$ and $\lightNormalVec[\mathcal{S}', e,i]$, respectively, by removing the rows belonging to the source vertex for every player. Then $\hat{\vec{L}}_{\mathcal{S}}$ is non-singular, since we assume $\mathcal{S}$ is non-$a$-degenerate, and $\hat{\vec{L}}_{\mathcal{S}} = \hat{\vec{L}}_{\mathcal{S'}} + \slightNormalVec[\mathcal{S}',e,i] \snormalVec^{\top}$ we obtain
\begin{equation} \label{eq:thm:uniqueBoundaryCrossing:subclaimProof}
\det (\hat{\vec{L}}_{\mathcal{S}'}) = \det \Big(\hat{\vec{L}}_{\mathcal{S}} \big( \vec{I}_{(n-1)k} + \hat{\vec{L}}_{\mathcal{S}}^{-1} \slightNormalVec[\mathcal{S}',e,i] \snormalVec^{\top} \big) \Big)
= \det \big(\hat{\vec{L}}_{\mathcal{S}} \big) \, \big( 1 + \snormalVec^{\top} \hat{\vec{L}}_{\mathcal{S}}^{-1} \slightNormalVec[\mathcal{S}',e,i] \big)
.
\end{equation}
We note that the term $\vec{v}^{\top} \, \vec{L}^+_{\mathcal{S}} \vec{v}$ is independent of the choice of the generalized inverse $\vec{L}^+_{\mathcal{S}}$ for any vector $\vec{v}$ belonging to the row- and column space of $\vec{L}$ \cite[Theorem~9.4.1]{harville1997matrix}. Thus, we obtain
\[
1 + \normalVec^\top \vec{L}_{\mathcal{S}}^+ \lightNormalVec[\mathcal{S}',e,i]
= 1 + \normalVec^\top \hat{\vec{L}}_{\mathcal{S}}^{*} \lightNormalVec[\mathcal{S}',e,i]
= 1 + \snormalVec^{\top} \hat{\vec{L}}_{\mathcal{S}}^{-1} \slightNormalVec[\mathcal{S}',e,i]
\stackrel{\eqref{eq:thm:uniqueBoundaryCrossing:subclaimProof}}{=} \frac{\det (\hat{\vec{L}}_{\altSupport})}{\det (\hat{\vec{L}}_{\support})}  \neq 0
\]
where $\vec{L}^*_{\mathcal{S}}$ is the generalized inverse of $\vec{L}_{\mathcal{S}}$ obtained from $\hat{\vec{L}}_{\mathcal{S}}$ by adding the appropriate zero rows and columns.
As we assume non-$a$-degeneracy, the determinants $\det (\hat{\vec{L}}_{\mathcal{S}})$ and $\det (\hat{\vec{L}}_{\mathcal{S}'})$ are non-zero. 
We then obtain \emph{(\ref{thm:neighboringLaplacians:inverses})} with the Sherman-Morrison-Woodbury formula for generalized inverses (see, e.g., \cite[Theorem~18.2.14]{harville1997matrix}).

To proceed, we need a second claim that describes the relationship between the matrices $\tilde{\vec{C}}_{\support}$ and $\tilde{\vec{C}}_{\altSupport}$ of two neighboring supports $\support$ and $\altSupport$.
\begin{claim} \label{clm:thm:neighboringLaplace:neighboringCtilde}
If $\support$ and $\altSupport$ are $(e,i)$-neighboring, then, for every player~$j \in \range{k}$,
\[
\vec{u}_{\tilde{e},j}^{\top} \tilde{\vec{C}}_{\mathcal{S}'} = 
\begin{cases}
\frac{\kappa^{\mathcal{S}}_e + 1}{\kappa^{\mathcal{S}'}_e + 1} \vec{u}_{e,j}^{\top} \tilde{\vec{C}}_{\mathcal{S}} +  \frac{\sigma^{\mathcal{S}}_{e,i}}{\kappa^{\mathcal{S}'}_{e} + 1} \big( \vec{u}^{\top}_{e,j} - \vec{u}^{\top}_{e,i} \big) \vec{A}^{-1}
	&\text{if } \tilde{e} = e, \\
\vec{u}_{\tilde{e},j}^{\top} \tilde{\vec{C}}_{\mathcal{S}}
	&\text{if } \tilde{e} \neq e.
\end{cases}
\]
\end{claim}
\begin{proof}[Proof of Claim~\ref{clm:thm:neighboringLaplace:neighboringCtilde}]
We compute
\begin{align*}
\frac{\kappa^{\mathcal{S}}_e + 1}{\kappa^{\mathcal{S}'}_e + 1} (\vec{I}_{km} - \vec{\Omega}_{\mathcal{S}} \vec{K}_{\mathcal{S}} ) \vec{u}_{e, j}
&= 
\frac{\kappa^{\mathcal{S}}_e + 1}{\kappa^{\mathcal{S}'}_e + 1} \vec{u}_{e, j} - \vec{\Omega}_{\mathcal{S}} \vec{K}_{\mathcal{S}} \frac{\kappa^{\mathcal{S}}_e + 1}{\kappa^{\mathcal{S}'}_e + 1} \vec{u}_{e, j} \\
&= \vec{u}_{e, j} - \frac{\tau_{e,i}^{\mathcal{S}}}{\kappa^{\mathcal{S}'}_e + 1} \vec{u}_{e, j} - (\vec{\Omega}_{\mathcal{S}'} - \vec{u}_{e,i} \vec{u}_{e,i}^{\top} \vec{T}_{\mathcal{S}}) \vec{K}_{\mathcal{S}'} \vec{u}_{e, j} \\
&= (\vec{I}_{km} - \vec{\Omega}_{\mathcal{S}'} \vec{K}_{\mathcal{S}}' )  - \frac{\tau_{e,i}^{\mathcal{S}}}{\kappa^{\mathcal{S}'}_e + 1} \vec{u}_{e, j} + \vec{u}_{e,i} 
\smash{\underbrace{  \vec{u}_{e,i}^{\top} \vec{T}_{\mathcal{S}} \vec{K}_{\mathcal{S}'} \vec{u}_{e, j}  }_{= \frac{\tau^{\mathcal{S}}_{e,i}}{\kappa^{\mathcal{S}'}_e + 1}   } } \\
&= (\vec{I}_{km} - \vec{\Omega}_{\mathcal{S}'} \vec{K}_{\mathcal{S}'} ) + \frac{\tau^{\mathcal{S}}_{e,i}}{\kappa^{\mathcal{S}'}_e + 1}  \big( \vec{u}_{e, i}  - \vec{u}_{e, j}  \big)
.
\end{align*}
Multiplying this equation with the matrix $\vec{A}^{-1}$ yields the first case of the claim.
Let $\tilde{e} \neq e$. Then, the number of players using the edge $\tilde{e}$ is the same for support $\support$ as for support $\altSupport$, i.e., $\kappa^{\mathcal{S}'}_{\tilde{e}} = \kappa^{\mathcal{S}}_{\tilde{e}}$. This implies that $\vec{K}_{\mathcal{S}} \vec{u}_{\tilde{e},j} = \vec{K}_{\mathcal{S}'}  \vec{u}_{\tilde{e},j}$. Using again that $\vec{\Omega}_{\mathcal{S}'}  = \vec{\Omega}_{\mathcal{S}} + \vec{u}_{e,i} \vec{u}^{\top}_{e,i} \vec{\Sigma}_{\mathcal{S}}$ yields
\begin{align*}
(\vec{I}_{km} - \vec{\Omega}_{\mathcal{S}} \vec{K}_{\mathcal{S}} ) \vec{u}_{\tilde{e}, j}
&= \vec{u}_{\tilde{e}, j} - (\vec{\Omega}_{\mathcal{S}'} - \vec{u}_{e,i} \vec{u}_{e,i}^{\top} \vec{T}_{\mathcal{S}}) \vec{K}_{\mathcal{S}} \vec{u}_{\tilde{e}, j} 
= \vec{u}_{\tilde{e}, j} - (\vec{\Omega}_{\mathcal{S}'} - \vec{u}_{e,i} \vec{u}_{e,i}^{\top} \vec{T}_{\mathcal{S}}) \vec{K}_{\mathcal{S}'} \vec{u}_{\tilde{e}, j} \\
&= (\vec{I}_{km} - \vec{\Omega}_{\mathcal{S}'} \vec{K}_{\mathcal{S}'} ) \vec{u}_{\tilde{e}, j}
	+ \vec{u}_{e,i} \underbrace{ \vec{u}_{e,i}^{\top} \vec{T}_{\mathcal{S}}  \vec{K}_{\mathcal{S}'} \vec{u}_{\tilde{e}, j} }_{= 0}
=  (\vec{I}_{km} - \vec{\Omega}_{\mathcal{S}'} \vec{K}_{\mathcal{S}'} )  \vec{u}_{\tilde{e}, j}
.
\end{align*}
Again, multiplying by $\vec{A}^{-1}$ yields the second case of the claim.
\end{proof}
Together with the definition $\normalVec := \vec{\Sigma}_{\support} \tilde{\vec{C}}_{\support} \vec{G}^{\top}$, Claim~\ref{clm:thm:neighboringLaplace:neighboringCtilde} implies that $\normalVec[\altSupport, e,i] = - \frac{\kappa^{\support}_e + 1}{\kappa^{\altSupport}_e + 1} \normalVec$.
We use~\emph{(\ref{thm:neighboringLaplacians:inverses})} in order to obtain
\begin{align*}
\normalVec[\altSupport, e,i]^{\top} \Delta \vec{\pi}_{\mathcal{S}'}
&\stackrel{\mathclap{\text{\emph{(\ref{thm:neighboringLaplacians:inverses})}}}}{=}
-\frac{\kappa^{\mathcal{S}}_e + 1}{\kappa^{\mathcal{S}'}_e + 1}
\normalVec^\top \big( \vec{L}_{\mathcal{S}}^{+} - \frac{1}{1 + \normalVec^{\top} \vec{L}_{\mathcal{S}}^{+} \lightNormalVec} \vec{L}_{\mathcal{S}}^{+} \lightNormalVec \normalVec^{\top} \vec{L}^{+}_{\mathcal{S}} \big) \Delta \vec{y} \\
&= -\frac{\kappa^{\mathcal{S}}_e + 1}{\kappa^{\mathcal{S}'}_e + 1} \;
\frac{1}{1 + \normalVec^{\top} \vec{L}_{\mathcal{S}}^{+} \lightNormalVec} \,
\normalVec^{\top} \vec{L}^{+}_{\mathcal{S}} \Delta \vec{y}
= -\frac{\kappa^{\mathcal{S}}_e + 1}{\kappa^{\mathcal{S}'}_e + 1} \;
 \frac{\det (\hat{\vec{L}}_{\support})}{\det (\hat{\vec{L}}_{\altSupport})} \,
\normalVec^{\top} \Delta \vec{\pi}_{\mathcal{S}}
\end{align*}
proving \emph{(\ref{thm:neighboringLaplacians:directions})} by applying the sign-function to both sides.
\end{proof}

Theorem~\ref{thm:neighboringLaplacians} establishes an easy update formula for the block Laplacians and their generalized inverses as well as a relation between the orientation of the direction vectors $\Delta \vec{\pi}_{\support}$ of two neighboring supports.

For a feasible support $\support$, we are interested in neighboring feasible supports $\altSupport$. Neighboring supports with $\potentialRegion \cap \potentialRegion[\altSupport] \neq \emptyset$ are of particular interest, as the line segment $\potentialRegion$ is extended with the line segment $\potentialRegion[\altSupport]$. We call such pairs of neighboring support \emph{continuative neighbors} as the line segment $\potentialRegion$ is ``continued'' by $\potentialRegion[\altSupport]$.
In order to characterize these continuative neighbors, we introduce for every feasible support $\support$ the (possibly empty) sets
\begin{align*}
\lowerContNeighbors &:=
	\bigg\{
	N(\support, e, i)
	\; \bigg\vert \;
	\frac{\vec{w}_{e,i}^{\top} \big( \vec{b} \! - \! \vec{G}^{\top} \potentialOffset \big)}{\vec{w}_{e,i}^{\top} \vec{G}^{\top} \! \Delta \vec{\pi}} = \lambda^{\min}_{\support}
	\bigg\} \\
\upperContNeighbors &:=
	\bigg\{
	N(\support, e, i)
	\; \bigg\vert \;
	\frac{\vec{w}_{e,i}^{\top} \big( \vec{b} \! - \! \vec{G}^{\top} \potentialOffset \big)}{\vec{w}_{e,i}^{\top} \vec{G}^{\top} \! \Delta \vec{\pi}} = \lambda^{\max}_{\support}
	\bigg\}
\end{align*}
of $(e,i)$-neighboring supports of $\support$ that induce the values $\lambda^{\min}_{\support}$ and $\lambda^{\max}_{\support}$. We also refer to these sets as \emph{upper and lower continuative neighbors of $\support$}. The following lemma shows that these sets contain all continuative neighbors.

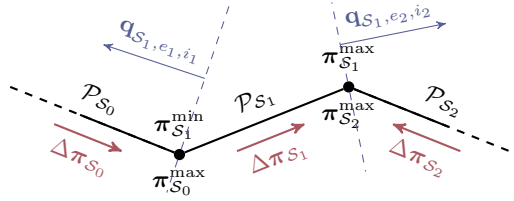
\begin{figure}[tb]
\begin{center}
\begin{tikzpicture}[scale=0.9]
\useasboundingbox (0,-1.25) rectangle (7.5,1.45);
\coordinate (pi0) at (0,0) {};	
\node[solid] (pi1) at (2.5,-1) {} 
	node[above=0 of pi1] {\footnotesize $\vec{\pi}^{\min}_{\mathcal{S}_1}$}
	node[below=0 of pi1] {\footnotesize $\vec{\pi}^{\max}_{\mathcal{S}_0}$};
\node[solid] (pi2) at (5,0) {} 
	node[above=0 of pi2] {\footnotesize $\vec{\pi}^{\max}_{\mathcal{S}_1}$}
	node[below=0 of pi2] {\footnotesize $\vec{\pi}^{\max}_{\mathcal{S}_2}$};
\coordinate (pi3) at (7.5,-1) {};	

\draw[gray!50!playercolor3, dashed]
	(pi1) -- ++(.75,2.25)
	(pi1) -- ++(-1/4*.75,-1/4*2.25);
\draw[gray!50!playercolor3, ->]
	($(pi1) + 1/2*(.75,2.25)$) -- ++(-2.25/1.5, .75/1.5)
		node[midway, above, sloped] {\footnotesize $\normalVec[\mathcal{S}_1,e_1,i_1]$};
	
\draw[gray!50!playercolor3, dashed]
	(pi2) -- ++(-.25,1.25)
	(pi2) -- ++(.25,-1.25);
	
\draw[gray!50!playercolor3, ->]
	($(pi2) + 1/2*(-.25,1.25)$) -- ++(1.25*1.25, .25*1.25)
		node[midway, above, sloped] {\footnotesize $\normalVec[\mathcal{S}_1,e_2,i_2]$};

\draw[thick, dashed]
	(pi0) -- (pi1)
	(pi2) -- (pi3);
\draw[thick]
	(pi0) edge[shorten <=1cm] node[midway, above, sloped] {\footnotesize $\potentialRegion[\mathcal{S}_0]$} (pi1) 
	(pi1) -- (pi2) node[midway, above, sloped] {\footnotesize $\potentialRegion[\mathcal{S}_1]$}
	(pi2) edge[shorten >=1cm] node[midway, above, sloped] {\footnotesize $\potentialRegion[\mathcal{S}_2]$} (pi3);
	
\draw[gray!50!playercolor1, thick, ->]
	($(pi0) + (.75, -.75/2.5) + (-.25/2.5 , -.25) $) -- ($(pi1) + (-.75, .75/2.5) + (-.25/2.5 , -.25) $) 
		node[midway, below, sloped] {\footnotesize $\Delta \vec{\pi}_{\mathcal{S}_0}$};
\draw[gray!50!playercolor1, thick, ->]
	($(pi1) + (.75, .75/2.5) + (.25/2.5 , -.25) $) -- ($(pi2) + (-.75, -.75/2.5) + (.25/2.5 , -.25) $) 
		node[midway, below, sloped] {\footnotesize $\Delta \vec{\pi}_{\mathcal{S}_1}$};
\draw[gray!50!playercolor1, thick, ->]
	($(pi3) + (-.75, .75/2.5) + (-.25/2.5 , -.25) $) -- ($(pi2) + (.75, -.75/2.5) + (-.25/2.5 , -.25) $)
		node[midway, below, sloped] {\footnotesize $\Delta \vec{\pi}_{\mathcal{S}_2}$};
\end{tikzpicture}
\caption{The polytope of $\lambda$-potentials $\potentialRegion[\mathcal{S}_1]$ of some support $\mathcal{S}_1$ and the respective polytopes of two continuative neighbors. The boundary potentials $\vec{\pi}$ of the neighboring supports coincide. The polytopes are separated by hyperplanes with normalvector $\normalVec$.}
\label{fig:continuativeNeighbors}
\end{center}
\end{figure}

\begin{lemma} \label{lem:continuativeNeighbors}
Denote by $\contNeighbors := \{ \altSupport \;\vert\; \potentialRegion \cap \potentialRegion[\altSupport] \neq \emptyset \}$ the set of all continuative neighbors of $\support$. Then, 
\[
\contNeighbors = \lowerContNeighbors \cup \upperContNeighbors
.
\]
Further, the sets  of $\lambda$-potentials do only intersect in exactly one potential, i.e,
\[
\potentialRegion \cap \potentialRegion[\altSupport] = \{ \vec{\pi}^* \} = \partial \potentialRegion \cap \partial \potentialRegion[\altSupport]
\]
whenever $\support$ and $\altSupport$ are continuative neighbors. Moreover, $\inducedflow_{\support} (\vec{\pi}^*) = \inducedflow_{\altSupport} (\vec{\pi}^*)$ in this case.
\end{lemma}
\begin{proof}
For every $e \in E$ and $i \in \range{k}$ with $\altSupport \in \lowerContNeighbors$, we observe
\begin{equation} \label{eq:lem:tightconstraint}
\vec{w}^{\top}_{\support,e,i} ( \vec{G}^{\top} \vec{\pi}^{\min}_{\support} - \vec{b} )
= \vec{w}^{\top}_{\support,e,i} ( \vec{G}^{\top} (\lambda^{\min}_{\support} \Delta \vec{\pi}_{\support} + \bar{\vec{d}}_{\support} ) - \vec{b} )
= 0
\end{equation}
by the definition of $\lowerContNeighbors$.
\begin{claim} \label{clm:lem:continuativeNeighbors:inducedflow}
If $\altSupport := N(\support, e,i) \in \lowerContNeighbors$, then $\inducedflow_{\support}(\vec{\pi}^{\min}_{\support}) = \inducedflow_{\altSupport} ( \vec{\pi}^{\min}_{\support} )$ and $\vec{L}_{\support} \vec{\pi}^{\min}_{\support} - \vec{d}_{\support} = \lambda^{\min}_{\support} \Delta \vec{y} = \vec{L}_{\altSupport} \vec{\pi}^{\min}_{\altSupport} - \vec{d}_{\altSupport}$.
\end{claim}
\begin{proof}[Proof of Claim~\ref{clm:lem:continuativeNeighbors:inducedflow}]
With Claim~\ref{clm:thm:neighboringLaplace:neighboringC} from the proof of Theorem~\ref{thm:neighboringLaplacians} we get
\begin{align*}
\inducedflow_{\altSupport} (\vec{\pi}^{\min}_{\support}) 
&= \vec{C}_{\altSupport} ( \vec{G}^{\top} \vec{\pi}^{\min}_{\support} - \vec{b}) \\ 
&= \big( \vec{C}_{\mathcal{S}} +
\big( \vec{I}_{km} - \vec{\Omega}_{\mathcal{S}'} \vec{K}_{\mathcal{S}'}  \big) \vec{u}_{e,i} \vec{u}_{e,i}^{\top} \vec{\Sigma}_{\mathcal{S}} \big( \vec{I}_{km} - \vec{K}_{\mathcal{S}} \vec{\Omega}_{\mathcal{S}}  \big) \vec{A}^{-1} \big) ( \vec{G}^{\top} \vec{\pi}^{\min}_{\support} - \vec{b}) \\
&= \vec{C}_{\support} ( \vec{G}^{\top} \vec{\pi}^{\min}_{\support} - \vec{b}) + (\vec{I}_{km} - \vec{\Omega}_{\mathcal{S}'} \vec{K}_{\mathcal{S}'}  \big) \vec{u}_{e,i} \vec{w}_{\support, e, i}^{\top} ( \vec{G}^{\top} \vec{\pi}^{\min}_{\support} - \vec{b}) \\
&\stackrel{\smash{\mathclap{\eqref{eq:lem:tightconstraint}}}}{=} \vec{C}_{\support} ( \vec{G}^{\top} \vec{\pi}^{\min}_{\support} - \vec{b}) = \inducedflow_{\support} (\vec{\pi}^{\min}_{\support}) 
.
\end{align*}
Multiplying this equation by $\vec{G}$ yields the second part of the claim.
\end{proof}
\begin{claim} \label{clm:lem:continuativeNeighbors:constraints}
If $\altSupport := N(\support, e,i) \in \lowerContNeighbors$, then $\vec{W}_{\altSupport} ( \vec{G}^{\top} \! \vec{\pi}^{\min}_{\support} - \vec{b} ) \geq \vec{0}$.
\end{claim}
\begin{proof}[Proof of Claim~\ref{clm:lem:continuativeNeighbors:constraints}]
We show that the inequality is satisfied for every row. First, we consider player-edge-pairs $(\tilde{e}, j)$ with $\tilde{e} \neq e$ and $j \in \range{k}$. Then, by Claim~\ref{clm:thm:neighboringLaplace:neighboringCtilde} from the proof of Theorem~\ref{thm:neighboringLaplacians}, we get $\vec{u}^{\top}_{\tilde{e},j} \vec{W}_{\altSupport} = \sigma^{\altSupport}_{\tilde{e}, j} \vec{u}^{\top}_{\tilde{e},j} \tilde{\vec{C}}_{\altSupport} = \sigma^{\support}_{\tilde{e}, j} \vec{u}^{\top}_{\tilde{e},j} \tilde{\vec{C}}_{\support} = \vec{u}^{\top}_{\tilde{e},j} \vec{W}_{\support}$. Thus, $\vec{u}^{\top}_{\tilde{e}, j} \vec{W}_{\altSupport} ( \vec{G}^{\top} \! \vec{\pi}^{\min}_{\support} - \vec{b} ) = \vec{u}^{\top}_{\tilde{e}, j} \vec{W}_{\support} ( \vec{G}^{\top} \! \vec{\pi}^{\min}_{\support} - \vec{b} ) \geq 0$.
Now, consider the edge $e$ and player $i$. Then, $\sigma^{\altSupport}_{e,i} = - \sigma^{\support}_{e,i}$. Using again Claim~\ref{clm:thm:neighboringLaplace:neighboringCtilde} we obtain
\begin{align*}
\vec{u}_{e,i}^{\top} \vec{W}_{\altSupport} ( \vec{G}^{\top} \! \vec{\pi}^{\min}_{\support} - \vec{b} ) 
&= \vec{u}_{e,i}^{\top}  \vec{\Sigma}_{\altSupport} \tilde{\vec{C}}_{\altSupport}  ( \vec{G}^{\top} \! \vec{\pi}^{\min}_{\support} - \vec{b} ) 
= \sigma_{e,i}^{\altSupport} \vec{u}_{e,i}^{\top} \tilde{\vec{C}}_{\altSupport} ( \vec{G}^{\top} \! \vec{\pi}^{\min}_{\support} - \vec{b} )  \\
&= - \sigma_{e,i}^{\support} \; \frac{\kappa^{\support}_e + 1}{\kappa^{\altSupport}_e + 1} \, \vec{u}_{e,i}^{\top} \tilde{\vec{C}}_{\support} ( \vec{G}^{\top} \! \vec{\pi}^{\min}_{\support} - \vec{b} ) 
= - \frac{\kappa^{\support}_e + 1}{\kappa^{\altSupport}_e + 1} \vec{w}^{\top}_{e,i} ( \vec{G}^{\top} \! \vec{\pi}^{\min}_{\support} - \vec{b} ) \stackrel{\smash{\mathclap{\eqref{eq:lem:tightconstraint}}}}{=}  0
.
\end{align*}
Finally, consider any player $j \neq i$ and the edge $e$. Then, we obtain with Claim~\ref{clm:thm:neighboringLaplace:neighboringCtilde}
\begin{align*}
\vec{u}_{e,j}^{\top} \vec{W}_{\altSupport} ( \vec{G}^{\top} \! \vec{\pi}^{\min}_{\support} - \vec{b} ) 
&= \sigma^{\altSupport}_{e,j} \bigg( \frac{\kappa^{\mathcal{S}}_e + 1}{\kappa^{\mathcal{S}'}_e + 1} \vec{u}_{e,j}^{\top} \tilde{\vec{C}}_{\mathcal{S}} +  \frac{\sigma^{\mathcal{S}}_{e,i}}{\kappa^{\mathcal{S}'}_{e} + 1} \big( \vec{u}^{\top}_{e,j} - \vec{u}^{\top}_{e,i} \big) \vec{A}^{-1} \bigg) ( \vec{G}^{\top} \! \vec{\pi}^{\min}_{\support} - \vec{b} )  \\
&= \sigma^{\support}_{e,j} \bigg( \vec{u}_{e,j}^{\top} \tilde{\vec{C}}_{\mathcal{S}} - \frac{\sigma^{\support}_{e,i}}{\kappa^{\mathcal{S}'}_e + 1}  \vec{u}_{e,j}^{\top} \tilde{\vec{C}}_{\mathcal{S}} + \frac{\sigma^{\mathcal{S}}_{e,i}}{\kappa^{\mathcal{S}'}_{e} + 1} \big( \vec{u}^{\top}_{e,j} - \vec{u}^{\top}_{e,i} \big) \vec{A}^{-1} \bigg) ( \vec{G}^{\top} \! \vec{\pi}^{\min}_{\support} - \vec{b} ) \\ 
&=\sigma^{\support}_{e,j} \bigg( \vec{u}_{e,j}^{\top} \tilde{\vec{C}}_{\mathcal{S}} - \frac{\sigma^{\mathcal{S}}_{e,i}}{\kappa^{\mathcal{S}'}_{e} + 1} \, \Big(   \underbrace{ \vec{u}^{\top}_{e, j} \vec{K}_{\support} \vec{\Omega}_{\support} }_{= \vec{u}^{\top}_{e, i} \vec{K}_{\support} \vec{\Omega}_{\support}}   - \vec{u}^{\top}_{e,i} \Big) \vec{A}^{-1} \bigg) ( \vec{G}^{\top} \! \vec{\pi}^{\min}_{\support} - \vec{b} ) \\
&= \vec{u}_{e,j}^{\top} \vec{W}_{\mathcal{S}} ( \vec{G}^{\top} \! \vec{\pi}^{\min}_{\support} - \vec{b} ) + \frac{1}{\kappa^{\altSupport}_e + 1} \, \vec{w}^{\top}_{\support, e, i} ( \vec{G}^{\top} \! \vec{\pi}^{\min}_{\support} - \vec{b} ) \\
&\stackrel{\smash{\mathclap{\eqref{eq:lem:tightconstraint}}}}{=} \vec{u}_{e,j}^{\top} \vec{W}_{\mathcal{S}} ( \vec{G}^{\top} \! \vec{\pi}^{\min}_{\support} - \vec{b} ) \geq 0
.
\end{align*}
Hence, the claim follows.
\end{proof}
Claim~\ref{clm:lem:continuativeNeighbors:inducedflow} and Claim~\ref{clm:lem:continuativeNeighbors:constraints} imply that whenever $\altSupport := N(\support, e, i) \in \lowerContNeighbors$, then $\vec{\pi}^{\min}_{\support} \in \potentialRegion[\altSupport]$. Thus, in particular, $\vec{\pi}^{\min}_{\support} \in \potentialRegion[\altSupport] \cap \potentialRegion$ and $\altSupport$ is a continuative neighbor of $\support$. Since Claim~\ref{clm:lem:continuativeNeighbors:inducedflow} and Claim~\ref{clm:lem:continuativeNeighbors:constraints} hold also if we replace $\lowerContNeighbors$ by $\upperContNeighbors$ and $\vec{\pi}^{\min}_{\support}$ by $\vec{\pi}^{\max}_{\support}$. Thus, we get  $\vec{\pi}^{\max}_{\support}  \in \potentialRegion[\altSupport] \cap \potentialRegion$ and $\altSupport$ is a continuative neighbor of $\support$, whenever $\altSupport := N(\support, e, i) \in \upperContNeighbors$. Overall, we established $\contNeighbors \supseteq \lowerContNeighbors \cup \upperContNeighbors$.

If, on the other hand, $\altSupport := N(\support, e, i) \in \contNeighbors$, then there is some potential $\vec{\pi}^* \in \potentialRegion \cap \potentialRegion[\altSupport]$. Thus we get with the same computation as in the proof of Claim~\ref{clm:lem:continuativeNeighbors:constraints} that
\begin{align*}
\vec{0} \leq \vec{u}^{\top}_{e,i} \vec{W}_{\altSupport} ( \vec{G}^{\top} \! \vec{\pi}^* - \vec{b} )
= - \frac{\kappa^{\support}_e+1}{\kappa^{\altSupport}_e + 1} \, \vec{u}^{\top}_{e,i} \vec{W}_{\support} ( \vec{G}^{\top} \! \vec{\pi}^* - \vec{b} ) \leq \vec{0}
.
\end{align*}
Hence, we established \eqref{eq:lem:tightconstraint} for the potential $\vec{\pi}^*$ implying that the constraint corresponding to the pair $(e,i)$ in $\potentialRegion$ and $\potentialRegion[\altSupport]$ is tight. This implies that $\vec{\pi}^*$ is a boundary potential and $\altSupport \in \lowerContNeighbors$ or $\upperContNeighbors$, establishing $\contNeighbors \subseteq \lowerContNeighbors \cup \upperContNeighbors$.
\end{proof}

Lemma~\ref{lem:continuativeNeighbors} can be interpreted as follows. The continuative neighbors of some support are fully described by $\lowerContNeighbors$ and $\upperContNeighbors$ and can thus be easily identified with the computation of $\lambda^{\min}_{\support}$ and $\lambda^{\max}_{\support}$. Moreover, the $\lambda$-potentials of two continuatively neighboring supports intersect in exactly one potential that is a boundary potential of both supports. And, finally, this potential in the intersection induces the same flow wrt. both supports, i.e., continuative neighbors do not only ``continue'' the sets of $\lambda$-potentials, but also the sets of Nash equilibrium flows.

\section{Membership in $\PPAD$}
\label{sec:membership}

With the insights of the previous section we now turn to show that \atomicsplittable{} is in $\PPAD$. The high level idea for this is the following: We consider all feasible supports as states of the $\PPAD$-graph. We use the support of the empty flow for $\lambda = 0$ as the start state. Then, under some non-degeneracy assumption, we can show that every support other than the initial support and the support of the equilibrium for $\lambda = 1$ have exactly one predecessor and successor support. In order to define these predecessors and successors, we use the previously defined lower and upper continuative neighbors that are unique if the game is non-degenerate.

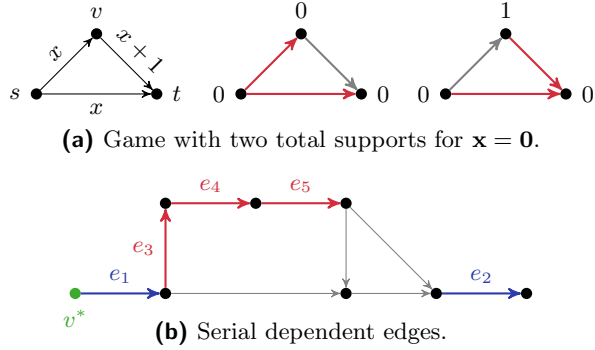
\begin{figure}[tb]
\begin{center}
\begin{subfigure}{0.5\textwidth} 
\begin{center}
\begin{tikzpicture}[scale=0.8]
\newcommand\dx{3.4}
\useasboundingbox (0,0.25) rectangle (2*\dx+2,1.5);
\node[solid] (a) at (0,0) {} node[left=0 of a] {\footnotesize $s$};
\node[solid] (b) at (2,0) {} node[right=0 of b] {\footnotesize $t$};
\node[solid] (c) at (1,1) {} node[above=0 of c] {\footnotesize $v$};
\draw[->,]
		(a) edge node[midway, below, black] {\footnotesize $x$} (b)
		(a) edge node[midway, above, sloped, black] {\footnotesize $x$} (c)
		(c) edge node[midway, above, sloped, black] {\footnotesize $x+1$} (b);

\foreach \x in {0,1} {
	\node[solid] (1+\x) at ({\dx*(\x+1)},0) {};
	\node[solid] (2+\x) at ({\dx*(\x+1)+2},0) {};
	\node[solid] (3+\x) at ({\dx*(\x+1)+1},1) {};
	
	\node[left=0 of 1+\x] {\footnotesize $0$};
	\node[right=0 of 2+\x] {\footnotesize $0$};
	\node[above=0 of 3+\x] {\footnotesize $\x$};
	
	\ifthenelse{\x=1}{\def\edecor{gray}}{\def\edecor{}};
	\ifthenelse{\x=0}{\def\fdecor{gray}}{\def\fdecor{}};
	\draw[->, thick, playercolor1]
		(1+\x) edge (2+\x)
		(1+\x) edge[\edecor] (3+\x)
		(3+\x) edge[\fdecor] (2+\x);
}
\end{tikzpicture}
\end{center}
\caption{Game with two total supports for $\vec{x}=\vec{0}$.}
\label{fig:serialDependence:multipleSupports}
\end{subfigure}
\begin{subfigure}{0.5\textwidth}
\begin{center}
\begin{tikzpicture}[scale=1.2]
\useasboundingbox (0,0.25) rectangle (5,1.5);
\node[solid, playercolor2] (1) at (0,0) {} node[below=0 of 1, playercolor2] {\footnotesize $v^*$};
\node[solid] (2) at (1,0) {};
\node[solid] (3) at (3,0) {};
\node[solid] (4) at (4,0) {};
\node[solid] (5) at (5,0) {};
\node[solid] (6) at (1,1) {};
\node[solid] (7) at (2,1) {};
\node[solid] (8) at (3,1) {};

\draw[->, gray]
	(1) edge[thick, playercolor3, "\footnotesize $e_1$"] (2)
	(2) edge (3)
	(3) edge (4)
	(4) edge[thick, playercolor3, "\footnotesize $e_2$"] (5)
	(2) edge[thick, playercolor1, "\footnotesize $e_3$"] (6)
	(6) edge[thick, playercolor1, "\footnotesize $e_4$"] (7)
	(7) edge[thick, playercolor1, "\footnotesize $e_5$"] (8)
	(8) edge (4)
	(8) edge (3);
\end{tikzpicture}
\end{center}
\caption{Serial dependent edges.} \label{fig:serialDependence:serialDependent}
\end{subfigure}
\caption{\textbf{\textsf{(a)}} A $1$-player game network (leftmost graph) that admits two potentials (values at the vertices) and two supports (active edges in red) for the equilibrium flow $\vec{x} = \vec{0}$ (middle and rightmost graph). \textbf{\textsf{(b)}} The thick edges with the same color are serial-dependent. In particular, $e_1$ is closer to the marked vertex $v^*$ than $e_2$, as well as $e_3$ is closer to $v^*$ than $e_4$ and $e_5$.}
\label{fig:serialDependence}
\end{center}
\end{figure}
Before defining the state set of our $\PPAD$-graph, we have to care about one further source of ambiguity. Consider for example the network from Figure~\ref{fig:serialDependence:multipleSupports}. If there is one player with source $s$ and sink $t$, there are, even if we restrict ourselves to total supports, two different potentials that induce the same flow (in this example the flow $\vec{x} = (0,0,0)^{\top}$). In particular these two potentials belong to different supports. This ambiguity makes it harder to define predecessor and successor functions in a well-defined way. To avoid this problem, we will restrict ourselves to special potentials---namely potentials that can be interpreted as shortest path potentials of the players. To ensure that all potentials are shortest path potentials, we restrict ourselves to certain supports, which we will characterize with the following definition.

We say in any graph $G = (V, E)$ two edges $e, e' \in E$ are \emph{serial-dependent} if there is a subset of vertices $V' \subseteq V$ such that $\delta^+(V) = \{e\}$ and $\delta^-(V) = \{e'\}$ or $\delta^+(V)= \{e'\}$ and $\delta^-(V)= \{e\}$, i.e., the edge-cut induced by $V$ contains exactly one incoming and one outgoing edge, and these cut edges are exactly $e$ and $e'$. 
For some vertex $v$ and two serial-dependent edges $e,e'$ we say \emph{$e$ is closer to $v$ than $e'$} if there is a subset of vertices $V' \subseteq V$ with $v \notin V'$ and $\delta^+(V) = \{e'\}$ and $\delta^-(V) = \{e\}$. See Figure~\ref{fig:serialDependence:serialDependent} for an illustration of these definitions.
We say a total support $\support$ is a \emph{shortest-path-support} if the following holds for all players $i$: If $e$ and $e'$ are serial-dependent in the graph $G$ and $e$ is closer to $s_i$ than $e$ then the edge $e$ is active for player $i$, i.e., $i \in S_e$. In fact, it is not hard to show that if $\support$ is a shortest-path-support and $\vec{\pi} \in \potentialRegion$, then $\vec{\pi}$ contains the length of shortest paths for all players to all vertices\footnote{For this to be true, we require that the graph $G$ is strongly connected. If this is not the case, additional edges can be added with latency functions with high offsets.} with respect to the edge lengths $a^i_e \bar{x}_e + b_e^i$, where $\vec{x} = \inducedflow (\vec{\pi})$ is the induced flow of potential $\vec{\pi}$. 
Regarding the example from Figure~\ref{fig:serialDependence:multipleSupports}, observe that the edges $(s,v)$ and $(v,t)$ are serial-dependent, since $\delta^+(\{v\}) = \{ (v,t) \}$ and $\delta^- (\{v\}) = \{ (s, v) \}$. Thus, only supports where $(s,v)$ is active are shortest path supports which excludes the second support in the rightmost graph.

Finally, we define the set of $\PPAD$-states as
\begin{align*}
\myState := \{ \support \; : \; \support \text{ is a feasible, non-}a\text{-degenerate, } 
	\text{shortest-path support} \}
.
\end{align*}
Next, we want to define successor and predecessor functions that yield predecessor and successor supports for every support $\support \in \myState$. These predecessors and successors will be based on upper and lower continuative neighbors. Although we eliminated many sources of ambiguity there still can be some degeneracy yielding multiple upper and/or lower continuative neighbors.
\begin{definition}
We say a support $\support \in \myState$ is \emph{$b$-degenerate} if $|\lowerContNeighbors \cap \myState | > 1$ or $|\upperContNeighbors \cap \myState | > 1$ and if there is a unique support $\support \in \myState$ for the Nash equilibrium $\vec{x} = \vec{0}$. A game is called \emph{$b$-degenerate} if it has a $b$-degenerate support.
\end{definition}
For the remainder of this section, we assume that every game is non-$b$-degenerate. We discuss $b$-degenerate games in Section~\ref{sec:bDegeneracy}.
\begin{lemma} \label{lem:startandendstate}
Let $\support \in \myState$. Then,
\begin{enumerate}[(i)]
\item
	$\vec{\pi}^{\min}_{\support}$ is a $0$-potential if $\lowerContNeighbors \cap \myState = \emptyset$.
\item
	$\vec{\pi}^{\max}_{\support}$ is a $1$-potential if $\upperContNeighbors \cap \myState = \emptyset$.
\end{enumerate}
\end{lemma}
\begin{proof}
If $\lambda^{\min}_{\support} = 0$, then $\vec{\pi}^{\min}_{\support}$ is a $0$-potential by definition. Thus, in order to prove \emph{(i)}, it suffices to show that, whenever $\lambda^{\min}_{\support} > 0$, we have $\lowerContNeighbors \cap \myState \neq \emptyset$. By definition of $\lambda^{\min}_{\support}$, we know that $\lambda^{\min}_{\support} > 0$ implies that $\lowerContNeighbors \neq \emptyset$. Further we know, that all neighboring supports in $\lowerContNeighbors$ are feasible (as they are continuative neighbors). Thus, we only need to show that there is at least one shortest-path support in $\lowerContNeighbors$. If we assume that $\altSupport := N(\support, e, i) \in \lowerContNeighbors$ is not a shortest-path support, then edge~$e$ must be inactive for player~$i$ in $\altSupport$ and active in $\support$. (Otherwise, $\altSupport$ has one more active edge than the shortest-path support $\support$ implying that $\altSupport$ must be shortest-path itself.) In particular, there must be another edge $e'$ such that $e$ and $e'$ are serial dependent, and $e$ must be closer to $s_i$ than $e'$. (Otherwise, the shortest-path property would not be violated.) But since $e$ and $e'$ are serial-dependent, the flow $x_{e,i} = x_{e',i} = 0$ is zero on both edges in the boundary potential $\vec{\pi}^{\min}_{\support}$. Hence, $N(\support, e', i)$ must also be a continuative neighbor. Therefore, by iterating this argument, we obtain some continuatively neighboring support that must be shortest path.

The same argument holds also for $\vec{\pi}^{\max}_{\support}$ and the upper continuative neighbors $\upperContNeighbors$.
\end{proof}
Lemma~\ref{lem:startandendstate} states, that whenever there is no lower continuative neighbor, then the lower boundary potential $\vec{\pi}^{\min}_{\support}$ is a $0$-potential, i.e., a potential inducing the zero-flow. Similarly, whenever a support has no no upper continuative neighbor, then the support admits the $1$-potential $\vec{\pi}^{\max}_{\support}$, i.e., a potential inducing a Nash equilibrium for demand $\vec{r}$ and thus a solution to \atomicsplittable.

Finally, we introduce the \emph{orientation $\sigma_{\support}$ of the support $\support$} by defining
$
\sigma_{\support} := \sgn( \det ( \hat{\vec{L}}_{\support} ) )
.$
That is, the orientation of the support $\support$ is the sign of the determinant of the submatrix $\hat{\vec{L}}_{\support}$ obtained from $\vec{L}_{\support}$ by deleting the first row and first column.  Theorem~\ref{thm:neighboringLaplacians} implies that the relative orientation of the directions $\Delta \vec{\pi}_{\support}$ of two neighboring supports depends on the orientation of the respective supports. This implies that, for example, the intersection $\potentialRegion \cap \potentialRegion[\altSupport]$ contains the upper boundary potential $\vec{\pi}^{\max}_{\support}$ of one of the supports and the lower boundary potential $\vec{\pi}^{\min}_{\altSupport}$ whenever both supports have the same orientation. If two neighboring, continuative neighbors have opposite orientation the intersection of their $\lambda$-potentials $\potentialRegion \cap \potentialRegion[\altSupport]$ contains either both lower or both upper boundary potentials. For the example in Figure~\ref{fig:continuativeNeighbors}, we see that supports $\support_0$ and $\support_1$ have the same orientation, while $\support_1$ and $\support_2$ have opposite orientation.
Using the orientation, we finally define a predecessor function $\predF : \myState \to \myState \cup \{ \emptyset \}$ and successor function $\succF : \myState \to \myState \cup \{ \emptyset \}$ as follows:
\begin{align*}
\predF (\support) &:= 
\begin{cases}
	\emptyset &
		\text{if } \vec{\pi}^{\min}_{\support} \text{ is a }0\text{-potential}, \\
	\altSupport \in \lowerContNeighbors \cap \myState &
		\text{if }\sigma_{\support} = 1 \text{ and } \vec{\pi}^{\min}_{\support}
		\text{ is not a }0\text{-potential}, \\
	\altSupport \in \upperContNeighbors \cap \myState 
		&\text{if } \sigma_{\support} = -1 \text{ and } \vec{\pi}^{\min}_{\support} 
		\text{ is not a }0\text{-potential}	,
\end{cases}
\\
\succF (\support) &:= 
\begin{cases}
	\emptyset &
		\text{if } \vec{\pi}^{\max}_{\support} \text{ is a }1\text{-potential}, \\
	\altSupport \in \upperContNeighbors \cap \myState &
		\text{if }\sigma_{\support} = 1 \text{ and } \vec{\pi}^{\max}_{\support}
		\text{is not a }1\text{-potential}, \\
	\altSupport \in \lowerContNeighbors \cap \myState &
		\text{if } \sigma_{\support} = -1 \text{ and } \vec{\pi}^{\max}_{\support} 
		\text{is not a }1\text{-potential}.
\end{cases}
\end{align*}
Additionally, we define the constant function $\startF$, that maps to the unique support $\support^0$ of the zero flow. The following lemma shows that the $\support^0$ as well as the functions $\predF$ and $\succF$ are well-defined, computable in polynomial time and compliant.
\begin{lemma} \label{lem:startPredSucc}
The functions $\startF, \predF, \succF$ are well-defined and computable in polynomial time. Furthermore,
\begin{enumerate}[(i)]
\item
	The support $\support^0 := \startF$ has positive orientation.
\item
	If $\predF( \support ) \neq \emptyset$, then $\succF(\predF( \support )) = \support$.
\item
	If $\succF( \support ) \neq \emptyset$, then $\predF(\succF( \support )) = \support$.
\item
	If $\succF ( \support ) = \emptyset$, then there is a $1$-potential with support $\support$.
\end{enumerate}
\end{lemma}
\begin{proof}
Since we assume that the game is non $b$-degenerate, there is a unique support $\support^0$ for the Nash equilibrium $\vec{x} = \vec{0}$. Further, this support can be computed in polynomial time by computing shortest-path form $s_i$ to all vertices $v$ for all players~$i$ with respect to the offsets $b_{e,i}$. Then all edges that are on the (unique) shortest paths are considered active, edges not on shortest paths are considered inactive. Since we assume by non-$b$-degeneracy that all shortest-path are unique, the active edges for every player wrt. $\support^0$ form a spanning tree. Denote by $\check{\vec{G}}$ the matrix obtained from $\vec{G}$ by deleting all columns corresponding to inactive edges and all rows corresponding to the source vertices~$s_i$. Then $\check{\vec{G}}$ is quadratic and it is not hard to show that $\det (\check{\vec{G}}) = \pm 1$ (In fact, $\check{\vec{G}}$ is a block-diagonal matrix, and the diagonal matrices are vertex-edge incidence matrices of the spanning trees of active edges.)
Further, let $\check{\vec{C}}$ be the matrix obtained from $\tilde{\vec{C}}$ by deleting all rows and columns corresponding to inactive edges. Then $\hat{\vec{L}} = \check{\vec{G}} \check{\vec{C}} \check{\vec{G}}^{\top}$. Since $\check{\vec{C}}$ is strictly diagonal dominant by definition, we get that $\sigma_{\support} = \det( \hat{\vec{L}} ) = \det ( \check{\vec{C}} )  > 0$.

\end{proof}

Lemma~\ref{lem:startPredSucc} immediately proves that \atomicsplittable{} is in $\PPAD$.
\begin{theorem}\label{thm:PPADmembership}
\atomicsplittable\ is in $\PPAD$.
\end{theorem}

We obtain the following corollary which mirrors a similar result for non-degenerated bimatrix games. For the proof, we use that support changes appear only for a nullset of demands since there are only finitely many supports. This implies that except for a nullset of demands, the equilibria appear in the relative interior of the polytopes $\mathcal{P}_{\mathcal{S}}$ for some supports $\mathcal{S}$. This further implies that these equilibria are all unique startpoints or endpoints of a $\PPAD$-paths. Subtracting the artificial equilibrium at $\lambda = 0$, we have shown that there is an odd number of equilibria.   

\begin{corollary}
\label{cor:odd}
A non-$b$-degenerate atomic splittable congestion game has an odd number of Nash equilibria except for a nullset of demands.
\end{corollary}

In Section~\ref{sec:bDegeneracy} we prove that a small perturbation of the additive coefficients $b_{e,i}$ makes a $b$-degenerate game non-$b$-degenerate. Further, we show in Corollary~\ref{cor:nonbdegenerate:as} that for almost all $\vec{b} = (b_{e,i})_{e \in E, i \in \range{k}} \in \R^{mk}$ the game with offset vector $\vec{b}$ is non-$b$-degenerate.

\shortversion{}{\section{$\PPAD$-hardness}}
\shortversion{To complement our results, we show}{
In this section, we show} that it is $\mathsf{PPAD}$-hard to compute an equilibrium in an integer-splittable congestion game with affine player-specific costs.

\newcommand{\winlosegame}{\textsc{Approximate-Nash-Win-Lose-Game}}

\begin{appendixproof*}
\shortversion{\subsection{Proof of Theorem~\ref{thm:hardness}}}{}
To prove that computing an equilibrium is $\PPAD$-hard, we reduce from the problem of computing an $\epsilon$-approximate Nash equilibrium in Win-Lose-Games.

\begin{framed}
\noindent\winlosegame\\[6pt]
\begin{tabular}{@{} ll}
\noindent \textsc{Input:~} & matrices $\vec U, \vec V \in \{0,1\}^{n \times n}$, $\epsilon > 0$\\
\noindent \textsc{Output:~} & strategies $\bar{\vec y}, \bar{\vec z} \in \R^{n}_{\geq 0}$  \\[.5ex]
&with $\sum_{j=1}^n \bar{y}_j = \sum_{j=1}^n 
	\bar{z}_j = 1$ such that \\
	&$\quad  \vec y^{\top} \vec U \bar{\vec z} \leq \bar{\vec y}^{\top} \vec U \bar{\vec z} + \epsilon$ \\
	&$\quad \bar{\vec y}^{\top} \vec V \vec z \leq \bar{\vec y}^{\top} \vec V \bar{\vec z} + \epsilon$ \\
	&for all strategies $\vec y,\vec z$.
\end{tabular}
\end{framed}

The following result is due to Chen et al.~\cite{chen2007}.

\begin{theorem}[Chen et al.~\cite{chen2007}]
For any constant $\beta > 0$, \winlosegame\ for matrices $\vec U, \vec V \in \{0,1\}^{n \times n}$ and $\epsilon = n^{-\beta}$ is $\mathsf{PPAD}$-complete.
\end{theorem}

\shortversion{We are now ready to prove Theorem~\ref{thm:hardness}.}{We show the following result.}
\end{appendixproof*}

\begin{theorem} \label{thm:hardness}
\atomicsplittable\ is $\PPAD$-hard.
\end{theorem}

\begin{appendixproof}{Theorem~\ref{thm:hardness}}
We reduce from \winlosegame. Let $(\vec U, \vec V)$ be a win-lose game for some matrices $\vec U, \vec V \in \{0,1\}^{n \times n}$. We proceed to describe the construction of a corresponding atomic splittable congestion games $(G_{\vec U, \vec V},\{1,2\},l)$ with two players with demands $d_1 = d_2 = 1$.
The macro structure of the underlying graph $G_{\vec U, \vec V}$ is shown in Figure~\ref{subfig:reduction_macro}. We note that our construction uses constant cost functions (in particular cost functions with constant cost $0$) for the benefit of a simpler exposition. It is not hard to see that the same construction is valid for cost functions with non-constant functions with very small slopes $a_{e,i} = \delta$ for a very small $\delta > 0$.

There are $n^2$ gadgets $G_{r,c}$ with $r,c \in [n]$ that are arranged in a grid like fashion. The horizontal edges of the grid as well as the edges connecting $s_1$ and $t_1$ to the grid, shown dashed and blue in Figure~\ref{subfig:reduction_macro}, have constant cost $0$ for player~$1$ and constant cost $4n$ for player~$2$. We call these edges \emph{type-1 auxiliary edges}. Similarly, the vertical edges of the grid as well as the edges connecting $s_2$ and $t_2$ to the grid, shown dot-dashed and red in Figure~\ref{subfig:reduction_macro} have constant cost $4n$ for player~$1$ and constant cost $0$ for player~$2$. These edges are called \emph{type-2 auxiliary edges}. Every gadget $G_{r,c}$ with $r,c \in \{1,\dots,n\}$ has four designated verices $\bar{s}_r, \bar{t}_r, \bar{s}_c, \bar{t}_c$. In the macro structure, incoming auxiliary type-1 edges to $G_{r,c}$ from the left are connected to $\bar{s}_r$, incoming auxiliary type-2 edges from above are connected to $\bar{s}_c$, outgoing auxiliary type-1 edges to the right are connected to $\bar{t}_r$, and outgoing auxiliary type-2 edges to below are connected to $\bar{t}_c$. 

\begin{figure}
\centering
\begin{subfigure}{0.48\textwidth}
\centering
\scriptsize
\begin{tikzpicture}[scale=0.85]
\draw (-4,  0) node[solid] (s1) {} node [left] {$s_1$};
\draw ( 0, -4) node[solid] (t2) {} node [below]  {$t_2$};
\draw ( 4,  0) node[solid] (t1) {} node [right] {$t_1$};
\draw ( 0,  4) node[solid] (s2) {} node [above] {$s_2$};
\draw(-2,2) node[oct] (g11) {$\!G_{1,1}\!$};
\draw(-0.5,2) node[oct] (g12) {$\!G_{1,2}\!$};
\draw(2,2) node[oct] (g1n) {$\!G_{1,n}\!$};
\draw(-2,0.5) node[oct] (g21) {$\!G_{2,1}\!$};
\draw(-0.5,0.5) node[oct] (g22) {$\!G_{2,2}\!$};
\draw(2,0.5) node[oct] (g2n) {$\!G_{2,n}\!$};
\draw (0.75,2) node (g1dots) {$\dots$};
\draw (0.75,0.5) node (g2dots) {$\dots$};
\draw (0.75,-2) node (gndots) {$\dots$};
\draw(-2,-2) node[oct] (gn1) {$\!G_{n,1}\!$};
\draw(-0.5,-2) node[oct] (gn2) {$\!G_{n,2}\!$};
\draw(2,-2) node[oct] (gnn) {$\!G_{n,n}\!$};
\draw(-2,-0.75) node (gdots1) {$\vdots$};
\draw(-0.5,-0.75) node (gdots2) {$\vdots$};
\draw(2,-0.75) node (gdotsn) {$\vdots$};
\draw[rowedge] (s1) -- (g11.west);
\draw[rowedge] (g11) -- (g12);
\draw[rowedge] (g12) -- (g1dots);
\draw[rowedge] (g1dots) -- (g1n);
\draw[rowedge] (g1n.east) -- (t1);
\draw[rowedge] (s1) -- (g21.west);
\draw[rowedge] (g21) -- (g22);
\draw[rowedge] (g22) -- (g2dots);
\draw[rowedge] (g2dots) -- (g2n);
\draw[rowedge] (g2n.east) -- (t1);
\draw[rowedge] (s1) -- (gn1.west);
\draw[rowedge] (gn1) -- (gn2);
\draw[rowedge] (gn2) -- (gndots);
\draw[rowedge] (gndots) -- (gnn);
\draw[rowedge] (gnn.east) -- (t1);
\draw[coledge] (s2) -- (g11.north);
\draw[coledge] (g11) -- (g21);
\draw[coledge] (g21) -- (gdots1);
\draw[coledge] (gdots1) -- (gn1);
\draw[coledge] (gn1.south) -- (t2);
\draw[coledge] (s2) -- (g12.north);
\draw[coledge] (g12) -- (g22);
\draw[coledge] (g22) -- (gdots2);
\draw[coledge] (gdots2) -- (gn2);
\draw[coledge] (gn2.south) -- (t2);
\draw[coledge] (s2) -- (g1n.north);
\draw[coledge] (g1n) -- (g2n);
\draw[coledge] (g2n) -- (gdotsn);
\draw[coledge] (gdotsn) -- (gnn);
\draw[coledge] (gnn.south) -- (t2);
\end{tikzpicture}
\caption{Macro structure of the reduction.\label{subfig:reduction_macro}}
\end{subfigure}
\begin{subfigure}{0.48\textwidth}
\scriptsize
\centering
\begin{tikzpicture}[scale=0.85]
\draw (-4,  0) node[solid] (sr) {} node [left] {$\bar{s}_r$};
\draw ( 0, -4) node[solid] (tc) {} node [below]  {$\bar{t}_c$};
\draw ( 4,  0) node[solid] (tr) {} node [right] {$\bar{t}_r$};
\draw ( 0,  4) node[solid] (sc) {} node [above] {$\bar{s}_c$};

\draw ( 0,  2.5) node[solid] (nr) {};
\draw ( 0, -2.5) node[solid] (nc) {};

\draw (-3,2.5) node[solid] (sr1) {};
\draw (-1,2.5) node[solid] (tr1) {};
\draw[-latex,thick] (sr1) -- (tr1) node[above,midway] {$e^1_{r,c,1}$};
\draw (-3,1.0) node[solid] (sr2) {};
\draw (-1,1.0) node[solid] (tr2) {};
\draw[-latex,thick] (sr2) -- (tr2) node[above,midway] {$e^1_{r,c,2}$};
\draw (-3,-0.5) node[solid] (sr3) {};
\draw (-1,-0.5) node[solid] (tr3) {};
\draw[-latex,thick] (sr3) -- (tr3) node[above,midway] {$e^1_{r,c,3}$};
\draw (-2,-1.5) node (rdots) {$\vdots$};
\draw (-3,-2.5) node[solid] (srT) {};
\draw (-1,-2.5) node[solid] (trT) {};
\draw[-latex,thick] (srT) -- (trT) node[above,midway] {$e^1_{r,c,T}$};
\draw[rowedge] (sr) -- (sr1);
\draw[rowedge] (sr) -- (sr2);
\draw[rowedge] (sr) -- (sr3);
\draw[rowedge] (sr) -- (srT);
\draw[rowedge] (tr1) -- (nr);
\draw[rowedge] (tr2) -- (nr);
\draw[rowedge] (tr3) -- (nr);
\draw[rowedge] (trT) -- (nr);
\draw[coledge] (sc) -- (sr1);
\draw[coledge] (tr1) -- (sr2);
\draw[coledge] (tr2) -- (sr3);
\draw[coledge] (tr3) -- (rdots.north);
\draw[coledge] (rdots.south) -- (srT);
\draw (1,2.5) node[solid] (sc1) {};
\draw (3,2.5) node[solid] (tc1) {};
\draw[-latex,thick] (sc1) -- (tc1) node[above,midway] {$e^2_{r,c,1}$};
\draw (1,1.0) node[solid] (sc2) {};
\draw (3,1.0) node[solid] (tc2) {};
\draw[-latex,thick] (sc2) -- (tc2) node[above,midway] {$e^2_{r,c,2}$};
\draw (1,-0.5) node[solid] (sc3) {};
\draw (3,-0.5) node[solid] (tc3) {};
\draw[-latex,thick] (sc3) -- (tc3) node[above,midway] {$e^2_{r,c,3}$};
\draw (2,-1.5) node (cdots) {$\vdots$};
\draw (1,-2.5) node[solid] (scT) {};
\draw (3,-2.5) node[solid] (tcT) {};
\draw[-latex,thick] (scT) -- (tcT) node[above,midway] {$e^2_{r,c,T}$};
\draw[coledge] (trT) -- (nc);
\draw[coledge] (nc) -- (sc1);
\draw[coledge] (nc) -- (sc2);
\draw[coledge] (nc) -- (sc3);
\draw[coledge] (nc) -- (scT);
\draw[rowedge] (nr) -- (sc1);
\draw[rowedge] (tc1) -- (sc2);
\draw[rowedge] (tc2) -- (sc3);
\draw[rowedge] (tc3) -- (cdots.north);
\draw[rowedge] (cdots.south) -- (scT);
\draw[rowedge] (tcT) -- (tr);
\draw[coledge] (tc1) -- (tc2);
\draw[coledge] (tc2) -- (tc3);
\draw[coledge] (tc3) -- (tcT);
\draw[coledge] (tcT) -- (tc);
\end{tikzpicture}
\caption{Gadget $G_{r,c}$ used in the macro construction.\label{subfig:micro_reduction}}
\end{subfigure}

\caption{Construction for the $\PPAD$-hardness of computing a Nash equilibrium in an atomic splittable congestion game.}
\end{figure}
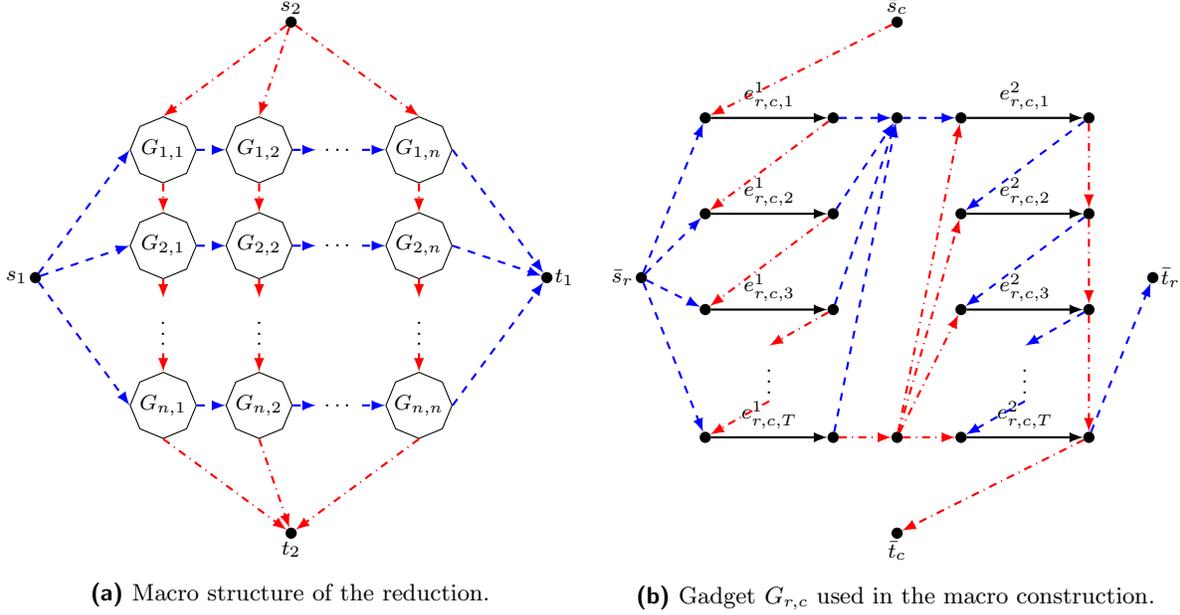

We introduce an additional parameter $T \in \N$ in our construction of the gadgets $G_{r,c}$, parametrizing the number of edges in every gadget. We specify $T$ in Claim~\ref{clm:right_beta} below. For some given $T \geq 1$, 
every gadget $G_{r,c}$ has $T$ \emph{type-1 main edges} $e^1_{r,c,1},\dots,e^1_{r,c,T}$ and $T$ \emph{type-2 main edges} $e^2_{r,c,1},\dots,e^2_{r,c,T}$ as well as some auxiliary edges, see Figure~\ref{subfig:micro_reduction}. The dashed blue edges are type-1 auxiliary edges that have constant cost $0$ for player~$1$ and constant cost $4n$ for player~$2$ while the dash-dotted red edges are type-2 auxiliary edges that have constant cost $4n$ for player~$1$ and constant cost $0$ for player~$2$. The structure of $G_{r,c}$ is such that every path from $\bar{s}_r$ to $\bar{t}_r$ that does not use type-2 auxiliary edges uses exactly one of the type-1 main edges $e_{r,c,1}^1,\dots,e_{r,c,T}^1$ and all type-2 main edges $\{e_{r,c,1},\dots,e_{r,c,T}^2\}$. Similarly, every path from $\bar{s}_c$ to $\bar{t}_c$ that does not use any of the type-1 auxiliary edges uses exactly one of the type-2 main edges $e^2_{r,c,1},\dots,e^2_{r,c,T}$ as well as all of the type-1 main edges $\{e_{r,c,1}^1,\dots,e_{r,c,T}^1\}$. 

For the type-1 and type-2 main edges, we define the player-specific affine costs as follows. The player-specific costs are of the form $c_{e,i}(x) = a_{e,i} x + b_{e,i}$ with 
{ \allowdisplaybreaks
\begin{align*}
a_{e,1} &=
\begin{cases}
1-u_{r,c} & \text{ if } e = e^1_{r,c,t} \text{ with } r,c \in [n], t \in [T]\;,\\
0 & \text{ otherwise}\;;\\
\end{cases} \\
a_{e,2} &= 
\begin{cases}
1-v_{r,c} & \text{ if } e = e^2_{r,c,t} \text{ with } r,c \in [n], t \in [T]\;,\\
0 & \text{ otherwise}\;;
\end{cases} \\
b_{e,i} &= \hphantom{\Biggl\{} 0 \hphantom{1-}\;\text{ for all } e^1_{r,c,t}, e^2_{r,c,t} \text{ with } r,c \in [n], t \in [T]\;.
\end{align*}
}

We claim that in every Nash equilibrium, player~$1$ does not use any of the auxiliary type-2 edges and player~$1$ does not use any of the auxiliary type-1 edges.

\begin{claim}
\label{clm:no_aux}
Let $\vec x$ be a Nash equilibrium of $(G_{\vec U, \vec V},\{1,2\},l)$. Then, $x^1_e = 0$ for every type-2 auxiliary edge~$e$ and $x^2_e = 0$ for every type-1 auxiliary edge. 	
\end{claim}

\begin{proof}[Proof of Claim~\ref{clm:no_aux}]
We show the claim only for player~$1$ since the argumentation for player~$2$ is symmetrical. By Lemma~\ref{lem:equilibrium:marginalcost}, player~$1$ only uses paths with minimal marginal costs in $\vec x$. The marginal cost of every path containing a type-2 auxiliary edge~$e$ is at least
\begin{align*}
a_{e,i} \bar{x}_e + b_{e,i} + a_{e,i} x_e^i \geq 4n.	
\end{align*}
On the other hand, every path that does not contain a type-2 auxiliary edge traverses all gadgets $G_{r,1},\dots,G_{r,n}$ for some $r \in [n]$. In each gadget, $G_{r,c}$ with $c \in [n]$, player~$1$ traverses one of the type-1 main edges $e_{r,c,t_c}$ with $t_c \in [T]$ as well as some type-1 auxiliary edges and some type-2 main edges each of which have no cost for player-1. Thus, the marginal cost of the path is determined by the $n$ type-1 main edges. Let $E_r = \{e_{r,c,t_c}^1 \mid c \in [n]\}$ be the set of these edges. The marginal cost of the path is then given by
\begin{align*}
\sum_{e \in E_r} a_{e,i} \bar{x}_e + b_{e,i} + a_{e,i} x_{e}^i \leq \sum_{e \in E_r} \bar{x}_e + x_{e,i} \leq 3n.
\end{align*}
We conclude that every path of player~$1$ not containing a type-2 auxiliary edge has lower marginal cost than a path containing a type-2 auxiliary edge and the claim follows.
\end{proof}

Notice that player~$1$ has an exponential number of paths from $s_1$ to $t_1$ that do not use any of the type-2 edges. Every such path visits one row of gadgets $G_{r,1},\dots,G_{r,n}$ for some $r \in \{1,\dots,n\}$. Within each gadget $G_{r,c}$, $c \in \{1,\dots,n\}$ the path uses one of the $T$ type-1 main edges $e^1_{r,c,1},\dots,e^1_{r,c,T}$. We claim that in every Nash equilibrium, all paths that visit the same row of gadgets have the same flow of player~$1$.

\begin{claim}
\label{clm:equal_use}
Let $\vec x$ be a Nash equilibrium of $(G_{\vec U, \vec V},\{1,2\},l)$. Then, $x^i_{e^i_{r,c,t}} = x^i_{e^i_{r,c,t'}}$ for all $i \in \{1,2\}$, $r,c \in [n]$, and $t,t' \in [T]$.	
\end{claim}

\begin{proof}[Proof of Claim~\ref{clm:equal_use}]
We show the claim only for player~$1$ since the argumentation for player~$2$ is symmetrical. Let $r,c \in [n]$ and $t,t' \in [T]$ be arbitrary and consider the edges $e = \smash{e^1_{r,c,t}}$ and $e' = \smash{e^1_{r,c,t'}}$. Since $\vec x$ is a Nash equilibrium, by Claim~\ref{clm:no_aux}, we may assume that player~$1$ does not use any type-2 auxiliary edges and player~$2$ does not use any type-1 auxiliary edges. Since every path of player~$2$ that does not use any type-1 auxiliary edges and traverses the gadget $G_{r,c}$ uses all type-1 main edges, this implies in particular that $x^2_{e} = x^2_{e'}$.

If $x^1_{e} = x^1_{e'} = 0$, there is nothing left to show, so it is without loss of generality to assume that $x^1_{e} > 0$. This implies that there is a path $P$ from $s_1$ to $t_1$ carrying positive flow. Let $P'$ be the path that uses the same edges as $P$ except that in gadget $G_{r,c}$ edge $e'$ is used instead of edge $e$. (This means that also some auxiliary type-1 edges are swapped in $G_{r,c}$ but since they have no cost for player~$1$, we may ignore them for the following arguments.) Lemma~\ref{lem:equilibrium:marginalcost} implies
\begin{align}
\label{eq:marginal_costs}
a_{e,1} \bar{x}_e + b_{e,1} + a_{e,1} x_e^1 \leq a_{e',1} \bar{x}_{e'} + b_{e',1} + a_{e',1} x_{e'}^1 
\end{align}
which is equivalent to
\begin{align*}
(1-u_{r,c})(2x_e^1 + x_e^2) \leq  (1-u_{r,c})(2x_{e'}^1 + x_{e'}^2)
\end{align*}
using the definition of the cost functions. Further using that $x_{e}^2 = x_{e'}^2$ and that $1- u_{r,c} \geq 0$ this implies $x_{e}^1 \leq x_{e'}^{1}$ and, thus, $x_{e'}^1 > 0$. We conclude that \eqref{eq:marginal_costs} is actually satisfied with equality implying that $x_{e}^1 = x_{e'}^1$.
\end{proof}

For player~$1$ and $r \in \{1,\dots,n\}$ let $x^1_r$ denote the total flow sent along the paths using the $r$th gadget row $G_{r,1},\dots,G_{r,n}$. Similarly, for player~$2$ and $c \in \{1,\dots,n\}$ let $x^2_c$ denote the total flow sent along the paths using the $c$th gadget column $G_{1,c},\dots,G_{n,c}$.

\begin{claim}
\label{clm:right_beta}
Let $\vec x$ be a Nash equilibrium of $(G_{\vec U, \vec V},\{1,2\},l)$ with $T = 2n^{\beta+1}$ for some $\beta > 0$. Then, the strategy profile $(\bar{\vec y}, \bar{\vec z})$ with $\bar{y}_r = x^1_r$ for all $r \in [n]$ and $\bar{z}_c = x^2_c$ for all $c \in [n]$ is an $n^{-\beta}$-approximate Nash equilibrium of $(\vec U, \vec V)$.
\end{claim}

\begin{proof}[Proof of Claim~\ref{clm:right_beta}]
For the sake of a contradiction, suppose that $(\bar{\vec y},\bar{\vec z})$ is not an $n^{-\beta}$-approximate Nash equilibrium of $(\vec U, \vec V)$. Due to symmetry, it is without loss of generality to assume that player~$1$ has an alternative strategy $\vec y$ with $\vec y^{\top} \vec U \bar{\vec z} > \bar{\vec y}^{\top} \vec U \bar{\vec z} + n^{-\beta}$. This implies in particular that
\begin{align}
\label{eq:no_approx}
\max_{r \in  [n]} \sum_{c \in [n]} u_{r,c} \bar{z}_c > \min_{r \in [n] : \bar{y}_r > 0} \sum_{c=1}^n u_{r,c} \bar{z}_c + n^{-\beta}.	
\end{align}
Let $r^* \in \arg\max_{r \in [n]} \sum_{c \in [n]} u_{r,c} \bar{z}_c$ and let $r' \in \argmin_{r \in [n] : \bar{y}_r > 0} \sum_{c=1}^n u_{r,c} \bar{z}_c$ be arbitrary.
Since $\vec x$ is a Nash equilibrium for $G_{\vec U,\vec V}$, all used paths have the same marginal total cost, and unused paths have higher marginal total cost. Using Claim~\ref{clm:equal_use}, this implies in particular that
\begin{align*}
\frac{\partial}{\partial x_{r'}^1} \Biggl(\sum_{c \in [n]} T (1-u_{r',c})\biggl(\frac{x^1_{r'}}{T} + x^2_c\biggr) \frac{x^1_{r'}}{T} \Biggr) 
&\leq \frac{\partial}{\partial x_{r^*}^1} \Biggl( \sum_{c \in [n]} T (1-u_{r^*,c})\biggl(\frac{x^1_{r^*}}{T} + x^2_c\biggr) \frac{x^1_{r^*}}{T} \Biggr)
\intertext{which gives}
\sum_{c \in [n]} \frac{2 (1-u_{r',c}) x^1_{r'}}{T} + (1-u_{r',c}) x^2_c 
 &\leq \sum_{c \in [n]} \frac{2(1-u_{r^*,c}) x^1_{r^*}}{T} + (1-u_{r^*,c})x^2_c.
\intertext{Using the definition of $(\bar{\vec y}, \bar{\vec z})$ we obtain}
\sum_{c \in [n]} \frac{2(1 - u_{r',c}) \bar{y}_{r'}}{T} + (1-u_{r',c}) \bar{z}_c 
 &\leq \sum_{c \in [n]} \frac{2(1-u_{r^*,c}) \bar{y}_{r^*}}{T} + (1- u_{r^*,c})\bar{z}_c.
\end{align*}
Finally, we obtain
\begin{align*}
\sum_{c \in [n]} u_{r^*,c} \bar{z}_c
&\leq \sum_{c \in [n]} u_{r',c} \bar{z}_c + \sum_{c \in [n]} \frac{2(1-u_{r^*,c}) \bar{y}_{r^*}-2(1-u_{r',c})\bar{y}_{r'}}{T} \\
&\leq \sum_{c \in [n]} u_{r',c} \bar{z}_c + \sum_{c \in [n]} \frac{2\bar{y}_{r^*}}{T} \leq \sum_{c \in [n]} u_{r',c} \bar{z}_c + \frac{2n}{T} = \sum_{c=1}^n u_{r',c} \bar{z}_c + n^{-\beta},
\end{align*}
contradicting \eqref{eq:no_approx}.
\end{proof}
Using that it is $\PPAD$-complete to compute an $n^{-\beta}$-approximate equilibrium of a two-player win-lose game for any $\beta > 0$, we conclude that the computation of a Nash equilibrium of an atomic-splittable congestion game is $\PPAD$-hard as well. 
\end{appendixproof}

With similar arguments, we can also show that the computation of a multi-class Wardrop equilibrium is $\PPAD$-complete. A multi-class Wardrop equilibrium is a multi-commodity flow that satisfies the characterization via shortest path potentials of Lemma~\ref{lem:equilibrium:potentials} for the original cost functions $l_{e,i}$ instead for the marginal cost functions $\mu_{e}^i$, i.e, $\vec x$ is a Wardrop equilibrium if and only if for all $i \in [k]$ there is a potential vector $\vec \pi^i$ with $\pi_w^i - \pi_v^i = l_{e,i}(\bar{x}_e)$ if $x_{e}^i > 0$, and $\pi_w^i - \pi^i_v \leq l_{e,i}(\bar{x}_e)$ if $x_e^i = 0$ for all $e = (v,w) \in E$. We prove the following result settling an open question in \cite{meunier2019}.

\begin{theorem}
\label{thm:wardrop_hardness}
Computing a multi-class Wardrop equilibrium for commodity-specific affine costs is $\PPAD$-hard. 	
\end{theorem}

\begin{appendixproof}{Theorem~\ref{thm:wardrop_hardness}}
We use a similar construction as in the proof of Theorem~\ref{thm:hardness} with the only exception that player~$1$ is replaced by a population of players of size $1$ and player~$2$ is replaced by a population of players of size $1$. 

\begin{claim}
\label{clm:no_aux_wardrop}
Let $\vec x$ be a Wardrop equilibrium of $(G_{\vec U,\vec V},\{1,2\},l)$. Then $x_e^1 = 0$ for every type-2 auxiliary edge $e$ and $x_e^2 = 0$ for every type-1 auxiliary edge $e$.
\end{claim}
	
\begin{proof}[Proof of Claim~\ref{clm:no_aux_wardrop}]
We again show the claim only for population $1$ since the argumentation for population $2$ is symmetric. In a Wardrop equilibrium only paths with minimum cost are used. The total cost of any path containing a type-2 auxiliary edge $e$ is at least $4n$. On the other hand, every path that does not contain a type-2 edge contains $n$ type-1 main edges contained in some set $E_r = \{e_{r,c,t_c} \mid c \in [n]\}$ with $t_c \in [T]$ for all $c \in [n]$. The total cost of these edges is equal to
\begin{align*}
\sum_{e \in E_r} a_{e,i} \bar{x}_e + b_{e,i} \leq \sum_{e \in E_r} x_{e}^1 + x_e^2 \leq 2n, 	
\end{align*}
 so that we conclude that no path containing a type-2 edge is used in a Wardrop equilibrium.
\end{proof}

We proceed to show that population~$i$ chooses all type-$i$ main edges within a gadget $G_{r,c}$ with the same flow.

\begin{claim}
\label{clm:equal_share_wardrop}
Let $\vec x$ be a Wardrop equilibrium of $(G_{\vec U, \vec V},\{1,2\},l)$. Then, $\smash{x^i_{e^i_{r,c,t}} = x^i_{e^i_{r,c,t'}}}$ for all $i \in \{1,2\}$, $r,c \in [n]$ and $t,t' \in [T]$.	
\end{claim}
\begin{proof}[Proof of Claim~\ref{clm:equal_share_wardrop}]
We again show the result only for population $1$ since the argumentation for population~$2$ is symmetric. Let $r,c \in [n]$ and $t,t' \in [T]$ be arbitrary and consider the type-1 main edges $e = e^1_{r,c,t}$ and $e = e^1_{r,c,t'}$ in gadget $G_{r,c}$. Using Claim~\ref{clm:no_aux_wardrop}, no population~$i$ uses auxiliary edges that are not of type $i$. This implies in particular that the flow of population~$2$ on edges $e$ and $e'$ is equal, i.e., $x_{e}^2 = x_{e'}^2$.

If $x_{e}^1 = x_{e'}^1 = 0$ there is nothing left to show, so it is without loss of generality to assume that $x_{e}^1 > 0$. This implies that there is a path $P$ from $s_1$ to $t_1$ carrying positive flow. Let $P'$ be the path that uses the same edges as $P$ except that in gadget $G_{r,c}$ the type-$1$ main edge $e'$ is used instead of edge $e$. The equilibrium condition of Wardrop flows implies that
\begin{align}
\label{eq:wardrop_condition}
a_{e,1} \bar{x}_e + b_{e,1} \leq a_{e',1} \bar{x}_{e'} + b_{e',1}.
\end{align}
Using the definition of the cost functions, this is equivalent to
\begin{align*}
(1- u_{r,c})(x_e^1 + x_e^2) \leq (1-u_{r,c})(x_{e'}^1 + x_{e'}^2).	
\end{align*}
Further using that $x_{e}^2 = x_{e'}^2$ this implies that $x_{e}^1 \leq x_{e'}^1$ and, thus, $x_{e'}^1 > 0$. We conclude that \eqref{eq:wardrop_condition} was satisfied with equality implying that $x_{e}^1 = x_{e'}^1$.
\end{proof}

For population~$1$ and $r \in [n]$, let $x_r^1$ denote the total flow sent along the paths using the $r$th gadget row $G_{r,1},\dots,G_{r,n}$. Similarly, define $x_{c}^2$ as the total flow sent by population~$2$ along the paths in the $c$th gadget column $G_{1,c},\dots,G_{n,c}$.

\begin{claim}
\label{clm:right_beta_wardrop}
Let $\vec x$ be a Wardrop equilibrium of $(G_{\vec U,\vec V},\{1,2\},l)$ with $T= n^{\beta+1}$ for some $\beta > 0$. Then, the strategy profile $(\bar{\vec y},\bar{\vec z})$ with $\bar{y}_r = x_{r}^1$ for all $r \in [n]$ and $\bar{z}_c = x_c^2$ for all $c \in [n]$ is an $n^{-\beta}$-approximate Nash equilibrium of $(\vec U,\vec V)$.	
\end{claim}

\begin{proof}[Proof of Claim~\ref{clm:right_beta_wardrop}]
For the sake of a contradiction, suppose that $(\bar{\vec y}, \bar{\vec z})$ is not an $n^{-\beta}$-approximate Nash equilibrium of $(\vec U, \vec V)$. Due to symmetry, it is without loss of generality to assume that player~$1$ has an alternative strategy $\vec y$ with $\vec y^{\top} \vec U \bar{\vec{z}} > \bar{\vec y}^{\top} \vec U \bar{\vec z} + n^{-\beta}$. This implies in particular that
\begin{align}
\label{eq:wardrop_contradiction}
\max_{r \in [n]} \sum_{c \in [n]} u_{r,c} \bar{z}_c > \min_{r \in [n]: \bar{y}_r>0} \sum_{c \in [n]} u_{r,c} \bar{z}_c + n^{-\beta}.	
\end{align}
Let $r^* \in \arg\max_{r \in [n]} \sum_{c \in [n]} u_{r,c} \bar{z}_c$ and let $r' \in \arg\min_{r \in [n] : \bar{y}_r > 0} u_{r,c} \bar{z}_c$ be arbitrary. Since $\vec x$ is a Wardrop equilibrium for $G_{\vec U, \vec V}$, all used paths have the same cost, and unused paths have higher costs. Using Claim~\ref{clm:equal_share_wardrop}, this implies
\begin{align*}
\sum_{c \in [n]} T(1- u_{r',c})	\biggl(\frac{x^1_{r'}}{T} + x_c^2\biggr) 
&\leq \sum_{c \in [n]} T(1- u_{r^*,c})\biggl(\frac{x^1_{r^*}}{T} + x^2_c\biggr).
\end{align*}
Using the definition of $(\bar{\vec y}, \bar{\vec z})$, we obtain
\begin{align*}
\sum_{c \in [n]}	 T(1- u_{r',c})\biggl(\frac{\bar{y}_{r'}}{T} + \bar{z}_c\biggr)\
&\leq \sum_{c \in [n]} T (1- u_{r^*,c}) \biggl(\frac{\bar{y}_{r^*}}{T} + \bar{z}_c\biggr)
\intertext{which yields}
\sum_{c \in [n]} Tu_{r^*,c} \bar{z}_c 
&\leq \sum_{c \in [n]} Tu_{r',c}\bar{z}_c + \sum_{c \in [n]} (1 - u_{r^*,c})\bar{y}_{r^*} - (1-u_{r',c})\bar{y}_{r'}
\intertext{and equivalently}
\sum_{c \in [n]}  u_{r^*,c} \bar{z}_c 
 &\leq \sum_{c \in [n]} u_{r',c}\bar{z}_c + \sum_{c \in [n]} \frac{(1 - u_{r^*,c})\bar{y}_{r^*} - (1-u_{r',c})\bar{y}_{r'}}{T} \\
&\leq  \sum_{c \in [n]} u_{r',c} \bar{z}_c + \sum_{c \in [n]} \frac{\bar{y}_{r^*}}{T}
\leq \sum_{c \in [n]} + n^{-\beta}
\end{align*}
contradicting \eqref{eq:wardrop_condition}.
\end{proof}
Since the computation of an $\epsilon$-approximate equilibrium for win-lose games is $\PPAD$-complete, we conclude that also the computation of a multi-class Wardrop equilibrium is $\PPAD$-complete.
\end{appendixproof}

\section{Player-independent cost functions}

In this section, we briefly discuss the special case of player-independent cost functions, i.e., games where the cost functions on the edges are
$
l_{e,i} (x_e) := l_e(x_e) = a_e x_e + b_e
$
and, thus, independent of the player. In this case, we obtain the following result for the block Laplacians of any total support.
\begin{lemma} \label{lem:symmetric:blockLaplacian}
If the coefficients of the cost functions $l_{e,i}$ are independent of the player $i$, then, for any total support $\support$,
	$\vec{L}_{\support}$ is symmetric and positive semi-definite
	and
	$\kernel (\vec{L}_{\support}) = \mathcal{N}$.%
\end{lemma}
\begin{proof}
In the case of player independent cost functions, the matrix $\vec{C}$ contains the off-diagonal blocks
$
- \vec{C}^{ij} := - \diag\Big(\tfrac{1}{(\kappa_{e_1}+1)a_{e_1}}\omega_{e_1}^{ij}, \dots, \tfrac{1}{(\kappa_{e_m}+1)a_{e_m}}\omega_{e_m}^{ij}\Bigr) = - \vec{C}^{ji}
$%
.
Thus, the matrix $\vec{C}$ is symmetric. Moreover, in every row of $\vec{C}$, the diagonal element is $\tfrac{\kappa_{e}+1-\omega_{e}^i}{(\kappa_{e_1}+1)a_{e}}$ while the (non-zero) off-diagonal elements are $- \tfrac{1}{(\kappa_{e}+1)a_{e}}\omega_{e}^{ij}$. Thus, the difference between diagonal and off-diagonal elements in every row is $\tfrac{1}{(\kappa_e + 1) a_e} > 0$ and, hence, $\vec{C}$ is diagonal-dominant. (In the player-specific case case this only holds true for the matrix $(\vec{I}_{km} - \vec{K} \vec{\Omega}$.)
As $\vec{C}$ is symmetric and diagonal-dominant it is positive semi-definite and, hence, $\vec{L} = \vec{G} \vec{C} \vec{G}^{\top}$ is positive semi-definite as well. Furthermore, this implies that $\vec{G} \vec{C} \vec{G}^{\top} \vec{v} = \vec{0}$ if and only if $\vec{C} \vec{G}^{\top} \vec{v} = \vec{0}$. With Lemma~\ref{lem:matricesC}, $\kernel (\vec{L}_{\support}) = \mathcal{N}$ follows.
\end{proof}

Sylvester's criterions states that a matrix is positive semi-definite if and only if all principal minors of the matrix are positive semi-definite. Thus, by Lemma~\ref{lem:symmetric:blockLaplacian}, the principal minor $\sigma_{\support} = \det(\hat{\vec{L}}_{\support})$ is non-negative.
Additionally, Lemma~\ref{lem:symmetric:blockLaplacian}\longversiononly{$(ii)$} states that all supports are non-$a$-degenerate implying that $\sigma_{\support} > 0$, and we obtain the following result.
\begin{theorem} \label{thm:symmetric}
If the coefficients of the cost functions $l_{e,i}$ are independent of the player $i$, then%
\begin{enumerate}[(i)]
\item
	the game is non-$a$-degenerate.
\item
	the orientation $\sigma_{\support}$ of every total support $\support$ is positive.
\item
	For every demand vector $\vec{r}$, there is a unique Nash equilibrium $\vec{x}$.
\end{enumerate}
\end{theorem}
\begin{proof}
Statements \emph{(i)} and \emph{(ii)} follow directly from Sylvester's criterion.
If we assume that there are two distinct equilibria for some demand $\vec{r}$, then these equilibria must have two different supports---otherwise all convex combinations of the two equilibria would also be equilibria, yielding a contradiction to the non-$a$-degeneracy assumption.
Using the $\predF$ function, we can trace back a sequence of supports for both of these equilibria until we either end up with two different supports for the zero flow, or with a common predecessor support. Both of these facts imply that the game must be $b$-degenerate. 
But even if the game is $b$-degenerate, the assumption would imply that for very small perturbation $\pertVar$ as described in Section~\ref{sec:bDegeneracy}, the perturbed game still has two equilibria with different supports for the same demand. But, as argued in Section~\ref{sec:bDegeneracy}, we know that this game must be non-$b$-degenerate and, thus, we again obtain a contradiction and the claim follows.
\end{proof}
\shortversion{\paragraph{Parametric computation.}}{\section{Parametric computation}}
\longversiononly{\label{sec:parametric}}

Using the trivial $\PPAD$-algorithm, beginning with the support given by $\startF$, we can iterate through (possibly exponentially many) supports using the function $\succF$, until we eventually obtain a support $\support$ for the Nash equilibrium for the demand $\vec{r}$.

As a byproduct of this procedure, we obtain a sequence of supports $\support_0, \dotsc, \support_\tau$, where $\support_{l}$ is the support of Nash equilibria for demands $\lambda \vec{r}$ with  $ \lambda \in [ \lambdaMin[\support_l], \lambdaMax[\support_l] ]$. By construction of the successor function, we know that $\bigcup_{l=0}^{\tau} [ \lambdaMin[\support_l], \lambdaMax[\support_l] ] = [0,1]$. Thus, we can define a function $\lambda \mapsto \support(\lambda)$ that maps every $\lambda$ to a support of some Nash equilibrium for demand $\lambda \vec{r}$. Given the support $\support(\lambda)$ and $\lambda$, we can easily compute a Nash equilibrium for the demand $\lambda$ as we have discussed in subsection~\ref{sec:laplacian}. Hence, we obtain a function $\vec{f} :[0, 1] \to \R^{mk}$ such that $\vec{f} (\lambda)$ is a Nash equilibrium for the demand $\lambda \vec{r}$, i.e., a function solving \parametricatomicsplittable. 

\begin{theorem} \label{thm:parametric:PSPACE}
\parametricatomicsplittable{} can be solved in polynomial space.
\end{theorem}
\begin{appendixproof}{Theorem~\ref{thm:parametric:PSPACE}}
Since the function $\lambda \mapsto \vec{\pi} (\lambda)$ that maps a value $\lambda$ to the potential of a Nash equilibrium for demand $\lambda \vec{r}$ is linear for every support (see Lemma~\ref{lem:lambdapotentialsLineSegment}), and flows depend linearly on the potentials, the function $\vec{f}$ is piecewise-linear. Thus, it is enough to compute the potentials at the breakpoints $\lambdaMin, \lambdaMax$ explicitly. For any support, this can be clearly done in polynomial time. Since the functions $\startF$ and $\succF$ can be computed in polynomial time as well, \parametricatomicsplittable{} can be computed in polynomial space.
\end{appendixproof}

A game has unique equilbria for all demands if and only if, for every feasible support $\support$, we have $\sigma_{\support} > 0$. In this case, the function $\vec{f}$ computes \emph{all Nash Equilibria} of the game. If, in addition, the game is non-$b$-degenerate, then every support corresponds to a breakpoint of the piecewise affine function $\vec{f}$, yielding an output polynomial algorithm. Using Theorem~\ref{thm:symmetric} this implies in particular an output-polynomial algorithm for the parametric Nash equilibrium problem for player-independent costs.

\begin{theorem} \label{thm:parametric:sym:complexity}
\parametricatomicsplittable{} can be solved in output-polynomial time for non-$b$-degenerate games with $\sigma_{\support} > 0$ for all feasible supports $\support$. In particular, the runtime is in $\mathcal{O}((kn)^{2.4} + \tau (kn)^2 )$, where $\tau$ is the number of breakpoints of the piecewise affine function $\vec f$ returned by the algorithm. 
\end{theorem}
\begin{appendixproof}{Theorem~\ref{thm:parametric:sym:complexity}}
To obtain the first support $\support_0$, we need to solve $k$ shortest-path-problems and then check $km$ many inequalities for equality. Then we need to setup the Laplacian matrix $\vec{L}_{\support}$ and compute its (generalized) inverse. Using a fast matrix multiplication algorithm, the latter can be done in $\mathcal{O}( (kn)^{2.4} )$ time, e.g., with the Coppersmith-Winograd algorithm \cite{coppersmith1987matrix}.

Given the (generalized) inverse of the Laplacian $\vec{L}^*_{\support}$, potentials and flows can be computed in $\mathcal{O}((nk)^2)$ time as it only requires the multiplication of $nk \times nk$ matrices with vectors of dimension $nk$.  We note that by definition the matrix $\tilde{\vec{C}}_{\support}$ has at most $k$ non-zero entries in every row. The vectors $\vec{w}_{\support, e, i}, \normalVec$, and $\lightNormalVec$ have also only $\mathcal{O}(k)$ many non-zero elements. This implies that every value necessary for the computation of $\lambda^{\min}_{\support}$ and $\lambda^{\max}_{\support}$ can be obtained in $\mathcal{O}(k)$ time. Thus the computation of $\lambda^{\min}_{\support}$ and $\lambda^{\max}_{\support}$ itself needs $\mathcal{O}(km \cdot k) = \mathcal{O}((kn)^2)$ time. The bottleneck is thus the computation of the (generalized) inverse of the Laplacian matrix. However, the inverse does not need to be computed from scratch, but can be obtained with the update formula from Theorem~\ref{thm:neighboringLaplacians}. Since this formula also depends only on the sparse vector $\normalVec$ and $\lightNormalVec$ with $\mathcal{O}(k)$-many non-zero entries, this update step can also be computed in $\mathcal{O}((nk)^2)$-time.
\end{appendixproof}

\section*{Appendix}
\appendix
\section{Degeneracy} \label{sec:degeneracy}

\subsection{$a$-Degeneracy} \label{sec:aDegeneracy}

In this section, we show how the functions $\startF, \predF$, and $\succF$ can be modified to still work as intended, even if $a$-degenerate supports exist. First, recall that a support $\mathcal{S}$ is $a$-degenerate, if $\dim(\ker(\vec{L}_{\mathcal{S}})) > \dim(\mathcal{N}) = k$, i.e., if $\hat{\vec{L}}_{\mathcal{S}}$ is singular.
This means that equation $\vec{L}_{\mathcal{S}} \vec{\pi} = \vec{y}$ has no longer a unique solution in $\potentialSpace$. In fact, if there is an $a$-degenerate support $\mathcal{S}$ that is feasible, i.e., $\vec{L}_{\mathcal{S}} \vec{\pi} = \vec{y}$ has a solution, then there must be infinitely many solutions to this equation. Lemma~\ref{lem:matricesC} shows that the induced flow function $\inducedflow$ is injective. Thus, if there exists an $a$-degenerate support, then there must be also infinitely many Nash equilibrium flows for the same excess vector $\vec{y}$ and, hence, infinitely many Nash equilibria for the same demand. See Example~\ref{ex:8playerInfinite} below for a concrete example where such support exists.

In general, for an $a$-degenerate support, all potentials induce the same excess vector $\vec{y}$. Thus, $\lambda^{\min}_{\mathcal{S}} = \lambda^{\max}_{\mathcal{S}}$ (we prove this formally in Lemma~\ref{lem:neighboringDegenerateSupport}) and the potentials $\vec{\pi}^{\min}_{\mathcal{S}}$ and $\vec{\pi}^{\max}_{\mathcal{S}}$ are not unique. Therefore, the functions $\predF$ and $\succF$ are not well-defined in the case of an $a$-degenerate support. We now proceed by adapting these functions for the case of $a$-degenerate games.
As before, we still (only) use 
\begin{equation*}
\begin{split}
\myState := \{ \mathcal{S} : \mathcal{S} &\text{ is a feasible,}\text{ non-}a\text{-degenerate,} 
\text{ shortest-path-support}  \}
\end{split}
\end{equation*}
as the set of states, i.e., even in an $a$-degenerate game, we only consider non-$a$-degenerate supports as states. We observe, that $\startF$ as defined before yields the support $\mathcal{S}^0$ which is non-a-degenerate with the same proof as in Lemma~\ref{lem:startPredSucc}. Hence, we only need to define new successor and predecessor functions.

The first observation that we are going to make is that any (continuative) neighboring support $\mathcal{S}'$ of some non-$a$-degenerate support $\mathcal{S}$ is not ``too degenerate''. Formally, we say a support $\mathcal{S}$ is \emph{weakly-$a$-degenerate} if $\rank(\vec{L}_{\mathcal{S}}) = k (n-1) - 1$ (or equivalently, if $\dim(\ker(\vec{L}_{\support})) = k + 1$). With the rank-$1$-update formula from Theorem~\ref{thm:neighboringLaplacians} we directly obtain
\begin{lemma}
Let $\mathcal{S}$ be a non-$a$-degenerate support. Then every neighboring support $\mathcal{S}'$ of $\mathcal{S}$ is non-$a$-degenerate or weakly-$a$-degenerate.
\end{lemma}
For every weakly-$a$-degenerate support $\mathcal{S}$ we know by definition that $\dim(\ker(\hat{\vec{L}}_{\mathcal{S}})) = 1$, i.e., there is a direction vector $\nullspaceD$ such that 
$\ker(\hat{\vec{L}}_{\mathcal{S}}) = \vspan(\nullspaceD)$. (This direction is of course not unique, but it is unique up to a scalar multiplication.)
We refer to this direction as the \emph{nullspace direction} of the weakly-$a$-degenerate support $\mathcal{S}$. As it will turn out, if, for example, the successor $\mathcal{S}' := \succF(\mathcal{S})$ of some non-$a$-degenerate support $\mathcal{S}$ is $a$-degenerate, then we can use this direction to find another, non-$a$-degenerate support $\mathcal{S}''$ by following the nullspace direction $\nullspaceD[\mathcal{S}']$. We can then use $\mathcal{S}''$ as successor instead.
In order to prove this formally, we first need the following technical observations.

\begin{lemma} \label{lem:neighboringDegenerateSupport}
Let $\mathcal{S}$ be a non-$a$-degenerate support and $\mathcal{S}'$ be a $(e,i)$-neighboring, $a$-degenerate support.
\begin{enumerate}[(i)]
\item
	The nullspace direction $\nullspaceD[\altSupport] \in \vspan ( \vec{L}_{\support}^+ \lightNormalVec)$.
\item
	If $\normalVec^{\top} \Delta \vec{\pi}_{\support} \neq 0$ and $\altSupport$ is feasible, then
	the equation $\vec{L}_{\altSupport} \vec{\pi} = \lambda \Delta \vec{y} + \vec{d}_{\altSupport}$ has a solution if and only if $\lambda = \lambda^*_{\altSupport}$ for some fixed $\lambda^*_{\altSupport} \in [0,1]$.
\item
	Some potential vector $\vec{\pi} \in \potentialSpace$ satisfies $\vec{L}_{\altSupport} \vec{\pi} = \lambda^*_{\altSupport} \Delta \vec{y} + \vec{d}_{\altSupport}$ if and only if there is some $\nullspaceParameter \in \R$ such that $\vec{\pi} = \vec{L}^+_{\altSupport} \big( \lambda^*_{\altSupport} \Delta \vec{y} + \vec{d}_{\altSupport}) + \nullspaceParameter \nullspaceD[\altSupport]$ for some generalized inverse $\vec{L}^+_{\altSupport}$ with $\vec{L}^+_{\altSupport} \big( \lambda^* \Delta \vec{y} + \vec{d}_{\altSupport} \big) \in \potentialSpace$.
\end{enumerate}
\end{lemma}
\begin{proof}
Using a corollary from the Sherman-Morrison-Woodbury formula for generalized inverses \cite[Cor.~18.5.2]{harville1997matrix} we get
$
\rank(\vec{L}_{\mathcal{S}'}) = \rank(\vec{L}_{\mathcal{S}}) + \rank(1 + \normalVec^{\top} \vec{L}_{\support}^+ \lightNormalVec) - 1
.
$
By assumption, $\rank(\vec{L}_{\mathcal{S}'}) - \rank(\vec{L}_{\mathcal{S}}) = -1$ and, hence, we obtain
$
\normalVec^{\top} \vec{L}_{\support}^+ \lightNormalVec = -1
.
$
Together with Theorem~\ref{thm:neighboringLaplacians} this yields
\begin{equation} \label{eq:lem:neighboringDegenerateSupport:subclaim1}
\vec{L}_{\altSupport} \vec{L}^+_{\support} \lightNormalVec
= \vec{L}_{\support} \vec{L}^+_{\support} \lightNormalVec + \lightNormalVec \normalVec^{\top} \vec{L}^+_{\support} \lightNormalVec 
= \lightNormalVec - \lightNormalVec = \vec{0}
.
\end{equation} 
Thus, $\vspan (\vec{L}_{\mathcal{S}}^+ \vec{w}'_{\mathcal{S}',e,i}) \subset \ker(\vec{L}_{\mathcal{S}'}) \cap \potentialSpace$. Since $\dim(\ker(\vec{L}_{\mathcal{S}'}) \cap \potentialSpace) = 1$ we therefore obtain that
$
\ker( \vec{L}_{\altSupport} ) \cap \potentialSpace = \vspan ( \vec{L}^+_{\support} \lightNormalVec )
$. This implies, that the nullspace direction $\nullspaceD[\altSupport]$ is a multiple of the vector $\vec{L}^+_{\support} \lightNormalVec$ and, hence, proves \emph{(i)}.

Assume that  $\normalVec^{\top} \Delta \vec{\pi}_{\support} \neq 0$ and  $\altSupport$ is feasible. Then, by feasibility, there is a $\lambda^*_{\altSupport}$ such that $\vec{L}_{\altSupport} \vec{\pi} = \lambda^*_{\altSupport} \Delta \vec{y} + \vec{d}_{\altSupport}$ has the solution $\vec{\pi} = \vec{L}^+_{\altSupport} \big(\lambda^*_{\altSupport} \Delta \vec{y} + \vec{d}_{\altSupport})$. Now assume that $\lambda$ is chosen such that $\tilde{\vec{\pi}} = \vec{L}^+_{\altSupport} \big(\lambda \Delta \vec{y} + \vec{d}_{\altSupport})$ is a solution of $\vec{L}_{\altSupport} \vec{\pi} = \lambda \Delta \vec{y} + \vec{d}_{\altSupport}$, then we obtain
\begin{align*}
(\lambda^*_{\altSupport} - \lambda) \normalVec^{\top} \Delta \vec{\pi}_{\support} 
= \normalVec^{\top} \vec{L}^+_{\support} (\lambda^*_{\altSupport} - \lambda) \Delta \vec{y}
= \normalVec^{\top} \vec{L}^+_{\support} \vec{L}_{\altSupport} (\vec{\pi} - \tilde{\vec{\pi}}) 
\stackrel{\mathclap{\eqref{eq:lem:neighboringDegenerateSupport:subclaim1}}}{=} 0
\end{align*}
and, hence, $\lambda = \lambda^*_{\altSupport}$ proving \emph{(ii)}.

Since $\vec{L}_{\altSupport} \nullspaceD[\altSupport] = \vec{0}$, every vector $\vec{\pi} = \vec{L}^+_{\altSupport} \big(\lambda^*_{\altSupport} \Delta \vec{y} + \vec{d}_{\altSupport}) + \xi \nullspaceD[\altSupport]$ with $\xi \in \R$ is a solution of $\vec{L}_{\altSupport} \vec{\pi} = \lambda \Delta \vec{y} + \vec{d}_{\altSupport}$ and, by \emph{(ii)}, these are the only solution. Thus, \emph{(iii)} follows.
\end{proof}
Note, that the statements of Lemma~\ref{lem:neighboringDegenerateSupport} hold for any (fixed) choice of the generalized inverses of $\vec{L}_{\mathcal{S}}$ and $\vec{L}_{\mathcal{S}'}$. In the non-$a$-degenerate case, we use the fact that there is a unique generalized inverse $\vec{L}^{*}_{\mathcal{S}}$ that maps into $\potentialSpace$. For $a$-degenerate regions, there in no longer a unique such choice of the generalized inverse. For the remainder of this section we assume we have chosen some fixed generalized inverse $\vec{L}^{+}_{\altSupport}$ that maps into $\potentialSpace$.

Assume that the continuative $(e,i)$-neighbor $\altSupport$ of some non-$a$-degenerate support $\support$ is $a$-degenerate. Then, we know that $\altSupport$ is feasible (since $\support$ and $\altSupport$ share a boundary potential), thus we can use Lemma~\ref{lem:neighboringDegenerateSupport} to describe the polytope of $\lambda$-potentials $\potentialRegion[\altSupport]$ of $\altSupport$ as 
\begin{align*}
\potentialRegion[\altSupport] &=
\big\{
\vec{\pi} \in \potentialSpace
\; \big\vert \;
\vec{L}_{\altSupport} \vec{\pi} - \vec{d}_{\altSupport}  = \lambda^*_{\altSupport} \Delta \vec{y}
, 
\vec{W}_{\altSupport}  \big( \vec{G}^{\top} \! \vec{\pi} - \vec{b} \big) \geq \vec{0}
\big\} \\
&=
\big\{
\vec{L}^+_{\altSupport}  ( \lambda^*_{\altSupport} \Delta \vec{y} + \vec{d}_{\altSupport}  ) 
 + \nullspaceParameter \nullspaceD[\altSupport]
\; \big\vert \; \exists \nullspaceParameter \in \R : 
\vec{W}_{\altSupport}  \big( \vec{G}^{\top} \!
 \big(  
 \vec{L}^+_{\altSupport}  ( \lambda^*_{\altSupport} \Delta \vec{y} + \vec{d}_{\altSupport}  ) 
 + \nullspaceParameter \nullspaceD[\altSupport]
 \big) 
- \vec{b} \big) \geq \vec{0}
\big\}
.
\end{align*}
We see that the polytope $\potentialRegion[\altSupport]$ of an $a$-degenerate region is parametrized by $\xi \in \R$ along the nullspace direction $\nullspaceD[\altSupport]$, in contrast to a non-$a$-degenerate region $\support$, where $\potentialRegion$ is parametrized by $\lambda$ and the potential direction $\Delta \vec{\pi}_{\support}$.
%
Similar to Lemma~\ref{lem:lambdapotentialsLineSegment}, we can rewrite the polytope of $\lambda$-potentials as
\[
\potentialRegion[\altSupport] =
\big\{
\vec{L}^+_{\altSupport} ( \lambda^*_{\altSupport} \Delta \vec{y} + \vec{d}_{\altSupport} )  + \nullspaceParameter \nullspaceD[\altSupport]
\; \vert \;
\nullspaceMin[\altSupport] \leq \nullspaceParameter \leq \nullspaceMax[\altSupport]
\big\}
.
\]
We note that $\nullspaceMin[\altSupport], \nullspaceMax[\altSupport]$ are in fact finite values by the following argument. The induced flow of the nullspace direction $\inducedflow (\nullspaceD[\altSupport])$ is a flow circulation that satisfies the potential equations $\smash{\Delta \pi^{N,i}_w - \Delta \pi^{N,i}_v = a_e^i \Delta \bar{x}_e}$. Hence, there must be at least one edge with positive total induced flow and one edge with negative total induced flow implying that adding or subtracting $\inducedflow (\nullspaceD[\altSupport])$ violates at some point the non-negativity of the flow for some edge.
We therefore know that a non-$a$-degenerate support $\altSupport$ has well-defined boundary potentials $\bpMin[\altSupport], \bpMax[\altSupport]$. As for non-$a$-degenerate supports we can find continuative neighbors by looking at the $(e,i)$-neighbors that are shortest-path-supports for all pairs $(e,i)$ that induce $\nullspaceMin[\altSupport]$ and $\nullspaceMax[\altSupport]$.

For every boundary potential of the degenerate support $\support$, we can find a continuative neighbor. This neighbor is also unique, as long as the $a$-degenerate support is not also $b$-degenerate---we cover this case of $a$-$b$-degeneracy in the next section.
The following lemma proves the existence, uniqueness, and the non-$a$-degeneracy of such neighboring supports.

\begin{lemma} \label{lem:skippingWeaklyDegenerateSupports}
Let $\mathcal{S}$ be a non-$a$-degenerate support. Further, let $\mathcal{S}'$ be a continuative $(e,i)$-neighbor of $\support$ that is $a$-degenerate. Then, there is a uniquely defined support $\support^{*}_{\support, \altSupport} \neq \support$ that is a neighbor of $\altSupport$ and non-$a$-degenerate.
\end{lemma}
\begin{proof}
Since $\nullspaceMin[\altSupport], \nullspaceMax[\altSupport]$ are finite, there must be such neighboring supports, and as long as we assume non-$a$-$b$-degeneracy, there must also be a unqiue support. It remains to be shown that this support is also non-$a$-degenerate.
To this end, let $e' \in E$ and $i' \in \range{k}$ be such that $\support^{*}_{\support, \altSupport} := N(\altSupport, e', i')$ is said unique support.
By definition, $(e',i')$ is the minimizer in the computation of $\nullspaceMax$ or the maximizer in the computation of $\nullspaceMin$. This implies that $\normalVec[\altSupport, e',i']^{\top} \nullspaceD \neq 0$ yielding $\lightNormalVec[\mathcal{S}'',e',i'] \normalVec[\altSupport, e',i']^{\top} \nullspaceD \neq \vec{0}$. Since $\rank(\lightNormalVec[\mathcal{S}'',e',i'] \normalVec[\altSupport, e',i']^{\top}) = 1$ we  know that $\ker(\lightNormalVec[\mathcal{S}'',e',i'] \normalVec[\altSupport, e',i']^{\top}) = \{ \vec{v} \in \R^{nk} \,|\, \vec{v} \notin \vspan(\nullspaceD) \}$. Finally, let $\vec{v} \in \ker(\vec{L}_{\mathcal{S}''})$. Then we have two cases.
\begin{enumerate}[a)]
\item
	$\vec{v} \in \vspan(\nullspaceD)$. Then $\vec{0} = \vec{L}_{\mathcal{S}''} \vec{v} = \vec{L}_{\mathcal{S}'} \vec{v} + \lightNormalVec[\mathcal{S}'',e',i'] \normalVec[\altSupport, e',i']^{\top} \vec{v} = \lightNormalVec[\mathcal{S}'',e',i'] \normalVec[\altSupport, e',i']^{\top} \vec{v} \neq \vec{0}$, which is a contradiction.
\item
	$\vec{v} \notin \vspan(\nullspaceD)$. Then $\vec{0} = \vec{L}_{\mathcal{S}''} \vec{v} = \vec{L}_{\mathcal{S}'} \vec{v} + \lightNormalVec[\mathcal{S}'',e',i'] \normalVec[\altSupport, e',i']^{\top} \vec{v} = \vec{L}_{\mathcal{S}'} \vec{v}$ which is equivalent to $\vec{v} \in \mathcal{N}$ ($\vec{v} \in \vspan(\nullspaceD)$ is excluded by assumption).
\end{enumerate}
Hence, we have $\ker(\vec{L}_{\mathcal{S}''}) = \mathcal{N}$, i.e., $\mathcal{S}''$ is non-degenerate.
\end{proof}

We define a new predecessor function
\[
\predF^* (\support) :=
\begin{cases}
\predF (\support)
	&\text{if } \predF (\support) \in \myState \\
\support^{*}_{\support, \predF(\support)}
	&\text{if } \predF(\support) \notin \myState
\end{cases}
\]
and a new successor $\succF^{*}$ function analogously that are based on the the basic functions $\predF$ and $\succF$ except for the cases where these function would map to a $a$-degenerate support $\altSupport \notin \myState$. Finally, the next theorem proves, that these functions are still compliant with each other.

\begin{theorem} \label{thm:aDegeneracyWellDefined}
The functions $\predF^*$ and $\succF^*$ are well-defined and computable in polynomial time. Further, 
\begin{enumerate}[(i)]
\item
	$\predF^* (\mathcal{S}) \neq \emptyset \; \Rightarrow \; \succF^* (\predF^* (\mathcal{S})) = \mathcal{S}$
\item
	$\succF^* (\mathcal{S}) \neq \emptyset \; \Rightarrow \; \predF^* (\succF^* (\mathcal{S})) = \mathcal{S}$.
\end{enumerate}
\end{theorem}
\begin{proof}
The functions are clearly well-defined. Since we can find the support $\mathcal{S}^*_{\mathcal{S}, \mathcal{S}'}$ via the computation of $\nullspaceMin[\altSupport], \nullspaceMax[\altSupport]$ requiring the computation of a maximum or minimum over $km$-many values, the functions are also computable in polynomial time. 
Whenever $\predF(\mathcal{S}) \in \myState$ or $\succF(\mathcal{S}) \in \myState$, respectively, statements \emph{(i)} and \emph{(ii)} follow immediately from Lemma~\ref{lem:startPredSucc}. 
Thus, we only need to show \emph{(i)} and \emph{(ii)} for the case when $\predF(\mathcal{S}) \in \myState$ or $\succF(\mathcal{S}) \in \myState$ are $a$-degenerate supports.
In order to achieve this, we need an analogue property to Theorem~\ref{thm:neighboringLaplacians}\emph{(\ref{thm:neighboringLaplacians:directions})}. This is difficult since the orientation $\sigma_{\altSupport}$ of any $a$-degenerate support $\altSupport$ is zero. We overcome this difficulty by adding a small perturbation to the block Laplacian matrices restoring the regularity of $\hat{\vec{L}}_{\mathcal{S}'}$ for the weakly degenerate support $\mathcal{S}'$ enabling us to apply Theorem~\ref{thm:neighboringLaplacians} also in the degenerate setting.

Formally, for some $\delta > 0$, let $\vec{D}^i \in \R^{n \times n}$ be a matrix defined as follows: Take the $(n-1) \times (n-1)$-identity matrix multiplied by $\delta$ and insert a row and column at the position of the index of the vertex $s_i$ such that the resulting matrix has zero row- and column-sum. Thus,
\[
\vec{D}^i \! = \!
\begin{pmatrix}
\delta & 0 & \cdots & 0 & - \delta & 0 & \cdots & 0 \\
0 & \delta & \cdots & 0 & - \delta & 0 & \cdots & 0 \\
\vdots &  & \ddots &  & \vdots &  & & \vdots \\
0 & 0 & \cdots & \delta & - \delta & 0 & \cdots & 0 \\
-\delta & -\delta & \cdots & -\delta & (n-1) \delta & -\delta & \cdots & -\delta \\
0 & 0 & \cdots & 0 & - \delta & \delta & \cdots & 0 \\
\vdots &  &  &  & \vdots & \ddots & & \vdots \\
0 & 0 & \cdots & 0 & - \delta & 0 & \cdots & \delta \\
\end{pmatrix}
\!.
\]
Then let $\vec{D}$ be the block diagonal matrix containing the matrices $\vec{D}^i, i=1, \dotsc, k$ as block diagonal elements. Further, let $\hat{\vec{D}}$ be the matrix obtained from $\vec{D}$ by deleting the rows and columns belonging to the source vertices of the respective players. Then, in partiuclar, $\hat{\vec{D}}$ is a $k(n-1) \times k(n-1)$-diagonal matrix with $\delta$ on the diagonal. Hence, for every matrix $\vec{A} \in \R^{k(n-1) \times k(n-1)}$ the matrix $\vec{A} + \hat{\vec{D}}$ is non-singular for almost all $\delta > 0$.

For any block Laplacian matrix $\vec{L}$ define $\vec{L}^{\delta} := \vec{L} + \vec{D}$, $\hat{\vec{L}}^{\delta} := \hat{\vec{L}}^{\delta} + \hat{\vec{D}}$, and $\vec{L}^{\delta, *}$ as the matrix obtained from $\big( \hat{\vec{L}}^{\delta} \big)^{-1}$ by adding zero rows and columns for the deleted rows and columns in $\hat{\vec{L}}^{\delta}$. (The latter is well-defined for almost all $\delta>0$.) Then it can be shown that $\vec{L}^{\delta, *}$ is a generalized inverse of $\vec{L}^{\delta}$. 
Further, we define $\Delta \vec{\pi}^{\delta} := \sgn(\det(\hat{\vec{L}}^{\delta})) \vec{L}^{\delta,*} \Delta \vec{y}$.
By definition, $\vec{L}^{\delta} \xrightarrow{\delta \to 0} \vec{L}$ and $\hat{\vec{L}}
^{\delta} \xrightarrow{\delta \to 0} \hat{\vec{L}}$ element-wise. By the continuity of the inverse this also implies that $\vec{L}^{\delta, *} \xrightarrow{\delta \to 0} \vec{L}^*$ whenever $\hat{\vec{L}}$ is non-singlar. Thus, in these cases, $\Delta \vec{\pi}^{\delta} \xrightarrow{\delta \to 0} \Delta \vec{\pi}$.

Let $\support$ be some feasible support with the continuative $(e,i)$-neighbor $\altSupport := N(\support, e,i)$. Suppose that $\altSupport$ is $a$-degenerate and let $\support'' := \mathcal{S}^*_{\support, \altSupport} = N(\altSupport, e', i')$ for some $e' \in E$ and $i' \in \range{k}$.
Then,
\[
\vec{L}^{\delta}_{\mathcal{S}'} 
= \vec{L}_{\mathcal{S}'} + \vec{D}
= \vec{L}_{\mathcal{S}} + \lightNormalVec \normalVec^{\top} + \vec{D}
= \vec{L}^{\delta}_{\mathcal{S}} +  \lightNormalVec \normalVec^{\top}
\]
and, similarly, $\vec{L}^{\delta}_{\mathcal{S}''} = \vec{L}^{\delta}_{\mathcal{S}'} +   \lightNormalVec \normalVec^{\top}$.
Since $\det ( \hat{\vec{L}}^{\delta} )$ is a continuous, non-zero polynomial in the perturbation variable $\delta$ that is zero for $\delta = 0$. Thus, there is a $\delta^* > 0$, such that for all $0 < \delta < \delta^*$
\begin{compactenum}[a)]
\item
	the matrix $\hat{\vec{L}}^{\delta}_{\mathcal{S}'}$ is non-singular and the sign of the determinant $\sigma^*_{\mathcal{S}'} := \sgn(\hat{\vec{L}}^{\delta}_{\mathcal{S}'})$ is constant.
\item
	the signs $\sgn(\det(\hat{\vec{L}}^{\delta}_{\mathcal{S}})) = \sigma_{\mathcal{S}}$ and
	$\sgn(\det(\hat{\vec{L}}^{\delta}_{\mathcal{S}''})) = \sigma_{\mathcal{S}''}$ are constant.
\end{compactenum}

With the same arguments as in the proof of Theorem~\ref{thm:neighboringLaplacians}, we obtain
\begin{equation} \label{eq:aDegeneracyWellDefined:subclaim1}
\det(\hat{\vec{L}}^{\delta}_{\mathcal{S}'}) = \det(\hat{\vec{L}}^{\delta}_{\mathcal{S}}) (1 +  \normalVec^{\top} \vec{L}^{\delta,*}_{\mathcal{S}} \lightNormalVec)
.
\end{equation}
 We define a new direction $\Delta \vec{\pi}^{\delta}_{\mathcal{S}'} := \vec{L}^{\delta,*}_{\mathcal{S}'} \Delta \vec{y}$ for the support $\mathcal{S}'$. Then, using the Sherman-Morrison formula as in Theorem~\ref{thm:neighboringLaplacians},
\begin{align*}
\det(\hat{\vec{L}}^{\delta}_{\mathcal{S}'})  \Delta \vec{\pi}^{\delta}_{\mathcal{S}'} 
&= \det(\hat{\vec{L}}^{\delta}_{\mathcal{S}'})  \vec{L}^{\delta,*}_{\mathcal{S}'} \Delta  \vec{y} \\
&= \det(\hat{\vec{L}}^{\delta}_{\mathcal{S}'})  \bigg(
	\vec{L}^{*,\delta}_{\mathcal{S}} 
	- \frac{1}{1 \!+\! \normalVec^{\top} \vec{L}^{\delta,*}_{\mathcal{S}} \lightNormalVec} \vec{L}^{*,\delta}_{\mathcal{S}} \lightNormalVec \normalVec \vec{L}^{\delta,*}_{\mathcal{S}} \bigg) \Delta \vec{y} \\
&\stackrel{\mathclap{\eqref{eq:aDegeneracyWellDefined:subclaim1}}}{=}
	\det(\hat{\vec{L}}^{\delta}_{\mathcal{S}}) 
	\underbrace{
	(1 + \normalVec^{\top} \vec{L}^{\delta,*}_{\mathcal{S}}	 \lightNormalVec
	}_{
	\xrightarrow{\delta \to 0} 0 \text{ by Lemma~\ref{lem:neighboringDegenerateSupport}}
	}
	\vec{L}^{*,\delta}_{\mathcal{S}} \Delta \vec{y} 
	- \vec{L}^{*,\delta}_{\mathcal{S}} \normalVec^{\top} \lightNormalVec
	\underbrace{
	\det(\hat{\vec{L}}^{\delta}_{\mathcal{S}}) \vec{L}^{\delta,*}_{\mathcal{S}} \Delta \vec{y}
	}_{
	\xrightarrow{\delta \to 0} \det ( \hat{\vec{L}}_{\support} ) \Delta \vec{\pi}_{\mathcal{S}}
	}
	\\
&\xrightarrow{\delta \to 0}
	(- \det ( \hat{\vec{L}}_{\support} ) \normalVec^{\top} \Delta \vec{\pi}_{\mathcal{S}} ) \vec{L}^{*}_{\mathcal{S}} \lightNormalVec
	.
\end{align*}
Thus, by Lemma~\ref{lem:neighboringDegenerateSupport}\emph{(i)}, the direction $\Delta \vec{\pi}^{\delta}_{\mathcal{S}'}$ converges to (a multiple of) the nullspace direction $\nullspaceD[\mathcal{S}']$, i.e., there is $\alpha \in \R$ such that 
$
\Delta \vec{\pi}^{\delta}_{\mathcal{S}'} \xrightarrow{\delta \to 0} \alpha \nullspaceD[\mathcal{S}']
. $
Thus, in particular, $\normalVec^{\top} \Delta \vec{\pi}^{\delta}_{\mathcal{S}'} \xrightarrow{\delta \to 0} \alpha \, \normalVec^{\top} \nullspaceD[\mathcal{S}']$ and, hence,
\begin{equation} \label{eq:lem:aDegeneracyWellDefined:subclaim2}
\alpha \, \normalVec^{\top}\nullspaceD[\mathcal{S}'] 
= (- \det(\hat{\vec{L}}_{\mathcal{S}}) \normalVec^{\top} \Delta \vec{\pi}_{\mathcal{S}} ) \normalVec \vec{L}^{*}_{\mathcal{S}} \lightNormalVec
= \det(\hat{\vec{L}}_{\mathcal{S}}) \normalVec^{\top} \Delta \vec{\pi}_{\mathcal{S}}
.
\end{equation}
Similarly, we can show that $\Delta \vec{\pi}^{\delta}_{\mathcal{S}'} \xrightarrow{\delta \to 0} (- \det(\hat{\vec{L}}_{\mathcal{S}''}) \normalVec[\altSupport, e',i']^{\top} \Delta \vec{\pi}_{\mathcal{S}''} ) \vec{L}^{*}_{\mathcal{S}''} \lightNormalVec[\mathcal{S}'',e',i']$ and thus get similarly that
\begin{equation} \label{eq:lem:aDegeneracyWellDefined:subclaim3}
\alpha \, \normalVec[\altSupport, e',i']^{\top} \nullspaceD[\mathcal{S}'] = 
\det(\hat{\vec{L}}_{\mathcal{S}''}) \lightNormalVec[\mathcal{S}'',e',i'] \Delta \vec{\pi}_{\mathcal{S}''}
.
\end{equation}

By the definition of $\mathcal{S}''$, we know that the direction $\nullspaceD[\mathcal{S}']$ is directed away from the the hyperplane that separates $\mathcal{S}$ and $\mathcal{S}'$ and is directed towards the hyperplane that separates $\mathcal{S}'$ and $\mathcal{S}''$, i.e., that 
$\sgn(\normalVec[\mathcal{S}',e,i]^{\top} \nullspaceD[\mathcal{S}'] ) = - \sgn(  \normalVec[\mathcal{S}',e',i']^{\top} \nullspaceD[\mathcal{S}'] )$.
Finally, we obtain
\begin{align*}
\sgn( \normalVec^{\top} \Delta \vec{\pi}_{\mathcal{S}} )
&\stackrel{\mathclap{\eqref{eq:lem:aDegeneracyWellDefined:subclaim2}}}{=}
	\sgn( \alpha \,  \normalVec^{\top} \nullspaceD[\mathcal{S}'] ) 
= - \sgn( \alpha \, \normalVec[\altSupport, e, i]^{\top}\nullspaceD[\mathcal{S}'] ) \\
&= \sgn( \alpha \,  \normalVec[\altSupport, e' ,i']^{\top}\nullspaceD[\mathcal{S}'] ) 
\stackrel{\mathclap{\eqref{eq:lem:aDegeneracyWellDefined:subclaim3}}}{=}
	\sgn( \normalVec[\mathcal{S}',e',i']^{\top} \Delta \vec{\pi}_{\mathcal{S}''} ) 
= - \sgn( \normalVec[\mathcal{S}'',e',i']^{\top} \Delta \vec{\pi}_{\mathcal{S}''} )
.
\end{align*}
This sign condition implies that, the boundary potential $\vec{\pi} \in \supportBoundary \cap \supportBoundary[\mathcal{S}']$ that lies on the boundary between $\mathcal{S}$ and $\mathcal{S}'$ is $\vec{\pi}^{\min}_{\mathcal{S}}$ if and only if the potential $\vec{\pi}' \in \supportBoundary[\mathcal{S}'] \cap \supportBoundary[\mathcal{S}'']$  that lies on the boundary between $\mathcal{S}'$ and $\mathcal{S}''$ is $\vec{\pi}^{\max}_{\mathcal{S}''}$ and vice versa. Thus, we get that $\predF(\succF(\mathcal{S})) = \mathcal{S}$ and $\succF(\predF(\mathcal{S})) = \mathcal{S}$, respectively.
\end{proof}
\begin{example}\label{ex:8playerInfinite}
We consider a game with 8 players on the graph given in Figure~\ref{fig:infequilibria:graph}. Every player has two adjacent vertices as source and sink vertex and a demand of $r_i = 2$. We define the costs in a symmetric way, meaning that for every player the edge connecting the source and the sink vertex (e.g. $e_8$ for player~$1$) is equipped with the cost function $c_1(x) = 9 x + 3$ and the edges on the longer path from the source to the sink (e.g. $e_1, e_3, e_5$ for player~$1$) are equipped with the cost function $c_2 (x) = x + 6$. All other edges are assumed to have cost functions with high offsets and slopes such that they are never in the support of the respective players.
By definition, the game is player-symmetric, i.e., the cost functions look the same for every player. Yet, not all players have the same cost function on every edge.
Because of this player symmetry, there are overall three possible strategy profiles: \vspace{-1.5ex}
\begin{enumerate}[a)]
\item
	All players use the direct path between $s_i$ and $t_i$. \vspace{-1.5ex}
\item
	All players use both path between $s_i$ and $t_i$. \vspace{-1.5ex}
\item
	All players use the long path between $s_i$ and $t_i$. \vspace{-1.5ex}
\end{enumerate}
Since all strategies and all cost functions are symmetric we give all potentials, flows and directions for the first player routing demand from $v_1$ to $v_4$. All other potential and flow values are given implicitly.

For small demands $\lambda \vec{r}$, $\lambda \leq \frac{1}{2}$, because of the relatively high offset of the long path, all players will route their demand on the direct path to their source. The potential for the zero flow $\vec{x}^0 = \vec{0}$ is the shortest path potential $\vec{\pi}^{0} = (0,6,12,3)^{\top}$. This yields the initial support $\mathcal{S}^0$ where the edges $e_1, e_3$ and $e_8$ are active for the player $1$. (Note that we consider the edges $e_1, e_3$ active although they are not used by the player. They are only part of the shortest path network.) We obtain the potential direction $\Delta \vec{\pi}^0 = (0, 2, 4, 36)^{\top}$ that induces the flow direction $\Delta \vec{x}^0 = (0,0,0,0,0,0,0,2)^{\top}$.
For $\lambda = \frac{1}{2}$, we obtain the potential $\vec{\pi}^1 = (0,7,14,21) = \vec{\pi}^0 + \frac{1}{2} \Delta \vec{\pi}^0$ inducing the flow $\vec{x}^1 = (0,0,0,0,0,0,0,1)^{\top}$, i.e., every player routes exactly one flow unit on the direct path from source to sink.

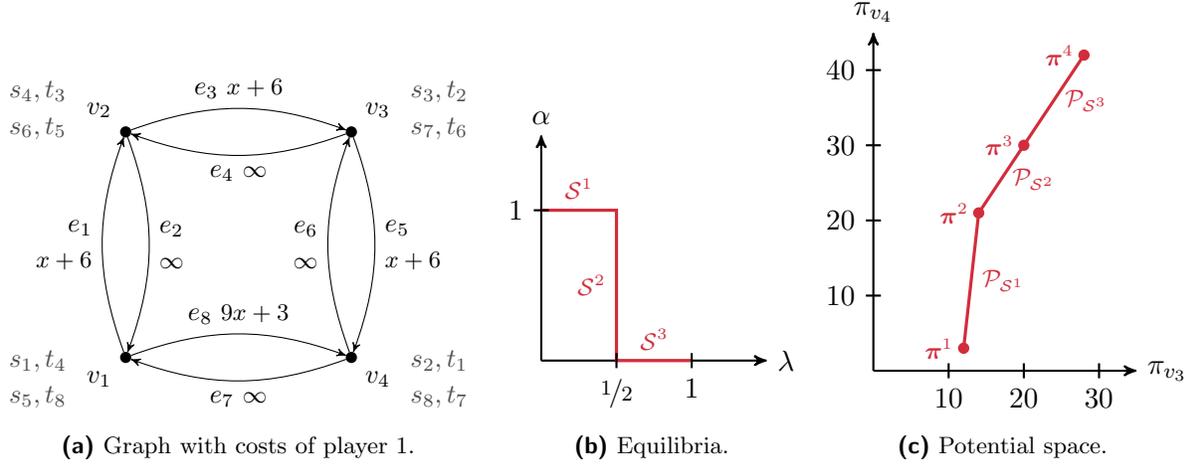
\begin{figure}
\begin{subfigure}[b]{.4\linewidth}
\begin{center}
\begin{tikzpicture}
\node[solid] (v1) at (0,0) {} node[below left=0cm of v1] (lv1) {\small $v_1$};
\node[solid] (v2) at (0,3) {} node[above left=0cm of v2] (lv2) {\small $v_2$};
\node[solid] (v3) at (3,3) {} node[above right=0cm of v3] (lv3) {\small $v_3$};
\node[solid] (v4) at (3,0) {} node[below right=0cm of v4] (lv4) {\small $v_4$};

\node[left=0 cm of lv1, align=right] {\small \textcolor{black!70}{$s_1, t_4$} \\ \small \textcolor{black!70}{$s_5, t_8$}};
\node[left=0 cm of lv2, align=right] {\small \textcolor{black!70}{$s_4, t_3$} \\ \small \textcolor{black!70}{$s_6, t_5$}};
\node[right=0 cm of lv3, align=right] {\small \textcolor{black!70}{$s_3, t_2$} \\ \small \textcolor{black!70}{$s_7, t_6$}};
\node[right=0 cm of lv4, align=right] {\small \textcolor{black!70}{$s_2, t_1$} \\ \small \textcolor{black!70}{$s_8, t_7$}};

\draw[bend left=20,->]
	(v1) edge node[midway, left, align=right] {\small $e_1$ \\ \footnotesize \textcolor{black}{$x + 6$}} (v2)
	(v2) edge node[midway, right, align=left] {\small $e_2$ \\ \footnotesize \textcolor{black}{$\infty$}} (v1)
	(v2) edge node[midway, above] {\small $e_3$ \textcolor{black}{\footnotesize $x + 6$}} (v3)
	(v3) edge node[midway, below] {\small $e_4$ \textcolor{black}{\footnotesize $\infty$}} (v2)
	(v3) edge node[midway, right, align=left] {\small $e_5$ \\ \footnotesize \textcolor{black}{$x + 6$}} (v4)
	(v4) edge node[midway, left, align=right] {\small $e_6$ \\ \footnotesize \textcolor{black}{$\infty$}} (v3)
	(v4) edge node[midway, below] {\small $e_7$ \textcolor{black}{\footnotesize $\infty$}} (v1)
	(v1) edge node[midway, above] {\small $e_8$ \textcolor{black}{\footnotesize $9 x + 3$}} (v4);
\end{tikzpicture}
\caption{Graph with costs of player~1.}
\label{fig:infequilibria:graph}
\end{center}
\end{subfigure}
\begin{subfigure}[b]{.25\linewidth}
\begin{center}
\begin{tikzpicture}[scale=2]
\draw[thick, ->] (0,0) -- (1.5,0) node[anchor=west] {$\lambda$};
\draw[thick, ->] (0,0) -- (0,1.5) node[anchor=south] {$\alpha$};

\draw[very thick, playercolor1]
	(0,1) -- node[midway, above] {\footnotesize $\mathcal{S}^1$} (1/2,1)
		  -- node[midway, left] {\footnotesize $\mathcal{S}^2$} (1/2,0) 
		  -- node[midway, above] {\footnotesize $\mathcal{S}^3$} (1,0);

\draw[thick] (1/2,.05) -- (1/2,-.05) node[anchor=north] {$\nicefrac{1}{2}$};
\draw[thick] (1,.05) -- (1,-.05) node[anchor=north] {$1$};
\draw[thick] (.05,1) -- (-.05,1) node[anchor=east] {$1$};
\end{tikzpicture}
\caption{Equilibria.}
\label{fig:infequilibria:equilibria}
\end{center}
\end{subfigure}
\begin{subfigure}[b]{.3\linewidth}
\begin{center}
\begin{tikzpicture}[scale=0.1]
\draw[thick, ->] (0,0) -- (35,0) node[anchor=west] {$\pi_{v_3}$};
\draw[thick, ->] (0,0) -- (0,45) node[anchor=south] {$\pi_{v_4}$};

\draw[very thick, playercolor1]
	(12,3) -- node[midway, right] {\footnotesize $\potentialRegion[\mathcal{S}^1]$} (14,21)
		   -- node[midway, right] {\footnotesize $\potentialRegion[\mathcal{S}^2]$}(20,30) 
		   -- node[midway, right] {\footnotesize $\potentialRegion[\mathcal{S}^3]$} (28,42);
\fill[playercolor1]	(12,3) circle(0.75) node[anchor=east] {\footnotesize $\vec{\pi}^1$}
					(14,21) circle(0.75) node[anchor=east] {\footnotesize $\vec{\pi}^2$}
					(20,30) circle(0.75) node[anchor=east] {\footnotesize $\vec{\pi}^3$}
					(28,42) circle(0.75) node[anchor=east] {\footnotesize $\vec{\pi}^4$};

\draw[thick] (10,1) -- (10,-1) node[anchor=north] {$10$};
\draw[thick] (20,1) -- (20,-1) node[anchor=north] {$20$};
\draw[thick] (30,1) -- (30,-1) node[anchor=north] {$30$};
\draw[thick] (1,10) -- (-1,10) node[anchor=east] {$10$};
\draw[thick] (1,20) -- (-1,20) node[anchor=east] {$20$};
\draw[thick] (1,30) -- (-1,30) node[anchor=east] {$30$};
\draw[thick] (1,40) -- (-1,40) node[anchor=east] {$40$};
\end{tikzpicture}
\caption{Potential space.}
\label{fig:infequilibria:potentialSpace}
\end{center}
\end{subfigure}
\caption{\textbf{\textsf{(a)}} The graph of Example~\ref{ex:8playerInfinite} with the cost functions of player~$1$. \textbf{\textsf{(b)}} The strategies of the players described by the fraction $\alpha$ of the demand that is routed on the direct edge between source and sink, parametrized by the demand multiplier $\lambda$. \textbf{\textsf{(c)}}  The polytopes of $\lambda$-potentials $\potentialRegion$ of the three supports with the boundary potentials in a projection of the potential space $\potentialSpace$.}
\label{fig:infequilibria}
\end{figure}

Denote by $e_{l_i}$ the last edge of the long path for player~$i$ (e.g. $e_{l_1} = e_5$). In the potential $\vec{\pi}^1$, all edges $e_{l_i}$ become active at once. (We note that this implies that the game is actually $b$-degenerate, but we ignore this for the benefit of an easier example.)
This defines a new support $\mathcal{S}^1$ that is $a$-degenerate---the Laplacian matrix has rank $23$ and $\ker( \vec{L}_{\mathcal{S}^1} ) \cap \potentialSpace = \vspan \big( (0,1,2,3)^{\top} \big)$. Hence, the nullspace direction is $\nullspaceD[\mathcal{S}^1] = (0,1,2,3)^{\top}$ inducing the circulation $\Delta \vec{x}^1 =\smash{\big(\frac{1}{3},0,\frac{1}{3},0,\frac{1}{3},0,0,-\frac{1}{3}\big)^{\top}}$. 
The potential $\vec{\pi}^2 := \vec{\pi}^1 + 3 \, \nullspaceD[\mathcal{S}^1] = (0,10,20,30)^{\top}$ induces the flow $\vec{x}^2 = (1,0,1,0,1,0,0,0)^{\top}$. Note that this flow is still an equilibrium for the demand $\frac{1}{2} \vec{r}$. In fact, all flows 
\[
\vec{x} (\alpha) = \big( (1-\alpha), 0, (1-\alpha), 0, (1-\alpha), 0, 0, \alpha\big)^{\top}
, 0 \leq \alpha \leq 1
\]
are equilibria for the demand $\frac{1}{2} \vec{r}$, i.e., there are infinitely many equilibria when every player wants to route a demand of $1$.

Finally, denote by $e_{d_i}$ the direct edge connecting the source and sink of player $i$ (e.g. $e_{d_1} = e_8$). All these edges become inactive in the potential $\vec{\pi}^2$ which defines the next support $\mathcal{S}_2$ which is non-$a$-degenerate. (Again, we ignore the $b$-degeneracy at this point.) With this support, we obtain the direction $\Delta \vec{\pi}^2 = (0,8,16,24)^{\top}$ inducing the flow direction $\Delta \vec{x} = (2,0,2,0,2,0,0,0)^{\top}$. For $\lambda = 1$ we then get the potential $\vec{\pi}^3 := \vec{\pi}^2 + \frac{1}{2} \Delta \vec{\pi}^3 = (0,14,28,42)^{\top}$ inducing the (solution) flow $\vec{x}^3 = (2,0,2,0,2,0,0,0)^{\top}$.
\end{example}

\subsection{$b$-Degeneracy} \label{sec:bDegeneracy}
In this section, we develop a rule for handling $b$-degenerate games. A game is $b$-degenerate, if there is a $b$-degenerate support. That is, a support with $|\contNeighborsMin \cap \myState| > 1$ or $|\contNeighborsMax \cap \myState | > 1$, i.e., a support with a boundary potential that is feasible in more than one $(e,i)$-neighboring supports.
This is the case if and only if there is not a unique minimizer or maximizer in the computation of $\lambda^{\min}_{\support}$ or $\lambda^{\max}_{\support}$.
Concretely, this means that when increasing or decreasing the demand, multiple edge-player-pairs change their activity status at once, e.g., if two edges become active for one player or one edge becomes inactive for two players and so on.
Roughly speaking, $b$-degeneracy occurs if the offsets of the cost functions $\vec{b}$ are such that the hyperplanes induced by the inequalities in the polytopes $\potentialRegion$ intersect exactly in some boundary potential of some support. Intuitively, a small random perturbation of the offsets $\vec{b}$  would translate the hyperplanes in a way that, almost surely, they do not intersect in boundary potentials anymore. Further, a small perturbation should not change the equilibrium to much.
In fact, we can show that there is an explicit, small perturbation $\vec{\epsilon} \in \R^{mk}$ that can be added to the vector $\vec{b}$ such that the game with these perturbed offsets is non-$b$-degenerate. Further, we can show that all feasible supports in this game are also feasible for the original, unperturbed game and, hence, it will be enough to compute the states $\mathcal{S}$ of the perturbed game. Additionally, we observe below that it is not necessary to carry out the perturbation explicitly. Rather,  we can use a lexicographic criterion.
We initially assume that every game is non-$a$-degenerate and discuss the interplay of $a$-degeneracy and $b$-degeneracy later.

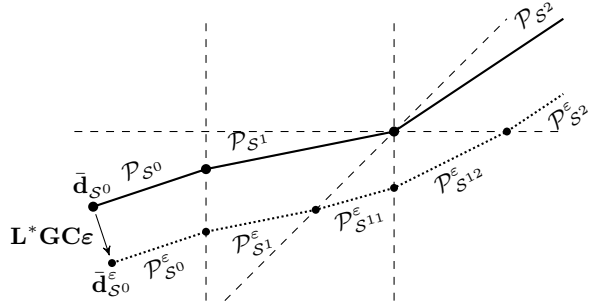
\begin{figure}
\begin{center}
\footnotesize
\begin{tikzpicture}
\useasboundingbox (-1.25,-0.5) rectangle (6.25,3.25);
\draw[thick, fill=black]
	(-0.5,0.5) circle(1.5pt)
	-- (1,1) circle(1.5pt) node[midway,above,sloped] {$\potentialRegion[\support^0]$}
	-- (3.5,1.5) circle(1.5pt) node[pos=0.25,above,sloped] {$\potentialRegion[\support^1]$}
	-- (5.75,3) node[pos=0.9,above,sloped] {$\potentialRegion[\support^2]$};
\node[anchor=south] at (-0.5,0.5) {$\bar{\vec{d}}_{\support^0}$};
\draw[dashed] (1,-0.75) -- (1,3);
\draw[dashed] (1.25,-0.75) -- (5,3);
\draw[dashed] (3.5,-0.75) -- (3.5,3);
\draw[dashed] (-0.75,1.5) -- (5.75,1.5);


\draw[thick, densely dotted]
	(-0.25,-0.25)
	-- (1,1.25/3-0.25) node[midway,below,sloped] {$\pertPotentialRegion[\support^0]$}
	-- (59/24, 1/6+1/5*35/24) node[pos=0.4,below,sloped] {$\pertPotentialRegion[\support^1]$}
	-- (3.5,0.75) node[midway,below,sloped] {$\pertPotentialRegion[\support^{11}]$}
	-- (5,1.5) node[midway,below,sloped] {$\pertPotentialRegion[\support^{12}]$}
	-- (46/8, 2) node[pos=0.9,below,sloped] {$\pertPotentialRegion[\support^2]$};
\node[anchor=north] at (-0.25,-0.25) {$\bar{\vec{d}}^{\pertVar}_{\support^0}$};
\fill
	(-0.25,-0.25) circle(1.5pt)
	(1,1.25/3-0.25) circle(1.5pt)
	(59/24, 1/6+1/5*35/24) 	circle(1.5pt)
	(3.5,0.75) circle(1.5pt)
	(5,1.5)	circle(1.5pt);
	
\draw[->, shorten >=3pt, shorten <=3pt]
	(-.5,.5) -- (-.25,-.25) node[midway, left] {$\vec{L}^* \vec{G} \vec{C} \pertVec$};
\end{tikzpicture}
\end{center}
\caption{The $\lambda$-potentials of a game (thick line segments) and its $\pertVar$-perturbed variant (dotted line segments). The dashed lines are the hyperplanes induced by the normal vectors $\normalVec$. The support $\support^0$ is non-$b$-degenerate, thus, both in the perturbed and the unperturbed game, there is the unique continuative neighbor $\support^1$. This support is $b$-degenerate in the unperturbed game: It has five continuative neighbors---for four of them $\potentialRegion$ is a single point. In the perturbed game, all supports are non-$b$-degenerate.}
\label{fig:perturbedcurve}
\end{figure}

Let $\pertVar > 0$. Then we define the \emph{perturbation vector}
\[
\pertVec :=
\big( \pertVar, \pertVar^2, \dotsc, \pertVar^{km - 1}, \pertVar^{km} \big)^{\top} \in \R^{mk}
\]
and the vector of perturbed offsets $\vec{b}^{\pertVar} := \vec{b} + \pertVec$. We call the game with this offset vector the \emph{$\pertVar$-perturbed game}. Analogous to Theorem~\ref{thm:lambdapotential} and Lemma~\ref{lem:lambdapotentialsLineSegment} we now obtain the following result for the $\pertVar$-perturbed game.

\begin{lemma}
For $\pertVar > 0$ and some support $\support$, let $\pertPotentialRegion$ be the set of all $\lambda$-potentials in the $\pertVar$-perturbed game. Then there are values $\lambda^{\min,\pertVar}_{\support}$ and $\lambda^{\max, \pertVar}_{\support}$ such that
\[
\pertPotentialRegion = \big\{ \lambda \Delta  \vec{\pi} + \bar{\vec{d}}^{\pertVar} \;\vert\; \lambda^{\min,\pertVar}_{\support} \leq \lambda \leq \lambda^{\max, \pertVar}_{\support} \big\}
,
\]
where $\bar{\vec{d}}^{\pertVar} := \vec{L}^* \vec{G} \vec{C} \vec{b}^{\pertVar}$.
\end{lemma}
\begin{proof}
First, we observe that, similar to Theorem~\ref{thm:lambdapotential},  $\vec{\pi}$ is a $\lambda$-potential for support $\support$ if and only if
\[
\vec{L} \vec{\pi} - \vec{d}^{\pertVar} = \lambda \Delta \vec{y}
\text{ and }
\vec{W} ( \vec{G}^{\top} \!\vec{\pi} - \vec{b}^{\pertVar} ) \geq \vec{0}
,
\]
where $\vec{d}^{\pertVar} = \vec{G} \vec{C} \vec{b}^{\pertVar}$. As in the proof of Lemma~\ref{lem:lambdapotentialsLineSegment}, we observe that 
\[
\pertPotentialRegion = \big\{ \vec{\pi} \in \potentialSpace \;\vert\; \exists \lambda \in [0,1] : \vec{W} ( \vec{G}^{\top}( \lambda \Delta \vec{\pi} + \bar{\vec{d}}^{\pertVar} ) - \vec{b}^{\pertVar} ) \geq \vec{0} \big\}
.
\]
Solving the constraints for $\lambda$, we get that the numbers
\begin{align*}
\lambda^{\min,\pertVar}_{\support} &:=
	\max \bigg\{
		\frac{\vec{w}^{\top}_{e,i} ( \vec{b}^{\pertVar} - \vec{G}^{\top} \bar{\vec{d}}^{\pertVar} )}%
			 {\vec{w}^{\top}_{e,i} \vec{G}^{\top} \! \Delta \vec{\pi}}
		\;\bigg\vert\;
		\vec{w}^{\top}_{e,i} \vec{G}^{\top} \! \Delta \vec{\pi} < 0
		\bigg\} \cup \{ 0 \} \\
\lambda^{\max,\pertVar}_{\support} &:=
	\min \bigg\{
		\frac{\vec{w}^{\top}_{e,i} ( \vec{b}^{\pertVar} - \vec{G}^{\top} \bar{\vec{d}}^{\pertVar} )}%
			 {\vec{w}^{\top}_{e,i} \vec{G}^{\top} \! \Delta \vec{\pi}}
		\;\bigg\vert\;
		\vec{w}^{\top}_{e,i} \vec{G}^{\top} \! \Delta \vec{\pi} > 0
		\bigg\} \cup \{ 1 \}
\end{align*}
yield the claimed representation of $\pertPotentialRegion$.
\end{proof}

With the next theorem, we prove two results for a sufficiently small perturbation $\pertVar^* > 0$. First, we show that every support that is feasible in the perturbed game is also feasible in the original justifying that it is sufficient to consider just these supports. Second, we show that the perturbed game is non-$b$-degenerate.
Figure~\ref{fig:perturbedcurve} illustrates the $\lambda$-potentials for $b$-degenerate and non-$b$-degenerate supports for the unperturbed and the $\pertVar$-perturbed game.

For any support $\support$, we say two edges $e, e'$ are \emph{serial-dependent for player $i$ in support $\support$} if $e, e'$ are serial-dependent\shortversiononly{ (see Section~\ref{app:spSupports} for a definition of serial dependence)} in the subgraph of $G$ containing only the active edges of player $i$.

\begin{theorem} \label{thm:bPerturbation}
Let $\support$ be a non-$a$-degenerate, total support. Then there is $\pertVar^*_{\support} > 0$ such that for all $0 < \pertVar < \pertVar^*_{\support}$ the following holds.
\begin{enumerate}[(i)]
\item \label{it:thm:bPerturbation:infeasible}
	If $\support$ is infeasible, then $\pertPotentialRegion = \emptyset$ as well.
\item \label{it:thm:bPerturbation:nondegenerate}
	If $\support$ is feasible and non-$b$-degenerate, then $\support$ is also non-$b$-degenerate in the $\pertVar$-perturbed game. Further, the support $\support$ has the same unique continuative neighbors in the unperturbed and the $\pertVar$-perturbed game.
\item \label{it:thm:bPerturbation:degenerate}
	If $\support$ is feasible and $b$-degenerate, then $\pertPotentialRegion = \emptyset$ or $\support$ is non-$b$-degenerate in the $\pertVar$-perturbed game.
\end{enumerate}
\end{theorem}
\begin{proof}
We begin the proof by introducing some additional notation. Given a fixed support $\support$, for every edge~$e$ and player~$i$, let
\begin{align*}
\alpha_{e,i} (\lambda) &:= (\vec{w}^{\top}_{e,i} \vec{G}^{\top} \! \Delta \vec{\pi} ) \, \lambda + \vec{w}^{\top}_{e,i} (\vec{G}^{\top} \bar{\vec{d}} - \vec{b}) \\
\beta_{e,i} (\pertVar) &:= \vec{w}^{\top}_{e,i} (\vec{G}^{\top} \bar{\vec{d}} - \vec{b}) \, \pertVec \\
\bar{\lambda}_{e,i} &:= \frac{\vec{w}^{\top}_{e,i} ( \vec{b} - \vec{G}^{\top} \bar{\vec{d}} )}{\vec{w}^{\top}_{e,i} \vec{G}^{\top} \! \Delta \vec{\pi}} \\
\bar{\mu}_{e,i} (\pertVar) &:= \frac{\vec{w}^{\top}_{e,i} ( \vec{I}_{km} - \vec{G}^{\top} \vec{L}^* \vec{G} \vec{C}) }{\vec{w}^{\top}_{e,i} \vec{G}^{\top} \! \Delta \vec{\pi}} \, \pertVec
.
\end{align*}
These definitions enable us to rewrite the (perturbed) potential regions as well as the definitions of $\lambda^{\max}$ and $\lambda^{\min}$.
\begin{claim} \label{clm:thm:bPerturbation:definitions}
The definitions above yield
\begin{compactenum}[(i)]
\item
	$\potentialRegion = \big\{ \vec{\pi} \in \potentialSpace \, \vert \, \exists \lambda \in [0,1] : \alpha_{e,i}(\lambda) \geq 0 \text{ for all } e \in E, i \in \range{k} \big\} $,
\item
	$\pertPotentialRegion = \big\{ \vec{\pi} \in \potentialSpace \, \vert \, \exists \lambda \in [0,1] : \alpha_{e,i}(\lambda) + \beta_{e,i}(\pertVar) \geq 0 \text{ for all } e \in E, i \in \range{k} \big\}$,
\item
	$\lambda^{\min} = \max \big\{ \bar{\lambda}_{e,i} \,\big\vert\, \vec{w}^{\top}_{e,i} \vec{G}^{\top} \! \Delta \vec{\pi} < 0 \big\} \cap \{0\} $ and $\lambda^{\max} = \min \big\{ \bar{\lambda}_{e,i} \,\big\vert\, \vec{w}^{\top}_{e,i} \vec{G}^{\top} \! \Delta \vec{\pi} > 0 \big\} \cap \{1\}$,
\item
	$\lambda^{\min, \pertVar} = \max \big\{ \bar{\lambda}_{e,i} + \bar{\mu}_{e,i} (\pertVar) \,\big\vert\, \vec{w}^{\top}_{e,i} \vec{G}^{\top} \! \Delta \vec{\pi} < 0 \big\} \cup \{0\} $ and \\ $\lambda^{\max, \pertVar} = \min \big\{ \bar{\lambda}_{e,i} + \bar{\mu}_{e,i} (\pertVar) \,\big\vert\, \vec{w}^{\top}_{e,i} \vec{G}^{\top} \! \Delta \vec{\pi} > 0 \big\} \cup \{1\}$.
\end{compactenum}
\end{claim}
\begin{proof}[Proof of Claim~\ref{clm:thm:bPerturbation:definitions}]
The claim follows directly from the definitions of the respective values and the linear perturbation $\vec{b}^{\pertVar} = \vec{b} + \pertVec$.
\end{proof}

Suppose $\support$ is infeasible, i.e., $\potentialRegion = \emptyset$. This implies, that for all $\lambda \in [0,1]$ there is a pair $e,i$ such that $\alpha_{e,i} (\lambda) < 0$. Thus the function $\alpha(\lambda) := \min_{e,i} \alpha_{e,i} (\lambda) < 0$ is strictly negative as well for all $\lambda \in [0,1]$. Since $\alpha (\lambda)$ is continuous in $\lambda$, $\alpha := \max_{\lambda \in [0,1]} \alpha(\lambda) < 0$ as well. Since $\beta_{e,i} (\pertVar) \xrightarrow{\smash{\pertVar \to 0}} 0$, there is a value $\pertVar^*_1 > 0$ such that $\max_{e,i} \beta_{e,i} (\pertVar) < - \alpha $ for all $0< \pertVar < \pertVar^*_1$. Hence, for every $\lambda \in [0,1]$, we get that $\min_{e,i} \alpha_{e,i} (\lambda) + \beta_{e,i} (\pertVar) < \min_{e,i} \alpha_{e,i} (\lambda) - \alpha < 0$ and, thus, $\pertPotentialRegion = \emptyset$ for all $0 < \pertVar < \pertVar^*_1$ proving \emph{(i)} if we choose some $\pertVar^*_{\support} \leq \pertVar^*_1$.

For claim \emph{(ii)}, assume that $\support$ is feasible and non-$b$-degenerate. Then, there is a unique maximizer (or minimizer, resp.) in the computation of $\lambda^{\min}$ (or $\lambda^{\max}$, resp.). For ease of exposition, we assume that the maximum (or minimum) is attained at some value $\bar{\lambda}_{e,i}$ and not at $0$ or $1$---the proof is the same for these cases. Let $e^*_1, i^*_1$ and $e^*_2, i^*_2$ be the respective maximizer or minimizer, i.e., $\lambda^{\min} = \bar{\lambda}_{e^*_1,i^*_1}$ and $\lambda^{\max} = \bar{\lambda}_{e^*_2, i^*_2}$. Let 
\[
\hat{\lambda} := \min \bigg\{ 
\min_{\substack{(e,i) \neq (e^*_1, i^*_1) : \\ \vec{w}^{\top}_{e,i} \vec{G}^{\top} \! \Delta \vec{\pi} <0}} \big( \lambda^{\min} - \bar{\lambda}_{e,i} \big)
,
\min_{\substack{(e,i) \neq (e^*_2, i^*_2): \\ \vec{w}^{\top}_{e,i} \vec{G}^{\top} \! \Delta \vec{\pi} > 0}} \big( \bar{\lambda}_{e,i} - \lambda^{\max}  \big)
\bigg\}
.
\]
For all $(e,i)$ with $\vec{w}^{\top}_{e,i} \vec{G}^{\top} \! \Delta \vec{\pi} \neq 0$, we have that $\bar{\mu}_{e,i} (\pertVar) \xrightarrow{\smash{\pertVar \to 0}} 0$. Thus, there is $\pertVar^*_2 > 0$ such that $|\bar{\mu}_{e,i} (\pertVar)| < \tfrac{\hat{\lambda}}{2}$ for all $e,i$ with $\vec{w}^{\top}_{e,i} \vec{G}^{\top} \! \Delta \vec{\pi} \neq 0$ and $0 < \pertVar < \pertVar^*_2$. Therefore, we obtain for every $(e,i) \neq (e^*_1, i^*_1)$
\[
\bar{\lambda}_{e^*_1, i^*_1} + \bar{\mu}_{e^*_1, i^*_1} (\pertVar) 
\geq \bar{\lambda}_{e,i} + \hat{\lambda} + \bar{\mu}_{e^*_1, i^*_1}(\pertVar)
> \bar{\lambda}_{e,i} + \frac{\hat{\lambda}}{2} > \bar{\lambda}_{e,i} + \bar{\mu}_{e,i}(\pertVar)
\]
and, hence, that $(e^*_1, i^*_1)$ is also the unique maximizer in the computation of $\lambda^{\min, \pertVar}$. Similarly, we obtain $\bar{\lambda}_{e^*_2, i^*_2} + \bar{\mu}_{e^*_2, i^*_2} (\pertVar) <  \bar{\lambda}_{e, i} + \bar{\mu}_{e, i}$ proving that $(e^*_2, i^*_2)$ is the unique minimizer in the computation of $\lambda^{\max, \pertVar}$. Thus, we have established \emph{(ii)} for all $0 < \pertVar < \pertVar^*_2$.

For the proof of statement \emph{(iii)} assume that $\pertVar^*_3 > 0$ is chosen small enough such that, for all $0 < \pertVar < \pertVar^*_3$, $\bar{\lambda}_{e,i} + \bar{\mu}_{e,i} (\pertVar) = \bar{\lambda}_{e',i'} + \bar{\mu}_{e',i'} (\pertVar)$ implies that $\bar{\lambda}_{e,i} = \bar{\lambda}_{e',i'}$ and $\bar{\mu}_{e,i} (\pertVar) =  \bar{\mu}_{e',i'} (\pertVar)$.
Now we assume that $\support$ is feasible and $b$-degenerate. Further, assume that $\pertPotentialRegion \neq \emptyset$. We proceed to show that in this case, if the maximizer in the computation of $\lambda^{\max,\pertVar}$ or the minimizer in the computation of $\lambda^{\min, \pertVar}$ is non unique, then it is tied only among pairs $e,i$ that are serial-dependent for some player~$i$.
We begin the proof by establishing a necessary condition for the case of non-unique maximizer or minimizer.
\begin{claim} \label{clm:thm:bPerturbation:equalLambda}
Given two edge-player pairs $(e,i), (e', i')$, if for all $\pertVar^*_3 > 0$ there is a $0 < \pertVar < \pertVar^*_3$ such that $\bar{\lambda}_{e,i} + \bar{\mu}_{e,i} (\pertVar) =   \bar{\lambda}_{e',i'} + \bar{\mu}_{e',i'} (\pertVar)$, then
\begin{equation*}
\big( \eta_{e,i} \, \vec{u}^{\top}_{e,i} - \eta_{e',i'} \, \vec{u}^{\top}_{e',i'} \big) \vec{X} = \vec{0}
\end{equation*}
where $\vec{X} := \vec{I}_{mk} -  \tilde{\vec{C}} \vec{G}^{\top} \vec{L}^* \vec{G} \vec{\Omega}$, and where $\eta_{e,i}$ and $\eta_{e',i'}$ are some non-zero constants.
\end{claim}
\begin{proof}[Proof of Claim~\ref{clm:thm:bPerturbation:equalLambda}]
Since we assume that all $\pertVar^*_3$ are chosen small enough, the assumption of the claim implies that there are infinity many values $0 < \pertVar < \pertVar^*_3$ such that $\bar{\mu}_{e,i} = \bar{\mu}_{e',i'}$. Let $\eta_{e,i} := \frac{\sigma_e^i}{\vec{w}^{\top}_{e,i} \vec{G}^{\top} \! \Delta \vec{\pi}}$ and $\eta_{e',i'} := \frac{\sigma_{\smash{e'}}^{\smash{i'}}}{\vec{w}^{\top}_{\smash{e',i'}} \vec{G}^{\top} \! \Delta \vec{\pi}}$. With $\vec{w}^{\top}_{e,i} = \vec{u}^{\top}_{e,i} \vec{W} = \vec{u}^{\top}_{e,i} \vec{\Sigma} \tilde{\vec{C}} = \sigma_{e}^i \vec{u}^{\top}_{e,i}  \tilde{\vec{C}}$ we obtain
\begin{align*}
0 = \bar{\mu}_{e,i} - \bar{\mu}_{e',i'}
&= \bigg( \frac{1}{\vec{w}^{\top}_{e,i} \vec{G}^{\top} \! \Delta \vec{\pi}} \vec{w}^{\top}_{e,i}
- \frac{1}{\vec{w}^{\top}_{e',i'} \vec{G}^{\top} \! \Delta \vec{\pi}} \vec{w}^{\top}_{e',i'} \bigg) \big( \vec{I}_{mk} -  \vec{G}^{\top} \vec{L}^* \vec{G} \vec{C} \big) \pertVec \\
&= \big( \eta_{e,i} \, \vec{u}^{\top}_{e,i} - \eta_{e',i'} \, \vec{u}^{\top}_{e',i'} \big) \tilde{\vec{C}} \big( \vec{I}_{mk} -  \vec{G}^{\top} \vec{L}^* \vec{G} \vec{\Omega} \tilde{\vec{C}} \big) \pertVec \\
&= \big( \eta_{e,i} \, \vec{u}^{\top}_{e,i} - \eta_{e',i'} \, \vec{u}^{\top}_{e',i'} \big) \vec{X} \tilde{\vec{C}} \, \pertVec
.
\end{align*}
Since the last term is a polynomial in $\pertVar$, the assumption implies that this polynomial must have infinitely many roots, i.e., it must be the zero polynomial, implying $\big( \eta_{e,i} \, \vec{u}^{\top}_{e,i} - \eta_{e',i'} \, \vec{u}^{\top}_{e',i'} \big) \vec{X} \tilde{\vec{C}} = \vec{0}$. Since $\tilde{\vec{C}}$ is non-singular (see the proof of Lemma~\ref{lem:matricesC}), the claim follows.
\end{proof}

Claim~\ref{clm:thm:bPerturbation:equalLambda} implies that, if $\pertPotentialRegion$ is degenerate, there is a special vector with only two non-zero entries in the nullspace of the matrix $\vec{X}^{\top}$. The next claim characterizes nullspace of this matrix.

\begin{claim} \label{clm:thm:bPerturbation:nullspaceX}
$\vec{v}^{\top} \vec{X} = \vec{0}$ if and only if there is a vector $\vec{z} \in \R^{nk}$ such that $\vec{v} = (\vec{G} \vec{\Omega})^{\top} \vec{z}$.
\end{claim}
\begin{proof}[Proof of Claim~\ref{clm:thm:bPerturbation:nullspaceX}]
Let $\vec{Y} := \tilde{\vec{C}} \vec{G}^{\top} \vec{L}^* \vec{G} \vec{\Omega} = \tilde{\vec{C}} \vec{G}^{\top} ( \vec{G} \vec{\Omega} \tilde{\vec{C}} \vec{G}^{\top} )^* \vec{G} \vec{\Omega}$. Then, $\vec{X} = \vec{I}_{mk} - \vec{Y}$, i.e., it suffices to show that $\vec{v}^{\top} \vec{Y} = \vec{v}$ if and only if $\vec{v} = (\vec{G} \vec{\Omega})^{\top} \vec{z}$.

We note that the space spanned by the colums of $\vec{G} \vec{\Omega}$ is a subspace of the space spanned by the columns of $\vec{L} = \vec{G} \vec{\Omega} \tilde{\vec{C}} \vec{G}^{\top}$. Thus, we can use \cite[Lemma~9.3.5]{harville1997matrix}, which yields
\[
\vec{G} \vec{\Omega} \vec{Y}
=
\vec{G} \vec{\Omega} \tilde{\vec{C}} \vec{G}^{\top} ( \vec{G} \vec{\Omega} \tilde{\vec{C}} \vec{G}^{\top} )^* \vec{G} \vec{\Omega} = \vec{G} \vec{\Omega}
\]
proving the if-direction.

For the only-if direction, observe that $\rank(\vec{Y}) \leq \rank(\vec{G} \vec{\Omega})$ by definition. On the other hand, we have that $\rank(\vec{G} \vec{\Omega}) = \rank(\vec{G} \vec{\Omega} \vec{Y}) \leq \rank(\vec{Y})$. Thus, $\rank(\vec{Y}) = \rank(\vec{G} \vec{\Omega})$. We conclude that for every $\vec{v}$ that is not a combination of the rows of $\vec{G} \vec{\Omega}$, $\vec{v}^{\top} \vec{Y} = \vec{0}$ (otherwise the rank of $\vec{Y}$ would be strictly greater than $\rank(\vec{G} \vec{\Omega})$). Thus, if there is no $\vec{z}$ with $\vec{v} = (\vec{G} \vec{\Omega})^{\top} \vec{z}$, then $\vec{v}^{\top} \vec{Y} = \vec{0}$.
\end{proof}

Claim~\ref{clm:thm:bPerturbation:nullspaceX} states that a vector $\vec{v}$ is in the left nullspace of the matrix $\vec{X}$ if and only if there is a potential vector $\vec{z}$ that assigns a potential value $z_v^i$ for every player to every vertex such that
\begin{equation} \label{eq:thm:bPerturbation:nullspaceX}
v_{e,i} = 
\begin{cases}
0
	&\text{ if }i \notin S_e, \\
z_w^i - z_v^i
	&\text{ if } i \in S_e
\end{cases}
\end{equation}
for every edge $e = (v,w)$ and every player $i \in \range{k}$.
Thus, in particular, the components of every vector of the left nullspace of $\vec{X}_{\support}$ that corresponds to a player-edge-pair with an inactive edge are zero.
We can use this insight to prove the following claim.

\begin{claim} \label{clm:thm:bPerturbation:serialdependent}
Given two edge-player pairs $(e,i), (e', i')$, if for all $\pertVar^*_3 > 0$ there is a $0 < \pertVar < \pertVar^*_3$ such that $\bar{\lambda}_{e,i} + \bar{\mu}_{e,i} (\pertVar) =   \bar{\lambda}_{e',i'} + \bar{\mu}_{e',i'} (\pertVar)$, one of the following is true.
\begin{compactenum}[(i)]
\item
	$\bar{\lambda}_{e,i} + \bar{\mu}_{e,i} (\pertVar) =   \bar{\lambda}_{e',i'} + \bar{\mu}_{e',i'} (\pertVar) = 0$.
\item
	We have that $i = i'$ and $e$ and $e'$ are serial-dependent for player~$i$.
\end{compactenum}
\end{claim}
\begin{proof}[Proof of Claim~\ref{clm:thm:bPerturbation:serialdependent}]
The claims~\ref{clm:thm:bPerturbation:equalLambda} and~\ref{clm:thm:bPerturbation:nullspaceX} imply that there is a vertex potential $\vec{z}$ such that $(\eta_{e,i} \vec{u}_{e,i}^{\top} + \eta_{e',i'} \vec{u}_{e',i'}^{\top}) = (\vec{G} \vec{\Omega})^{\top} \vec{z}$. Since $(\eta_{e,i} \vec{u}_{e,i}^{\top} + \eta_{e',i'} \vec{u}_{e',i'}^{\top})$ has exactly two non-zero components corresponding to $(e,i)$ and $(e', i')$, we conclude with~\eqref{eq:thm:bPerturbation:nullspaceX} that $i \in S_e$ and $i' \in S_{e'}$, i.e., both edges are active for their respective players. We now consider two cases:

Case 1: $i \neq i'$. In this case, the edge~$e=(v,w)$ is the only active edge with a non-zero potential difference $z^i_w - z^i_v \neq 0$. This implies that $e$ can not be contained in any (undirected) cycle inside the active network of player $i$. We observe that $\Delta \vec{x} := \vec{C} \vec{G}^{\top} \vec{L}^* \Delta \vec{y}$ is a non-negative flow satisfying the demand $r = 1$ for every player~$i$ since $\vec{G}\Delta \vec{x} = \Delta \vec{y}$. In particular, the flow $\Delta \vec{x}^i$ is a $s_i$-$t_i$-flow. Further, we observe that $\Delta x_e^i = \vec{u}^{\top}_{e,i} \omega^i_e  \tilde{\vec{C}} \vec{G}^{\top} \vec{L}^* \Delta \vec{y} = \vec{w}^{\top}_{e,i} \vec{G}^{\top} \! \Delta \vec{\pi} \neq 0$ (because only in this case the pair (e,i) contributes to the computation of $\lambda^{\max, \pertVar}$ and $\lambda^{\min, \pertVar}$). Hence, there is non-zero flow $\Delta x_e^i$ on the active edge $e$ in the active network of player $i$. Since there is no cycle in the active network of player~$i$, all flow (in particular also all flow of $\vec{x}^i$) must use edge $e$. This implies that the inequality corresponding to edge~$e$ and player~$i$ in $\pertPotentialRegion$ is $\lambda = x^i_e \geq 0$. Hence, \emph{(i)} follows.

Case 2: $i = i'$. In this case, there are exactly two edges in the active network of player~$i$ that have non-zero potential difference $z^i_w - z^i_v$. This is equivalent by definition to the fact that these edges are serial-dependent for player~$i$.
\end{proof}

Claim~\ref{clm:thm:bPerturbation:serialdependent} implies that whenever there is not a unique minimizer or maximizer in the computation of $\lambda^{\max, \pertVar}$ or $\lambda^{\min, \pertVar}$, then either the respective value is $0$ or the corresponding edges are serial-dependent for one single player $(e,i)$. The first case is no problem, since we do not need to define a predecessor for a region with $\lambda^{\min, \pertVar} = 0$ (since this is the start region). In the other case, we there is still a unique continuative neighbor in the $\myState$, since we need to make sure, that the neighboring support is also a shortest-path support (to be in $\myState$). In particular the $(e,i)$-neighbor with $(e,i)$ being the pair where $e$ is the farthest away from the source $s_i$ (see also the definition of serial-dependent edges) is the only continuative neighbor in $\myState$. Thus, the region $\pertPotentialRegion$ is not $b$-degenerate. Choosing $\pertVar^*_{\support} := \min \{ \pertVar^*_1, \pertVar^*_2, \pertVar^*_3 \}$ concludes the proof.
\end{proof}

Theorem~\ref{thm:bPerturbation} implies directly that any perturbed game is non-degenerate for $\pertVar$ chosen small enough.
\begin{corollary} \label{cor:bPerturbation}
Let $\pertVar^* := \min_{\support \in \supportSpace} \pertVar^*_{\support}$, where $\pertVar^*_{\support} > 0 $ are the values from Theorem~\ref{thm:bPerturbation}. Then $\pertVar^* > 0$ and for all $0 < \pertVar < \pertVar^*$, the perturbed game with offsets $\vec{b} + \pertVec$ is non-$b$-degenerate.
\end{corollary}

Theorem~\ref{thm:bPerturbation} gives also insights about how likely $b$-degenerate games are.
Given a fixed network with graph $G$ and $k$ players and fixed slopes $\vec{a}$ of the cost functions, let 
\[
\degenerateOffsets := \big\{ \vec{b} \in \R^{mk}_{\geq 0} \;\big\vert\;
\text{the game with offsets } \vec{b} \text{ is } b \text{-degenerate} \big\}
\]
be the set of all offsets leading to degenerate games. Then the following corollary shows that the set $\degenerateOffsets$ is in fact a null set and an arbitrary game is non-$b$-degenerate almost surely.

\begin{corollary} \label{cor:nonbdegenerate:as}
Let $\mathcal{D}$ be a continuous distribution on $\R^{mk}_{\geq 0}$. Then,
$
\mathbb{P}_{\mathcal{D}} \big[ \vec{b} \notin \degenerateOffsets \big] = 1
.
$
\end{corollary}
\begin{proof}
With the definitions from the proof of Theorem~\ref{thm:bPerturbation}, we get that, whenever some feasible support $\support$ is degenerate, $\bar{\lambda}_{e,i} = \bar{\lambda}_{e',i'}$. Similar to Claim~\ref{clm:thm:bPerturbation:equalLambda}, we can reformulate this as the condition
\begin{equation} \label{eq:cor:nonbdegenerate:as:subclaim}
\big( \eta_{e,i} \, \vec{u}^{\top}_{e,i} - \eta_{e',i'} \, \vec{u}^{\top}_{e',i'} \big) \vec{X}_{\support} \tilde{\vec{C}}_{\support} \, \vec{b} = 0
.
\end{equation}
We may assume that $\big( \eta_{e,i} \, \vec{u}^{\top}_{e,i} - \eta_{e',i'} \, \vec{u}^{\top}_{e',i'} \big) \vec{X}_{\support} \tilde{\vec{C}}_{\support}$ is not the zero vector. Otherwise, Claim~\ref{clm:thm:bPerturbation:serialdependent} implies as argued in the proof of Theorem~\ref{thm:bPerturbation} that the support is non-$b$-degenerate. Hence, any offset vector $\vec{b}$ belonging to a $b$-degenerate game satisfies at least one non-trivial linear condition of the form~\eqref{eq:cor:nonbdegenerate:as:subclaim} for some support $\support$ and some pair of edge~$e$ and player~$i$. Thus,
\[
\degenerateOffsets \subseteq
\bigcup_{\support \in \myState} \bigcup_{(e,i) \in E \times \range{k}} \bigcup_{(e',i') \in E \times \range{k}}  \big\{ \vec{b} \in \R^{mk}_{\geq 0} \, \big\vert \, \big( \eta_{e,i} \, \vec{u}^{\top}_{e,i} +- \eta_{e',i'} \, \vec{u}^{\top}_{e',i'} \big) \vec{X}_{\support} \tilde{\vec{C}}_{\support} \, \vec{b}  = 0 \big\}
.
\]
Since every set $\big\{ \vec{b} \in \R^{nk} \, \big\vert \, \big( \eta_{e,i} \, \vec{u}^{\top}_{e,i} + \eta_{e',i'} \, \vec{u}^{\top}_{e',i'} \big) \vec{X}_{\support} \tilde{\vec{C}}_{\support} \, \vec{b}  = 0 \big\}$ is a lower dimensional subspace and the union above is finite, we get
\begin{align*}
\mathbb{P} \big[ \vec{b} \in \mathfrak{D} \big] 
&\leq \mathbb{P} \Big[ \vec{b} \in \bigcup_{\support \in \myState} \bigcup_{(e,i) \in E \times \range{k}} \bigcup_{(e',i') \in E \times \range{k}}  \big\{ \vec{b}' \in \R^{mk}_{\geq 0} \, \big\vert \, \big( \eta_{e,i} \, \vec{u}^{\top}_{e,i} - \eta_{e',i'} \, \vec{u}^{\top}_{e',i'} \big) \vec{X}_{\support} \tilde{\vec{C}}_{\support} \, \vec{b}'  = 0 \big\} \Big] \\
&\leq \sum_{\support \in \myState} \sum_{(e,i) \in E \times \range{k}} \sum_{(e',i') \in E \times \range{k}}    \underbrace{   \mathbb{P} \big[ \vec{b} \in \big\{ \vec{b}' \in \R^{mk}_{\geq 0} \, \big\vert \, \big( \eta_{e,i} \, \vec{u}^{\top}_{e,i} - \eta_{e',i'} \, \vec{u}^{\top}_{e',i'} \big) \vec{X}_{\support} \tilde{\vec{C}}_{\support} \, \vec{b}'  = 0 \big\} \big]   }_{   = 0    }
 = 0
\end{align*}
which proves the claim.
\end{proof}


\subsubsection{Implicit Perturbation}

Finally, we discuss how to use the fact that perturbed games are non-$b$-degenerate games in order to obtain a unique pivoting rule for all games. First, we define a new set of supports
\begin{align*}
\pertSupportSpace :=\{ \mathcal{S} : \;&\mathcal{S} \text{ is a feasible, non-}a\text{-degenerate, shortest-path-support for offsets } \vec{b} + \pertVec  \} 
\end{align*}
that we use as states of our $\PPAD$-graph. Theorem~\ref{thm:bPerturbation}\emph{(\ref{it:thm:bPerturbation:infeasible})} and  \emph{(\ref{it:thm:bPerturbation:nondegenerate})} imply that $\myState = \pertSupportSpace$, whenever the game is non-$b$-degenerate. We then redefine the predecessor and successor functions $\predF$ and $\succF$ such that they rely on the continuative neighbors $\pertLowerContNeighbors, \pertUpperContNeighbors$ of the perturbed game rather than on the actual continuative neighbors $\lowerContNeighbors, \upperContNeighbors$. Again, for non-$b$-degnerate games, this does not change these functions. But, by Corollary~\ref{cor:bPerturbation}, we get that now these functions are well-defined.

With the functions $\predF$ and $\succF$ depending on the continuative neighbors of the perturbed game we need to clarify how we can actually compute these neighboring supports. Computing these supports with an explicit value for $\pertVar$ via the values $\lambda^{\min, \pertVar}$ and $\lambda^{\max, \pertVar}$ is practically feasible---to begin with, it is not obvious how to choose $\pertVar^*$ without checking exponentially many supports. 
Fortunately, an computations with explicit values for $\pertVar$ are not necessary, as we show in the following lemma.

\begin{lemma}
For all $0 < \pertVar < \pertVar^*$, the unique continuative neighbors can be determined in polynomial times without using an explicit value of $\pertVar$.
\end{lemma}
\begin{proof}
Theorem~\ref{thm:bPerturbation} proves that these continuative neighbors are unique, since all supports are non-$b$-degenerate for all $0 < \pertVar < \pertVar^*$. Further, we obtain from the proof, that the pairs $(e,i)$ inducing the maximum (or minimum, resp.) in the computation of $\lambda^{\min, \pertVar}$ (or $\lambda^{\max, \pertVar}$, resp.) can be determined by finding the edge~$e$ and the player~$i$ that maximizes (or minimizes, reps.) $\bar{\mu}_{e,i} (\pertVar)$ among all pairs of edges and players that already maximize (or minimize, resp.) $\bar{\lambda}_{e,i}$. In order to see that this is possible without an explicit value for $\pertVar$, we observe (see Claim~\ref{clm:thm:bPerturbation:equalLambda}) that for two pairs $(e,i)$ and $(e',i')$ the difference of the $\bar{\mu}$ values is
\[
p_{e,i,e',i'} (\pertVar) := \bar{\mu}_{e,i} (\pertVar) - \bar{\mu}_{e',i'} (\pertVar) 
= \big( \eta_{e,i} \, \vec{u}^{\top}_{e,i}  - \eta_{e',i'} \, \vec{u}^{\top}_{e',i'} \big) \vec{X} \tilde{\vec{C}} \, \pertVec
,
\]
i.e., a polynomial in $\pertVar$ with coefficient vector $\big( \eta_{e,i} \, \vec{u}^{\top}_{e,i}  - \eta_{e',i'} \, \vec{u}^{\top}_{e',i'} \big) \vec{X} \tilde{\vec{C}}$. Since $p_{e,i,e',i'} (0) = 0$, and $p_{e,i,e',i'} (\pertVar)$ is not the zero polynomial, we infer that $\bar{\mu}_{e,i} (\pertVar) - \bar{\mu}_{e',i'} (\pertVar) > 0$ if and only if the coefficient vector $\big( \eta_{e,i} \, \vec{u}^{\top}_{e,i}  - \eta_{e',i'} \, \vec{u}^{\top}_{e',i'} \big) \vec{X} \tilde{\vec{C}}$ is lexicographically greater than the zero vector $\vec{0}$.
Equivalently, for comparing $\bar{\mu}_{e,i} (\pertVar)$ and $\bar{\mu}_{e',i'} (\pertVar)$ it suffices to compare the vectors $\eta_{e,i} \, \vec{u}^{\top}_{e,i} \vec{X} \tilde{\vec{C}}$ and $\eta_{e',i'} \, \vec{u}^{\top}_{e',i'} \vec{X} \tilde{\vec{C}}$ lexicographically. This can be done without an explicit value for $\pertVar$ and clearly in polynomial time.
\end{proof}

Overall, we conclude that we can compute the (unique) continuative neighbors of the perturbed game efficiently and without an explicit value for $\pertVar$. Thus, we can compute the predecessor and  successor function of the perturbed game. As stated above, in non-$b$-degenerate games this makes no difference, but in $b$-degenerate games it yields a unique pivoting rule even for degenerate supports. Hence, computing equilibria in degenerate games is still in $\PPAD$ and, moreover, our parametric computation approach is still feasible for $b$-degenerate games.

\subsubsection{$a$-$b$-Degeneracy}

In this subsection, we address the case of $a$-$b$-degeneracy. This kind of degeneracy occurs if for some feasible, $a$-degenerate support $\support$, the maximizer (or minimizer, resp.) in the computation of the values $\nullspaceMin$ and $\nullspaceMax$ are not unique. In this case, there is not a unique continuative neighbor. Since $a$-$b$-degeneracy depends on the values $\nullspaceMin$ and $\nullspaceMax$, Theorem~\ref{thm:bPerturbation} does not immediately prove that $a$-$b$-degeneracy does not occur in perturbed games. However, we can use the techniques of the proof of Theorem~\ref{thm:bPerturbation} in order to obtain a similar statement.

\begin{corollary}
There is $\pertVar^* > 0$ such that, for all $0< \pertVar < \pertVar^*$, $a$-$b$-degeneracy does not occur in the $\pertVar$-perturbed game.
\end{corollary}
\begin{proof}
Similarly as we did for for the values $\lambda^{\min}$ and $\lambda^{\max}$, we can compute the values $\nullspaceMin, \nullspaceMax$ as
\begin{align*}
\nullspaceMin &= 
	\min \big\{
	\bar{\xi}_{e,i}
	\;\big\vert\;
	\vec{w}^{\top}_{e,i} \vec{G}^{\top} \! \nullspaceD < 0
	\big\} &\text{and}&&
\nullspaceMax &= 
	\min \big\{
	\bar{\xi}_{e,i}
	\;\big\vert\;
	\vec{w}^{\top}_{e,i} \vec{G}^{\top} \! \nullspaceD > 0
	\big\},
\end{align*}
where
\begin{align*}
\bar{\xi}_{e,i} &= 
\frac
	{\vec{w}^{\top}_{e,i} \big( \vec{b} - \vec{G}^{\top} (\vec{L}^+ (\lambda^*_{\support} \Delta \vec{y} + \vec{d} ) \big) }
	{\vec{w}^{\top}_{e,i} \vec{G}^{\top} \! \nullspaceD}
.
\end{align*}
If we consider the perturbed game with the offset vector $\vec{b}^{\pertVar} = \vec{b} + \pertVec$, we obtain $\nullspaceMin = \min \big\{\bar{\xi}_{e,i} + \bar{\nu}_{e,i} (\pertVar) \;\big\vert\; \vec{w}^{\top}_{e,i} \nullspaceD < 0 \big\}$ and $\nullspaceMax = \min \big\{\bar{\xi}_{e,i} + \bar{\nu}_{e,i} (\pertVar) \;\big\vert\; \vec{w}^{\top}_{e,i} \nullspaceD > 0 \big\}$, where
\[
\bar{\nu}_{e,i} (\pertVar) = \frac{\vec{w}^{\top}_{e,i} ( \vec{I}_{km} - \vec{G}^{\top} \vec{L}^+ \vec{G} \vec{C}) }{\vec{w}^{\top}_{e,i} \vec{G}^{\top} \! \Delta \vec{\pi}} \, \pertVec
.
\]
With the same techniques as in the proof of Theorem~\ref{thm:bPerturbation}, we obtain that the maximizer (or minimizer, resp.) is non-unique only if there is a vector $\big( \eta_{e,i} \, \vec{u}^{\top}_{e,i} - \eta_{e',i'} \, \vec{u}^{\top}_{e',i'} \big)$ in the nullspace of the matrix $\tilde{\vec{X}} := \vec{I}_{mk} - \tilde{\vec{C}} \vec{G}^{\top} \vec{L}^+ \vec{G} \vec{\Omega}$. The rest of the proofs follows then similar to the proof of Theorem~\ref{thm:bPerturbation}, with the only difference being that the matrix $\tilde{\vec{X}}$ depends on the generalized inverse $\vec{L}^+$ rather than $\vec{L}^*$ as does the matrix $\vec{X}$. Since $\vec{L}^+$ has lower rank in an $a$-degenerate region compared to $\vec{L}^*$ in a non-$a$-degenerate region, in Claim~\ref{clm:thm:bPerturbation:nullspaceX} only the \emph{only-if} direction holds true. However, this is enough to establish Claim~\ref{clm:thm:bPerturbation:serialdependent}, hence, the statement of Theorem~\ref{thm:bPerturbation} can be transferred to the $a$-$b$-degenerate case. 
\end{proof}

\shortversiononly{
\section{Shortest-Path-Supports} \label{app:spSupports}

}

\includeappendixcollection{proof}{Missing Proofs}


 \bibliography{bibliography}

\end{document}